\documentclass{colt2012} 

\usepackage{algorithm,algorithmic}
\newtheorem{thm}{Theorem}
\newtheorem{prop}{Proposition}
\newtheorem{cor}[thm]{Corollary}

\usepackage{enumitem}
\usepackage{booktabs}       
\usepackage{fullpage}

\def \R {\mathbb{R}}

\def \y {\mathbf{y}}
\def \x {\mathbf{x}}
\def \a {\mathbf{a}}
\def \L {\mathcal{L}}

\def \v {\mathbf{v}}
\def \c {\mathbf{c}}
\def \X {\mathcal{X}}

\def \S {\mathcal{S}}

\def \z {\mathbf{z}}

\def \xh {\widehat{\x}}

\def \Er {\mathcal{E}}
\def \u {\mathbf{u}}

\def \v {\mathbf{v}}
\def \w {\mathbf{w}}

\def \B {\mathcal{B}}

\def \R {\mathbb{R}}

\def \b {\mathbf{b}}

\def \A {\mathcal{A}}
\def \B {\mathcal{B}}

\def \e {\mathbf{e}}

\def \xt {\widetilde{\x}}

\def \K {\mathcal{K}}

\def \nh {\widehat{\nabla}}
\def \Sb {\overline{\S}}
\def \X {\mathcal{X}}
\def \Y {\mathcal{Y}}

\def \R {\mathbb{R}}
\def \x {\mathbf{x}}

\def \L {\mathcal{L}}
\def \S {\mathcal{S}}
\def \e {\mathbf{e}}

\def \B {\mathcal{B}}

\def \u {\mathbf{u}}
\def \z {\mathbf{z}}
\def \xh {\widehat{\x}}

\def \y {\mathbf{y}}

\def \Sb {\overline{\S}}
\def \w {\mathbf{w}}

\def \v {\mathbf{v}}
\def \sgn {\mbox{sign}}
\def \a {\mathbf{a}}
\def \Er {\mathcal{E}}
\def \K {\mathcal{K}}

\begin{document}

\title[A Simple Homotopy Proximal Mapping Algorithm for Compressive Sensing]{A Simple Homotopy Proximal Mapping Algorithm for Compressive Sensing}
 \coltauthor{\Name{Tianbao Yang}$^\dagger$ \Email{tianbao-yang@uiowa.edu}\\
 \Name{Lijun Zhang}$^\ddagger$\Email{zhanglj@lamda.nju.edu.cn}\\
 \Name{Rong Jin}$^\natural$ \Email{jinrong.jr@alibaba-inc.com}\\
 \Name{Shenghuo Zhu}$^\natural$ \Email{shenghuo@gmail.com}\\
 \Name{Zhi-Hua Zhou}$^\ddagger$ \Email{zhouzh@lamda.nju.edu.cn}\\
   \addr$^\dagger$Department of Computer Science \\
    The University of Iowa, Iowa City, IA 52242 \\
   \addr$^\ddagger$National Key Laboratory for Novel Software Technology\\
       Nanjing University, Nanjing 210023, China\\
          \addr$^\natural$Alibaba Group, Seattle, USA
}


\maketitle

\begin{abstract}
In this paper, we present a novel yet simple homotopy proximal mapping algorithm for compressive sensing. The algorithm adopts a simple proximal mapping of the $\ell_1$ norm  at each iteration  and gradually reduces the regularization parameter for the $\ell_1$ norm.  We prove a global linear convergence of the proposed homotopy proximal mapping (HPM) algorithm for solving compressive sensing under three different settings (i) sparse signal recovery under noiseless measurements, (ii)  sparse signal recovery under noisy measurements, and (iii) nearly-sparse signal recovery under sub-gaussian noisy measurements. In particular, we show that when the measurement matrix satisfies Restricted Isometric Properties (RIP), our theoretical results in settings  (i) and (ii)  almost recover the best condition on the RIP constants for compressive sensing. In addition, in setting (iii), our results for sparse signal recovery are better than the previous results, and furthermore  our analysis explicitly exhibits that more observations lead to not only more accurate recovery but also faster convergence. Compared with previous studies on linear convergence for sparse signal recovery, our algorithm is simple and efficient, and our results are better and provide more insights. Finally our empirical studies provide further support for the proposed homotopy proximal mapping algorithm and verify the theoretical results. 
\end{abstract}


\section{Introduction}
The problem of sparse signal recovery is to reconstruct a sparse signal given a number of linear measurements of the signal.  The problem has been studied extensively  in the literature related to  compressive sensing~\citep{citeulike:2891154,Donoho06compressedsensing} and model selection~\citep{tibshirani96regression,Efron04leastangle,conf/isit/KyrillidisC12}. Numerous algorithms and  results have been developed  for sparse signal recovery under different settings and different conditions. Let $\x_*\in\R^d$ denote a target signal and $\y=U\x_* + \e\in\R^n$ denote $n<d$ measurements of $\x_*$, where $U\in\R^{n\times d}$ is a measurement matrix and $\e$ encodes potential noise in the observations.  In the earliest studies of compressive sensing~\citep{Candes:2005:DLP:2263433.2271950,citeulike:10486729,Chen:2001:ADB:588736.588850,journals/tit/DonohoT08},  the sparse signal recovery is cast into a linear programming problem: 
\begin{equation}\label{eqn:l1}
\begin{aligned}
\min_{\x\in\R^d}&\quad  \|\x\|_1\\
s.t.&\quad\|U\x-\y\|_2\leq \epsilon.
\end{aligned}
\end{equation}
It was shown that when the measurement matrix $U$ satisfies RIP with small RIP constants (c.f. the definition in \textbf{Definition 1}), the solution to~(\ref{eqn:l1}) denoted by $\bar\x$ can recover the sparse signal $\x_*$ up to  the noise level $\|\e\|_2$.  In their seminal work~\citep{Candes:2005:DLP:2263433.2271950}, Cand\`{e}s and Tao proved that when $\e=0$, i.e, there is no noise in the observations, $\x_*$ is the unique solution to~(\ref{eqn:l1}) provided that RIP constants of $U$ satisfy 
$\delta_s + \delta_{2s} + \delta_{3s}< 1$.
 The recovery result was  later generalized to a more general setting of nearly-sparse signal recovery with noisy observations, under the condition $\delta_{2s}\leq \sqrt{2}-1$ and $\epsilon\geq \|\e\|_2$~\citep{citeulike:10486729}.
Similar recovery results have been obtained  for the Dantzig selector~\citep{CT07}: 
\begin{equation}\label{eqn:dan}
\begin{aligned}
\min_{\x\in\R^d}&\quad  \|\x\|_1\\
s.t.&\quad\|U^{\top}(U\x-\y)\|_\infty\leq \lambda. 
\end{aligned}
\end{equation} by setting $\lambda\geq \|U^{\top}\e\|_\infty$. 
The sparse signal recovery is also closely related to the basis pursuit denoising problem (BPDN)~\citep{Chen:1998:ADB:305219.305222}, which aims to solve the following unconstrained $\ell_1$ regularized least-squares minimization problem: 
\begin{align}\label{eqn:lasso}
\min_{\x\in\R^d}\quad \underbrace{\frac{1}{2}\|U\x - \y\|_2^2}\limits_{f(\x)} + \lambda \|\x\|_1,
\end{align}
where $\lambda$ is a regularization parameter. Various properties of the optimal solution $\bar\x$ to~(\ref{eqn:lasso}) have been investigated~\citep{citeulike:2823189,journals/tit/Tropp06,Zhao:2006:MSC:1248547.1248637,zhang2008,citeulike:5426408,Bickel09simultaneousanalysis,GeeBue09,Wainwright:2009:STH:1669487.1669506}. In particular, it is known that under RIP for $U$, as long as $\lambda> c\|U^{\top}\e\|_\infty$, where $c$ is a universal constant, the optimal solution $\bar\x$ to~(\ref{eqn:lasso}) can recover a $s$-sparse signal $\x_*$ up to the noise level. 

In this paper, we study the problem of sparse signal recovery  by directly analyzing the convergence of a new optimization algorithm, namely the homotopy proximal mapping algorithm. {The algorithm adopts a proximal mapping for the $\ell_1$ norm regularization at each iteration:
\begin{align*}
    \x_{t+1} = \mathop{\arg\min}\limits_{\x \in \R^d} \frac{1}{2}\left\|\x - \left(\x_t- U^{\top}(U\x_t - \y)\right)\right\|_2^2 + \lambda_t\|\x\|_1,
\end{align*}
with a gradually reduced  regularization parameter $\lambda_t$. It is also known that the proximal mapping above is one proximal gradient step for solving~(\ref{eqn:lasso}) with $\lambda_t$, i.e., 
\begin{align*}
\x_{t+1}&= \arg\min_{\x\in\R^d}\;  \frac{1}{2}\|\x - \x_t\|_2^2 +\left[ f(\x_t) + (\x-\x_t)^{\top}\nabla f(\x_t) \right] + \lambda_t \|\x\|_1 \\
&=\arg\min_{\x\in\R^d}\frac{1}{2}\|\x- \x_t\|_2^2  + \x^{\top}U^{\top}(U\x_t - \y) + \lambda_t\|\x\|_1,
\end{align*}
where the terms in the square bracket can be considered as a Taylor expansion of $f(\x)$ around $\x_t$.}
We prove that under RIP conditions for $U$ the solution $\x_t$ will converge {\bf linearly} to a solution $\bar\x$ that recovers the sparse signal up to the noise level. In particular, we establish the convergence results in three settings.  In the following presentation, we let $\x^s$ denote the vector $\x$ with all but the s-largest entries (in magnitude) set to zero.
\begin{enumerate}
\item[\bf Setting I:] Sparse signal recovery under noiseless observations. For any $s$-sparse vector $\x_*$, if $\e=0$ and $U$ satisfies RIP such that 
\begin{align}\label{eqn:rip}
\gamma = \delta_s + \sqrt{2}\delta_{2s} + \delta_{3s}< 1,
\end{align}
then the sequence $\x_{t+1}$ can converge linearly to $\x_*$, e.g., 
\begin{align*}
\|\x_{t+1}-\x_*\|_2\leq \gamma^t\Delta_1,
\end{align*}
where $\Delta_1$ is an upper bound of $\|\x_1 - \x_*\|_2$. 
\item[\bf Setting II:] Sparse signal recovery under noisy observations. For any $s$-sparse vector $\x_*$, if $U$ satisfies RIP such that~(\ref{eqn:rip}) holds, then $\x_{t+1}$ can converge linearly to a solution $\bar\x$ that recovers $\x_*$ up to the noise level, e.g., 
\begin{align*}
\|\x_{t+1} - \x_*\|_2\leq \gamma^t\Delta_1 + \frac{1+\sqrt{2}}{1-\gamma}\sqrt{s}\|U^{\top}\e\|_\infty,
\end{align*}
where $\gamma$ is given in~(\ref{eqn:rip}).
\item[\bf Setting III:] Nearly sparse signal recovery under a sub-gaussian measurement matrix $U$. For a fixed vector $\x_*$, with a probability $1-2te^{-\tau}$, $\x_{t+1}$ can converge linearly to a solution $\bar\x$ that recovers $\x^s_*$ up to the noise level, e.g., 
\begin{align}\label{eqn:3s}
\|\x_{t+1} - \x^s_*\|_2\leq \gamma^t\Delta_1 + \frac{1+\sqrt{2}}{1-\gamma}\Lambda,
\end{align}
where $\gamma = (1+\sqrt{2})\eta<1$ with $\eta$ and $\Lambda$ satisfying 
\begin{align*}
\eta &\geq c\sqrt{\frac{\tau + s\log[d/s]}{n}},\\ 
\Lambda&{= }\sqrt{s}\|U^{\top}\e\|_\infty + c\sqrt\frac{\tau + s\log[d/s]}{n}\|\x_* - \x^s_*\|_2+ c\|(\x_* - \x^s_*)^s\|_2\nonumber,
\end{align*}
where $c$ is a universal constant. 
\end{enumerate}
In addition, in all three settings considered above we show that  $|\text{supp}(\x_t)\setminus \text{supp}(\x^s_*)|\leq s$, where $\text{supp}(\x)$ denotes the support set of $\x$, which implies that the number of non-zero elements beyond $\text{supp}(\x^s_*)$ is no more than $s$. 


{ We note that the results in Settings I and II of the proposed algorithm hinge on appropriately setting the sequence of regularization parameters $\lambda_t$ that depend on the RIP constants.  
In Setting III, we develop a more practical algorithm with no algorithmic dependence on the RIP constants. However, it is notable that  Setting III is under a weaker model where the result only holds for a fixed vector $\x_*$ unlike that Settings I and II apply to any sparse vector $\x_*$. Indeed, the result in Setting III holds for any Johnson-Lindenstrauss (JL) transforms that satisfy the JL lemma~\citep{citeulike:7030987} and the high probability is with respect to the randomness in the measurement matrix. }  In Section II,  we briefly discuss the above results in comparison  with previous work.

\section{Related Work}
We first compare our recovery results with state of the art results for (nearly) sparse signal recovery and then discuss about the optimization algorithms for sparse signal recovery. 

\paragraph{Sparse signal recovery with noiseless observations} \cite{Candes:2005:DLP:2263433.2271950} analyzed the recovery result for solving the $\ell_1$ minimization problem~(\ref{eqn:l1}) with noiseless observations $\y=U\x_*$,  and showed that for any $s$-sparse signal $\x_*$ when $U$ satisfies RIP~\footnote{Using the restricted orthogonality constant $\theta_{s,s'}$ defined in {\bf Definition 2}, a better condition on RIP constants can be established  in their result as well as in our analysis. We use the restricted isometry constant $\delta_{s}$ in order to compare with other works and benefit from previous methods that estimate $\delta_s$.} such that
\begin{align}~\label{eqn:opt}
\delta_s + \delta_{2s} + \delta_{3s}< 1, 
\end{align}
then the optimal solution to~(\ref{eqn:l1}) with $\epsilon=0$ is unique and is equal to $\x_*$. Comparing the inequality~(\ref{eqn:rip}) and~(\ref{eqn:opt}), our condition for exact recovery is close to the above condition. The exact recovery was also indicated in Cand\`{e}s' later work~\citep{citeulike:10486729} but with a slightly different RIP condition $\delta_{2s}\leq \sqrt{2}-1$. 
\paragraph{Sparse signal recovery with noisy observations} \citet{citeulike:10486729}  proved a recovery result for noisy observations. For any $s$-sparse vector $\x_*$, when $U$ satisfies RIP such that $\delta_{2s}\leq \sqrt{2} -1$, the optimal solution $\bar\x$ to~(\ref{eqn:l1}) by setting $\epsilon\geq \|\e\|_2$ obeys  
\[
\|\bar\x - \x_*\|_2 \leq C_2 \epsilon,
\]
where $C_2$ 
is a constant depending on $\delta_{2s}$. In comparison, our recovery error  in {\bf Setting II} depends on $\sqrt{s}\|U\e\|_\infty$ which could be smaller than $\|\e\|_2$ (e.g., when the entries in $U$ are sub-gaussian as stated in Proposition~\ref{lem:l2linf} in the appendix). 

\paragraph{Nearly sparse signal recovery with noisy observations} A more general recovery result was also established in~\citep{citeulike:10486729}. For any vector $\x_*$, when $U$ satisfies RIP such that $\delta_{2s}\leq \sqrt{2} -1$, the optimal solution $\bar\x$ to~(\ref{eqn:l1}) by setting $\epsilon\geq \|\e\|_2$ obeys  
\[
\|\bar\x - \x_*\|_2 \leq C_0\frac{\|\x_*-\x^s_*\|_1}{\sqrt{s}} +  C_2 \epsilon,
\]
where $C_0$ is a constant depending on $\delta_{2s}$. 
Similar results have also been developed for the Dantzig selector~(\ref{eqn:dan}).  Namely,  when the RIP constant $\delta_{2s}$ of $U$ satisfies $\delta_{2s}\leq \sqrt{2} -1$, the optimal solution $\bar\x$ to~(\ref{eqn:dan}) by setting $\lambda\geq \|U^{\top}\e\|_\infty$ satisfies 
\[
\|\bar\x - \x_*\|_2\leq C_0\frac{\|\x_* - \x_*^s\|_1}{\sqrt{s}} + C_3\sqrt{s}\lambda,
\]
where $C_3$ is a constant depending on $\delta_{2s}$. 
In contrast, in {\bf Setting III}, we established a better recovery result for a fixed signal $\x_*$. From~(\ref{eqn:3s}), we can see that the full recovery error $\|\bar\x - \x_*\|_2$ depends on the $\ell_2$ norm $\|\x_* - \x^s_*\|_2$ instead of $\|\x_* - \x^s_*\|_1/\sqrt{s}$. 

{It is worth mentioning that there exist a battery of studies on establishing sharper conditions on the RIP constants for exact or accurate recovery (see~\citep{DBLP:journals/tit/CaiZ14} and references therein). \cite{DBLP:journals/tit/CaiZ14} established sharpest condition on the RIP constant $\delta_{ts}$ for $t\geq 4/3$. In particular, they show that $\delta_{ts}\leq \sqrt{\frac{t-1}{t}}$ for $t\geq 4/3$ is sufficient for exact recovery under noiseless measurements and accurate recovery under noisy measurements. Nevertheless, we make no attempts to sharpen the condition on RIP constants but rather focus on the optimization algorithms and their recovery properties. }

\paragraph{Instance-level recovery result}{}
 {A weaker recovery result is that given a fixed signal $\x_*$, we can draw a random measurement matrix $U$ and with a high probability expect certain performance for the recovery of the signal $\x_*$. We refer to this type of guarantee as  instance-level recovery result~\citep{eldar2012compressed}. An advantage of the instance-level recovery is that  we can achieve a recovery error  in the form of  $\|\bar\x- \x^s_*\|_2\leq C\|\x_* - \x^s_*\|_2$ with $C$ being a constant and $\bar\x$ being the recovered signal. However, such a result  is impossible for any signal $\x_*$ without using a large number of observations, or in other words, such a result is only possible for any signal $\x_*$ when $n\geq cd$ for a constant $c>0$ (i.e., $n=\Omega(d)$).} 
 In~\citep{eldar2012compressed}, it was shown that when  the observations are free of noise and $U\in\R^{n\times d}$ is a sub-gaussian random matrix with  $n=O(s\log(d/s)/\delta_{2s}^2)$, then for a fixed signal $\x_*$ with a probability $1- 2\exp(-c_1\delta_{2s}^2n) - \exp(c_0n)$, the optimal solution $\bar\x$ to~(\ref{eqn:l1}) with  $\epsilon = 2\|\x_* -\x^s_*\|_2$ obeys 
\begin{align}
\|\bar\x  - \x^s_*\|_2&\leq 2C_2\|\x_* - \x^s_*\|_2,\label{eqn:ins}\\
\|\bar\x  - \x_*\|_2&\leq (2C_2+1)\|\x_* - \x^s_*\|_2,
\end{align}
where $C_2>4$ is a constant depending on $\delta_{2s}$. In contrast, our sparse signal recovery result for $\|\bar\x - \x^s_*\|_2$ in {\bf Setting III} (considering no noise) is much better than that in~(\ref{eqn:ins}) since the error is dominated by $O\left(\|(\x_* - \x^s_*)^s\|_2 + \sqrt{\frac{s\log[d/s]}{n}}\|\x_* - \x^s_*\|_2\right)$, where $\|(\x_* - \x^s_*)^s\|_2$ is the $\ell_2$ norm of the largest $s$ elements in $\x_* - \x^s_*$. To the best of our knowledge, this is the first such result in the literature. 

There are also many studies on analyzing the properties of the optimal solution $\bar\x$ to the $\ell_1$ regularized minimization problem in~(\ref{eqn:lasso})~\citep{citeulike:2823189,journals/tit/Tropp06,Zhao:2006:MSC:1248547.1248637,zhang2008,citeulike:5426408,Bickel09simultaneousanalysis,GeeBue09,Wainwright:2009:STH:1669487.1669506}. It is known that under RIP condition for $U$ and   $\lambda>c\|U^{\top}\e\|_\infty$ (for some universal constant $c$), we can obtain a recovery bound for any $s$-sparse signal $\x_*$
\[
\|\bar\x - \x_*\|_2\leq O(\sqrt{s}\lambda).
\]
In comparison, our analysis also exhibits that the final  value of $\lambda_t$ is $\Omega(\|U^{\top}\e\|_\infty)$ for sparse signal recovery. 
More literature on sparse signal recovery  can be found in~\citep{eldar2012compressed}.

\paragraph{Optimization algorithms} There have been extensive research on solving the $\ell_1$ minimization problems in~(\ref{eqn:l1}) and~(\ref{eqn:dan}), and the $\ell_1$ regularized minimization problem in~(\ref{eqn:lasso}). Various algorithms have been developed, including greedy algorithms~\citep{Davis2004,Tropp:2006:GGA:2263414.2271370,Needell:2010:CIS:1859204.1859229,Mallat:1993:MPT:2198030.2203996,4385788,Donoho:2012:SSU:2331864.2332326,Needell:2009:UUP:1530720.1530722}, interior-point methods~\citep{Chen:2001:ADB:588736.588850,citeulike:259101,kim2008interior}, proximal gradient methods~\citep{RePEc:cor:louvco:2007076,citeulike:6604666,Beck:2009:FIS:1658360.1658364,Becker:2011:NFA:2078698.2078702}, exact homotopy path-following methods~\citep{citeulike:7675119,Osborne99onthe,Efron04leastangle}, iterative hard-thresholding methods~\citep{Garg:2009:GDS:1553374.1553417,Blumensath_iterativehard,Foucart:2011:HTP:2340478.2340494,DBLP:journals/jmiv/KyrillidisC14}.  In~\citep{Garg:2009:GDS:1553374.1553417}, the authors gave a nice review of the convergence rates and their computational costs for different optimization algorithms.  Below, we focus on two classes of algorithms that are closely related to the proposed work, with one employing the iterative hard-thresholding and the other exploiting the iterative soft-thresholding. 

The hard-thresholding amounts to updating the solution based on the exact sparsification, i.e., 
\begin{align*}
\x_{t+1} = H_s\left(\x_t - \frac{1}{\gamma}U^{\top}(U\x_t - \y)\right),
\end{align*}
where $\gamma$ is a constant and $H_s(\x) = \x^s$ is the hard-thresholding operator that gives the best s-sparse approximation of a vector x, i.e., setting all elements in $\x$ to zeros except for the $s$ largest elements in magnitude. In~\citep{Blumensath_iterativehard}, the authors analyzed the iterative hard-thresholding algorithm with $\gamma=1$. They show that when $U$ satisfies RIP with $\delta_{3s}<1/\sqrt{32}$, the sequence $\{\x_t\}$ converges linearly to the best attainable solution up to a constant, i.e.,  
\begin{align}\label{eqn:hard}
\|\x_t-\x_*\|_2&\leq 2^{-t}\|\x_*\|_2+ 6\left[\|\x_* - \x^s_*\|_2 + \|\e\|_2 + \frac{1}{\sqrt{s}}\|\x_* - \x^s_*\|_1\right].
\end{align}
Similarly, \cite{Garg:2009:GDS:1553374.1553417} analyzed the iterative hard-thresholding with $\gamma=1+\delta_{2s}$ under the {\bf Settings I and II}, and showed the sequence $\{\x_t\}$ converges to a solution $\bar\x$ that recovers any $s$-sparse signal $\x_*$ signal up to the noise level, i.e., $\|\bar\x - \x_*\|_2\leq \frac{4}{1-\delta_{2s}}\|\e\|_2$ with a rate of $\left(\frac{8\delta_{2s}}{1-\delta_{2s}}\right)^t$ under the condition $\delta_{2s}\leq 1/3$. In contrast, the proposed algorithm in {\bf Settings I and II} only requires $\delta_s + \sqrt{2}\delta_{2s} + \delta_{3s}<1$, which is less restricted than $\delta_{3s}\leq 1/\sqrt{32}$ or $\delta_{2s}\leq 1/3$.  In {\bf Setting III}, we proved a recovery for a fixed signal $\x_*$ with a high probability. Comparing~(\ref{eqn:3s}) and~(\ref{eqn:hard}), we could see that the upper bound of the recovery of the proposed algorithm might be  tighter than that of the iterative hard-thresholding algorithm, since our bound depends on $\|\x_* - \x^s_*\|_2$ instead of $\|\x_* - \x^s_*\|_1/\sqrt{s}$. 

The iterative soft-thresholding algorithms (ISTA) are based on the proximal mapping of $\ell_1$ regularization for solving the $\ell_1$ regularized minimization problem~(\ref{eqn:lasso}), where the updates are given by  
\begin{align*}
\x_{t+1}=\arg\max_{\x\in\mathbb R^d} \x^{\top}\nabla_t + \frac{\gamma_t}{2}\|\x - \x_t\|_2^2 + \lambda\|\x\|_1,
\end{align*}
where $\nabla_t$  is set to the gradient of the square error w.r.t $\x_t$, and $\gamma_t$ is a step size. The proximal mapping springs from Nesterov's first order method for composite optimization~\citep{RePEc:cor:louvco:2007076}.  In~\citep{bredies-2008-linear,hale-2008-fixed}, the authors studied the soft-thresholding update with a constant step size and established local linear convergence rates as the iterates are close enough to the optimum. {There are several striking differences between ISTA and the proposed algorithms, including Algorithms~\ref{alg:1}, \ref{alg:2} and \ref{alg:3}. First, ISTA solves exactly the $\ell_1$ regularized least-squares  problem (i.e., the BPDN problem) with a {\it fixed} regularization parameter. The proposed algorithms  are proposed to directly reconstruct a sparse signal from noisy measurements. Second, if using the BPDN formulation to recover a sparse signal requires the algorithm needs to know the regularization parameter $\lambda$ such that $\lambda\geq \Omega(\|U^{\top}\e\|_\infty)$. However, the proposed Algorithm 3 does not need any knowledge about the order of $\|U^{\top}\e\|_\infty$. Instead, it uses the proximal mapping of an $\ell_1$ norm regularizer with a gradually decreasing regularization parameter $\lambda_t$ until the solution exceeds the target sparsity by two times. Third, the proposed algorithms enjoy global linear convergence, while ISTA  has only local linear convergence when the solution is close enough to the optimal solution. Last but not least, the presented algorithms and analysis provide a unified framework of optimization and recovery of sparse signal. In contrast,   ISTA is only an optimization algorithm which solely provides no guarantee on  the recovery of underlying true sparse signal.}

Recently, several algorithms  exhibit  global linear convergence for the BPDN problem.   \cite{DBLP:conf/nips/AgarwalNW10} studied an optimization problem~(\ref{eqn:lasso}) for statistical recovery. They used a  different update 
\begin{align}\label{eqn:cons}
\max_{\x\in\mathcal X} \x^{\top}\nabla_t + \frac{\gamma_u}{2}\|\x - \x_t\|_2^2 + \lambda\|\x\|_1,
\end{align}
where $\mathcal X = \{\x\in\R^d\:|\: \|\x\|_1\leq \rho\}$, and $\gamma_u$ is a parameter related to the restricted smoothness of the loss function. They proved a global linear convergence of the above update with $\rho =\Theta(\|\x_*\|_1)$ for finding a solution up to the statistical tolerance. \cite{DBLP:journals/siamjo/Xiao013} studied   a proximal-gradient homotopy gradient method for solving~(\ref{eqn:lasso}). They iteratively solve the problem~(\ref{eqn:lasso}) by the proximal gradient descent with a decreasing regularization parameter $\lambda$ and an increasing accuracy at each stage, and use the solution obtained at each stage to warm start the next stage. A global linear convergence was also established.

Although there are many parallels between this work and~\citep{DBLP:conf/nips/AgarwalNW10,DBLP:journals/siamjo/Xiao013},  there are big differences. (i) The proposed work is dedicated to sparse signal recovery, exhibiting the conditions in different settings under which the recovery is optimal. (ii) Different from~\citep{DBLP:conf/nips/AgarwalNW10} that updates the solution using the constrained proximal mapping in~(\ref{eqn:cons}), our algorithms  solve a simple proximal mapping of the $\ell_1$ norm regularization at each iteration. (iii) Different from~\citep{DBLP:journals/siamjo/Xiao013} that updates the solution using a stage-wise proximal gradient descent with pesky parameters, the proposed homotopy proximal mapping algorithm is much simpler as well as the analysis. (iv) Our algorithm and analysis provide better guarantees for the solutions. First, both the convergence rates and the recovery error of the proposed algorithms are directly related to the RIP constants (in {\bf Settings I and II}) or the number of observations (in {\bf Setting III}), implying that more observations lead to not only more accurate recovery  but also faster convergence. Second, our algorithm can guarantee that the support sets of the intermediate solutions  do not exceed the target support set by $s$, the target sparsity. In contrast, \citep{DBLP:conf/nips/AgarwalNW10} provides no explicit guarantee of sparsity bound for the intermediate solutions, and in~\citep{DBLP:journals/siamjo/Xiao013} the support sets of the intermediate solutions beyond the target support set could be much larger than $s$.

\section{Sparse Signal Recovery}
\subsection{Notations and Definitions}

Let $\x_* \in \R^d$ be a $s$-sparse high dimensional signal to be recovered, where the number of non-zero elements in $\x_*$ is $s$. We denote by $\S(\x)$ the support set for $\x$ that includes all the indices of the non-zero entries in $\x$, i.e.,
\begin{eqnarray}
    \S(\x) = \left\{i\in\{1,\ldots, d\}: [\x]_i \neq 0 \right\} \label{eqn:S},
\end{eqnarray}
where $[\x]_i$ denote the $i$-th element in $\x$.  Denote by $\S_1 \setminus \S_2$ a subset of $\S_1$ that contains all elements in $\S_1$ but not in $\S_2$. We also denote by $\Sb(\x) = \{1,\ldots, d\}\setminus \S(\x)$ the complementary set of $\S(\x)$. In particular, we use $\S_*, \Sb_*$ to denote the support set and its complementary set of $\x_*$.  Let $|\S|$ denote the cardinality of $\S$, and let $\x^s\in\R^d$ denote the vector $\x\in\R^d$ with all but the s-largest entries (in magnitude) set to zero.  Denote by $\|\x\|_2$, $\|\x\|_1$, $\|\x\|_\infty$ and $\|\x\|_0$ the $\ell_2$, $\ell_1$, $\ell_\infty$ and $\ell_0$ norm, respectively. 

Consider a vector $\x\in\R^d$ and a matrix $M\in\R^{n\times d}$.  Given a set $\S \subseteq \{1,\ldots, d\}$, we denote  by $[\x]_{\S}\in\R^{|\S|}$ the vector that only includes the entries of $\x$ in the subset $\S$, and by $M_{\S}$ a sub-matrix that only contains the columns of $M$ indexed by $\S$. Given two subsets $\A \subseteq \{1,\ldots, d\}$ and $\B \subseteq \{1,\ldots, d\}$, we denote by $[M]_{\A, \B}$ a sub-matrix that includes all the entries $(i,j)$ in matrix $M$ with $i \in \A$ and $j \in \B$. {$\|M\|_2$ denotes the spectral norm of a matrix $M$}.

Let $U\in\R^{n\times d}$ be a measurement matrix  and 
\begin{align}\label{eqn:obs}
\y=U\x_* + \e
\end{align} be the corresponding $n$ observations of the target signal $\x_*$. Similar to most work in compressive sensing,  we assume the measurement matrix $U$ satisfies the following restricted isometry properties (RIP) (with an overwhelming probability). 
\begin{definition}[$s$-restricted isometry constant]
Let  $\delta_s\geq 0$ be the smallest constant  such that for any subset $\mathcal T\subseteq\{1,\ldots, d\}$ with $|\mathcal T|\leq s$ and $\x\in\R^{|\mathcal T|}$,
\begin{align*}
(1-\delta_s)\|\x\|_2^2\leq \|U_{\mathcal T}\x\|_2^2\leq (1+\delta_s)\|\x\|_2^2
\end{align*}
where $U_{\mathcal T}$ denotes a sub-matrix of $U$ with column indices  from $\mathcal T$. 
\end{definition}
{\bf Remark 1:}  The RIP above implies that $U_{\mathcal T}^{\top}U_{\mathcal T}$ has all of its eigen-values in $[1-\delta_s, 1+\delta_s]$. As a result $\|(U_{\mathcal T}^{\top}U_{\mathcal T}-I)\x\|_2\leq \delta_s \|\x\|_2$. 

\begin{definition}[$s,s$-restricted orthogonality  constant]
Let  $\theta_{s,s}$ be the smallest constant  such that for any two disjoint  subsets $\mathcal T, \mathcal T'\subseteq\{1,\ldots, d\}$ with $|\mathcal T|\leq s$, $|\mathcal T'|\leq s$, $2s\leq d$, and for any $\x\in\R^{|\mathcal T|}$,  $\x'\in\R^{|\mathcal T'|}$, 
\begin{align*}
|\langle U_{\mathcal T}\x,  U_{\mathcal T'}\x'\rangle| \leq \theta_{s,s}\|\x\|_2\|\x'\|_2
\end{align*}
\end{definition}
{\bf Remark 2:} The above RIP implies that $\|U_{\mathcal T}^{\top}U_{\mathcal T'}\|_2\leq \theta_{s,s}$.  Although the results in the sequel  are stated using $\delta_s$ and $\theta_{s,s}$, we can easily obtained the results with only restricted isometry constants   by noting that $\theta_{s,s}\leq \delta_{2s}$~\citep{Candes:2005:DLP:2263433.2271950}. 

The above two constants are standard tools in the analysis  of compressive sensing.  It has been shown that several random measurement matrices including sub-gaussian measurement matrix, Fourier measurement matrix and incoherent measurement matrix satisfy the above RIP  with small $\delta_s$ and $\theta_{s,s}$~\citep{citeulike:2688127}.  

\subsection{Algorithm and Main Results}
To motivate our approach, we first consider the following optimization problem
\begin{eqnarray}
    \min\limits_{\x\in\mathbb R^d} \quad \L(\x) = \frac{1}{2}\|\x - \x_*\|_2^2 \label{eqn:1}.
\end{eqnarray}
Evidently, the optimal solution to (\ref{eqn:1}) is $\x_*$. We now consider a gradient descent method for optimizing the problem in (\ref{eqn:1}), leading to the following updating equation for $\x_t$
\begin{eqnarray}
    \x_{t+1} = \mathop{\arg\min}\limits_{\x\in\R^d} \frac{1}{2}\left\|\x -( \x_{t} -  \nabla \L(\x_t))\right\|_2^2 \label{eqn:update-1},
\end{eqnarray}
where $\nabla \L(\x_t) = \x_t - \x_*$. Since the problem in (\ref{eqn:1}) is both smooth and strongly convex, the above updating enjoys a linear convergence rate {with in fact only one step}, allowing an efficient reconstruction of $\x_*$.

However, the updating rule in (\ref{eqn:update-1}) can not be used because it requires knowing $\x_*$, the full information of the sparse signal to be recovered. In compressive sensing, the only available information about the target signal $\x_*$ is through a set of $n<d$ observations given in~(\ref{eqn:obs}). 
Using the observations, we construct an approximate gradient as
\begin{equation}
    \nh \L(\x_t)=U^{\top}(U\x_t - \y) = U^{\top}U(\x_t - \x_*) - U^{\top}\e \label{eqn:gradient}
\end{equation}
{As can be seen if $U^TU(\x_t - \x_*)$ is close to $\x_t - \x_*$  and $U^{\top}\e$  is not significantly large in magnitude, $\nh\L(\x_t)$ would provide an useful estimate  of $\nabla\L(\x_t)$.} To ensure this, we should assume certain restricted conditions  on $U$ and a small noise $\e$.

Next, we will use $\nh\L(\x_t)$ as an approximation of $\nabla \L(\x_t)$ and update the solution by performing the following  proximal  mapping: 
\begin{equation}\label{eqn:update-2}
\begin{aligned}
   & \x_{t+1}=\mathop{\arg\min}\limits_{\x \in \R^d} \lambda_t \|\x\|_1 + \langle\x - \x_t, \nh \L(\x_t) \rangle + \frac{1}{2}\|\x - \x_t\|_2^2 
    \end{aligned}
\end{equation}
where $\lambda_t > 0$ is a $\ell_1$ norm regularization parameter that decreases over iterations. 
The updating rule given in (\ref{eqn:update-2}) differs from (\ref{eqn:update-1}) in that (i) the true gradient $\nabla\L(\x_t)$ is replaced with an approximate gradient $\nh\L(\x_t)$ and (ii) an  $\ell_1$ regularization term $\lambda_t\|\x\|_1$ is added. With appropriate choice of $\lambda_t$, this regularization term will essentially remove the noise arising from the approximate  gradient and consequentially lead to a global linear  convergence rate. 

\begin{algorithm}[t] 
\caption{Homotopy Proximal Mapping (HPM) for Compressive Sensing}
\begin{algorithmic}[1]

\STATE {\bf Input:} The measurement  matrix $U \in \R^{n\times d}$, observations $\y = U\x_*+\e$, a sequence of regularization parameters $\lambda_1, \ldots, \lambda_T$

\STATE {\bf Initialize} $\x_1 = 0$.
\FOR{$t=1, \ldots, T$}
\STATE Compute $\displaystyle\xh_{t}= \x_t -U^{\top}(U\x_t - \y)$
    \STATE Update the solution $\displaystyle \x_{t+1} = sign(\xh_t)\left[|\xh_t| - \lambda_t\right]_+$
    \ENDFOR

\STATE {\bf Output} the final solution $\x_{T+1}$
\end{algorithmic}
\label{alg:1}
\end{algorithm}


To give the solution of $\x_{t+1}$ in a closed form, we write~(\ref{eqn:update-2}) as 
\begin{equation}\label{eqn:proxx}
    \x_{t+1} = \mathop{\arg\min}\limits_{\x \in \R^d} \frac{1}{2}\left\|\x - \left(\x_t- U^{\top}(U\x_t - \y)\right)\right\|_2^2 + \lambda_t\|\x\|_1 
\end{equation}  It is commonly known that  the value of $\x_{t+1}$ is given by~\citep{Beck:2009:FIS:1658360.1658364}
\begin{equation}\label{eqn:xt}
\x_{t+1} = sign(\xh_t)\left[|\xh_t| - \lambda_t\right]_+
\end{equation}
where $\xh_t$ denotes the intermediate solution before soft-thresholding given by 
\begin{align}\label{eqn:xh}
\xh_t = \x_t - U^{\top}(U\x_t - \y)
\end{align} and $[v]_+=\max(0, v)$.  
We present the detailed steps of the proposed algorithm in   Algorithm~\ref{alg:1}  for reconstructing the sparse signal given a set of noiseless/noisy observations.   To end this section, we present our main result in the following two theorems regarding the sparse signal recovery with noiseless observations and with noisy observations.  
\begin{thm} \label{thm:main1}
Let $\x_*\in\mathbb R^d$ be a $s$-sparse signal and $\y=U\x_*$ be a set of $n$ measurements of $\x_*$. Assume $U$ satisfies RIP such that 
\[
\gamma = \delta_s + \sqrt{2}\theta_{s,s} + \delta_{3s}<1. 
\] 
Let $\{\Delta_1,\ldots, \Delta_t\}$ be a sequence such that $\|\x_1 - \x_*\|_2\leq \Delta_1$, and 
\[
\Delta_{t+1} = (\delta_s + \sqrt{2}\theta_{s,s} + \delta_{3s})\Delta_{t}.
\]  If we run Algorithm 1 with $\displaystyle\lambda_t = \frac{\delta_s + \sqrt{2}\theta_{s,s}}{\sqrt{s}}\Delta_t$, 
then for all $t\geq 0$ \begin{itemize}
\item  $|\S_{t+1}\setminus\S_*|\leq s$ and, 
\item  $\|\x_{t+1} - \x_*\|_2 \leq \gamma^{t} \Delta_1$.
\end{itemize} 
\end{thm}
{\bf Remark 3:} Similar to iterative hard-thresholding algorithms~\citep{Garg:2009:GDS:1553374.1553417,Blumensath_iterativehard}, Algorithm 1 also requires  knowledge of sparsity $s$ and RIP constants. 

\begin{thm} \label{thm:main2}
Let $\x_*\in\mathbb R^d$ be a $s$-sparse signal and $\y=U\x_*+\e$ be a set of $n$ noisy measurements of $\x_*$. Assume $U$ satisfies RIP such that 
\[
\gamma = \delta_s + \sqrt{2}\theta_{s,s} + \delta_{3s}<1. 
\] 
Let $\{\Delta_1,\ldots, \Delta_t\}$ be a sequence such that $\|\x_1 - \x_*\|_2\leq \Delta_1$, and 
\[
\Delta_{t+1} = \gamma\Delta_t +  (1+\sqrt{2})\sqrt{s}\|U^{\top}\e\|_\infty, \quad t\geq 1. 
\]
If we run Algorithm 1 with 
\[
\displaystyle\lambda_t = \frac{\delta_s + \sqrt{2}\theta_{s,s}}{\sqrt{s}}\Delta_t + \|U^{\top}\e\|_\infty,
\] 
 then for all $t\geq 0$ \begin{itemize}
\item  $|\S_{t+1}\setminus\S_*|\leq s$ and, 
\item  $\displaystyle\|\x_{t+1} - \x_*\|_2 \leq  \gamma^t\Delta_1 + \frac{1-\gamma^t}{1-\gamma}(1+\sqrt{2})\sqrt{s}\|U^{\top}\e\|_\infty$.
\end{itemize} 
\end{thm}

{\bf Remark 4:} Similar to solving Dantzig selector~(\ref{eqn:dan}) and the $\ell_1$ regularized problem~(\ref{eqn:lasso}) for sparse signal recovery that requires $\lambda\geq c\|U^{\top}\e\|_\infty$, the regularization parameters in our algorithm are also larger than $\|U^{\top}\e\|_\infty$ and eventually $\lambda_t\geq c\|U^{\top}\e\|_\infty$, where $c$ depends RIP constants.

{\bf Remark 5:} While Theorems~\ref{thm:main1} and~\ref{thm:main2} are theoretically interesting, the value of $\lambda_t$ depends on the RIP constants. In Section~\ref{sec:nearly}, we present more practical algorithms for (nearly) sparse signal recovery with a sub-gaussian measurement matrix. 
\subsection{Proof of Theorem~\ref{thm:main1}}
We first prove the following proposition  regarding the magnitude of elements in $[\xh_t]_{\Sb_*}$. 
\begin{prop} \label{thm:sabc}
Let $\S_t$ be the support set of $\x_t$ ({the $t$-th iterate of Algorithm~\ref{alg:1}}) and $\S_*$ be the support set of $\x_*$. Define $\S^c_t = \S_t\cup \S_*$, $\S^a_{t} = \S^c_t \setminus \S_*$ and $\xt_t  = \x_t - U^{\top}U(\x_t - \x_*) $.  If we assume $|\S_t\setminus \S_*|\leq s$, then there are at most $s$ entries of $[\xt_t]_{\Sb_*}$ with magnitude larger than
$
      \displaystyle   \frac{\delta_{s}+\sqrt{2}\theta_{s,s} }{\sqrt{s}}\|\x_t - \x_*\|_2
$. 
\end{prop}
\begin{proof}
For any subset $\S' \subset \Sb_*$ of size $s$, let $\S'_1=\S'\cap\S_t^a$ and $\S'_2=\S'\setminus\S_t^a$. 
First, we have
\begin{equation*}
\begin{aligned}
   &\|[\xt_t]_{\S'}\|_2= \left\|[U^{\top}U(\x_t-\x_*) ]_{\S'}- [\x_t]_{\S'}\right\|_2=  \left\|U^{\top}_{\S'}U_{\S_*}\left[\x_t - \x_*\right]_{\S_*} + U^{\top}_{\S'}U_{\S^a_t}\left[\x_t\right]_{\S^a_t} - [\x_t]_{\S'}\right\|_2 \\
   \end{aligned}
   \end{equation*}
where the second equality is due to that the support of $\x_t - \x_*$ is $\S^c_t$ and we split that into two subsets $\S_t^a$ and $\S_*$ that do not intersect with each other. 
By noting that $\S'$ can be split into two subsets $\S_1'$ and $\S_2'$ that do not intersect with each other and that $\|[\v]_{\S'}\|_2 \leq \|[\v]_{\S_1'}\|_2 + \|[\v]_{\S_2'}\|_2$ with $\v = U^{\top}U_{\S^a_t}[\x_t]_{\S^a_t} - \x_t$, we have
 \begin{equation*}
\begin{aligned} 
   &\left\|U^{\top}_{\S'}U_{\S_*}\left[\x_t - \x_*\right]_{\S_*} + U^{\top}_{\S'}U_{\S^a_t}\left[\x_t\right]_{\S^a_t} - [\x_t]_{\S'}\right\|_2\\
       & \leq \left\|U^{\top}_{\S'}U_{\S_*}\left[\x_t - \x_*\right]_{\S_*}\right\|_2 +\left\|U^{\top}_{\S'}U_{\S^a_t}\left[\x_t\right]_{\S^a_t} - [\x_t]_{\S'}\right\|_2  \\
    & \leq \left\|U^{\top}_{\S'}U_{\S_*}\left[\x_t - \x_*\right]_{\S_*}\right\|_2 +\left\|U_{\S'_2}^{\top}U_{\S^a_t}\left[\x_t\right]_{\S^a_t} - [\x_t]_{\S'_2}\right\|_2  + \left\|U_{\S'_1}^{\top}U_{\S^a_t}\left[\x_t\right]_{\S^a_t} - [\x_t]_{\S'_1}\right\|_2\\
        & =\left\|U^{\top}_{\S'}U_{\S_*}\left[\x_t - \x_*\right]_{\S_*}\right\|_2 +\left\|U_{\S'_2}^{\top}U_{\S^a_t}\left[\x_t\right]_{\S^a_t}\right\|_2  + \left\|U_{\S'_1}^{\top}U_{\S^a_t}\left[\x_t\right]_{\S^a_t} - [\x_t]_{\S'_1}\right\|_2\\
        &\leq \left\|U^{\top}_{\S'}U_{\S_*}\right\|_2\left\|\left[\x_t - \x_*\right]_{\S_*}\right\|_2 +\left\|U_{\S'_2}^{\top}U_{\S^a_t}\right\|_2\left\|\left[\x_t\right]_{\S^a_t}\right\|_2 +  \left\|U_{\S^a_t}^{\top}U_{\S^a_t}\left[\x_t\right]_{\S^a_t} - [\x_t]_{\S^a_t}\right\|_2\\
    &= \left\|U^{\top}_{\S'}U_{\S_*}\right\|_2\left\|\left[\x_t - \x_*\right]_{\S_*}\right\|_2 +\left\|U_{\S'_2}^{\top}U_{\S^a_t}\right\|_2\left\|\left[\x_t\right]_{\S^a_t}\right\|_2 +  \left\|(U_{\S^a_t}^{\top}U_{\S^a_t}-I)\left[\x_t\right]_{\S^a_t} \right\|_2\\
    & \leq  \theta_{s,s}\|[\x_t - \x_*]_{\S_*}\|_2 + \theta_{s,s}\|[\x_t]_{\S^a_t}\|_2 + \delta_s\|[\x_t]_{\S^a_t}\|_2\\
        & = \theta_{s,s}\|[\x_t - \x_*]_{\S_*}\|_2 + \theta_{s,s}\|[\x_t - \x_*]_{\S^a_t}\|_2 + \delta_s\|[\x_t - \x_*]_{\S^a_t}\|_2\\
    & \leq ( \delta_s  + \sqrt{2}\theta_{s,s}) \|\x_t - \x_*\|_2
\end{aligned}
\end{equation*}
where the first equality uses the fact that $[\x_t]_{\S'_2}=0$ , the third inequality uses the fact $\S'_1\subseteq\S^a_t$,  the fourth   inequality uses the RIP conditions (see Remark 1 and Remark 2) by noting that $|\S^a_t|\leq s$, $|\S_2'|\leq s$, $|\S'|\leq s$ and $|\S_*|\leq s$, and the last inequality uses the fact $a+b\leq \sqrt{2(a^2+b^2)}$ for $a = \|[\x_t - \x_*]_{\S_*}\|_2$ and $b=\|[\x_t - \x_*]_{\S^a_t}\|_2$. Combining the above inequalities we have

\begin{equation}\label{eqn:fac}
\begin{aligned}
   &\|[\xt_t]_{\S'}\|_2 \leq ( \delta_s  + \sqrt{2}\theta_{s,s}) \|\x_t - \x_*\|_2
   \end{aligned}
   \end{equation}
Since the above inequality holds for any subset $\S' \subseteq \Sb_*$ of size $s$, we form a particular set $\S'$ by including the largest $s$ entries in absolute value of $[\xt_t]_{\Sb_*}$. Then the smallest absolute value in $[\xt_t]_{\S'}$ is less than $\displaystyle\frac{\delta_s + \sqrt{2}\theta_{s,s}}{\sqrt{s}}\|\x_t - \x_*\|_2$.  {If not, then $\|[\xt_t]_{\S'}\|_2\geq \sqrt{s}\frac{\delta_s + \sqrt{2}\theta_{s,s}}{\sqrt{s}}\|\x_t - \x_*\|_2 = (\delta_s + \sqrt{2}\theta_{s,s})\|\x_t - \x_*\|_2$, which contradicts to the result in~(\ref{eqn:fac}).}
By the construction of $\S'$,  the smallest entry (in magnitude) in $\S'$ is the $s$th largest entry (in magnitude)   in $[\x_t- U^{\top}U(\x_t-\x_*)]_{\Sb_*}$, we conclude that at most $s$ entries in  $[\xt_t]_{\Sb_*}=[\x_t- U^{\top}U(\x_t-\x_*)]_{\Sb_*}$ are larger than $ \displaystyle \frac{\delta_s + \sqrt{2}\theta_{s,s}}{\sqrt{s}}\|\x_t- \x_*\|_2$ in magnitude.
\end{proof}
As an immediate result of Proposition~\ref{thm:sabc}, we prove the following Corollary. 
\begin{cor}\label{cor:xt}
{Assume the noiseless setting $\e=0$.}  Let $\S_t$ be the support set of $\x_t$ and $\S_*$ be the support set of $\x_*$. If $|\S_t\setminus S_*|\leq s$ and $\lambda_t\geq \frac{\delta_s + \sqrt{2}\theta_{s,s}}{\sqrt{s}}\|\x_t-\x_*\|_2$, then  $|\S_{t+1}\setminus \S_*|\leq s$ and $|\S_*\cup \S_t\cup \S_{t+1}|\leq 3s$. 
\end{cor}
\begin{proof}
Note that in the noiseless setting when $\e=0$, {the intermediate solution $\xh_t$ defined in~(\ref{eqn:xh})} is equal to  
\begin{align}\label{eqn:int2}
\xh_t  = \x_t - U^{\top}(U\x_t - \y) = \x_t - U^{\top}U(\x_t - \x_*)
\end{align}
As shown in~(\ref{eqn:xt}), $\x_{t+1}$ is given by 
\begin{align*}
\x_{t+1} = sign(\xh_t)\left[\left|\x_t - U^{\top}U(\x_t - \x_*)\right| - \lambda_t\right]_+
\end{align*}
By Proposition~\ref{thm:sabc}, we know that there are at most $s$ entries in $\left|\left[\x_t - U^{\top}U(\x_t - \x_*)\right]_{\Sb_*}\right|$ whose absolute values are larger than $(\delta_s + \sqrt{2}\theta_{s,s})\|\x_t-\x_*\|_2/\sqrt{s}$, therefore $[\x_{t+1}]_{\Sb_*}$ has at most $s$ non-zero entries by setting the value of $\lambda_t\geq \frac{\delta_s + \sqrt{2}\theta_{s,s}}{\sqrt{s}}\|\x_t-\x_*\|_2$. It concludes that $|\S_{t+1}\setminus \S_*|\leq s$ and $|\S_*\cup \S_t\cup \S_{t+1}|\leq 3s$. 
\end{proof}

\begin{prop} \label{thm:induction}
{Assume the noiseless setting $\e=0$.}   Let $\S_t$ be the support set of $\x_t$ and $\S_*$ be the support set of $\x_*$. If $|\S_t\setminus S_*|\leq s$, $\|\x_t - \x_*\|_2\leq \Delta_t$, and $\displaystyle\lambda_t = \frac{\delta_s + \sqrt{2}\theta_{s,s}}{\sqrt{s}}\Delta_t$, 
 Then we have
\[
    \|\x_{t+1} - \x_*\|_2\leq  (\delta_s + \theta_{s,s} + \delta_{3s})\Delta_t
\]
\end{prop}
To prove above proposition, we need the following Lemma, whose proof is deferred to the appendix. 
\begin{lemma}\label{lem:fund}
Let $\x$ by any $s$-sparse vector and $\x_{t+1}$ given in~(\ref{eqn:proxx}), we have
\begin{align*}
\|\x_{t+1} - \x\|^2_2\leq& \lambda_t\sqrt{s}\|\x_{t+1} - \x\|_2 + |(\x_{t+1} - \x)^{\top}(U^{\top}(U\x_t-\y) -( \x_t - \x))|
\end{align*}
\end{lemma} 

\begin{proof}[of Proposition~\ref{thm:induction}]
Let $\mathcal T=\S_*\cup \S_t\cup \S_{t+1}$, by Corollary~\ref{cor:xt}, we have $|\mathcal T|\leq 3s$. Indicated by the RIP condition, $\|U^{\top}_{\mathcal T}U_{\mathcal T}- I\|_2\leq \delta_{3s}$.  First, since $\y=U\x_*$ we have
\begin{align*}
&(\x_{t+1} - \x_*)^{\top}\left(U^{\top}\left(U\x_t - \y\right) - (\x_t - \x_*)\right) =  (\x_{t+1} - \x_*)^{\top}(U^{\top}U-I) (\x_t - \x_*)
\end{align*} 
Due to RIP of $U$ and $|\S_*\cup \S_t\cup \S_{t+1}|\leq 3s$, we have
\begin{align*}
 &|(\x_{t+1} - \x_*)^{\top}(U^{\top}U-I)(\x_t -\x_*)|\leq \delta_{3s}\|\x_{t+1} - \x_*\|_2\|\x_t - \x_*\|_2
\end{align*}
Thus by applying Lemma~\ref{lem:fund} with $\x=\x_*$,  we have
\begin{align*}
\|\x_{t+1} - \x_*\|_2^2& \leq  \lambda_t \sqrt{s}\|\x_{t+1} - \x_*\|_2 + \delta_{3s}\|\x_{t+1}- \x_*\|_2\|\x_t-\x_*\|_2 
\end{align*}
Then we get
\begin{align*}
\|\x_{t+1}-\x_*\|_2\leq \lambda_t\sqrt{s} + \delta_{3s} \|\x_t - \x_*\|_2
\end{align*}
Assuming $\|\x_{t} - \x_*\|_2\leq \Delta_t$ and plugging the value of $\lambda_t$, we have
\begin{align*}
\|\x_{t+1} -\x_*\|_2\leq (\delta_s + \sqrt{2}\theta_{s,s} + \delta_{3s})\Delta_t
\end{align*}
\end{proof}

\begin{proof}[Proof of Theorem~\ref{thm:main1}]
We aim to prove $\|\x_{t+1} - \x_*\|_2\leq \gamma^t\Delta_1$ and $|\S_{t+1}\setminus \S_*|\leq s$ by induction. This is true when $t=0$ due to  the initialization and the assumption $\|\x_1 - \x_*\|_2\leq \Delta_1$. Next, assume we have $\|\x_{t} - \x_*\|_2\leq \gamma^{t-1}\Delta_1$ and $|\S_{t}\setminus \S_*|\leq s$ for any $t\geq 1$. We prove that it also holds for $t+1$. By the definition of $\Delta_t$, we have $\Delta_t = \gamma^{t-1}\Delta_1$. Thus $\|\x_t - \x_*\|_2\leq \Delta_t$. By the value of $\lambda_t$, we have $\lambda_t = \frac{\delta_s + \sqrt{2}\theta_{s,s}}{\sqrt{s}}\Delta_t\geq \frac{\delta_s + \sqrt{2}\theta_{s,s}}{\sqrt{s}}\|\x_t - \x_*\|_2$. Hence, the condition in Corollary~\ref{cor:xt} hold, and as a result $|\S_{t+1}\setminus \S_*|\leq s$. From Proposition~\ref{thm:induction}, we also have $\|\x_{t+1} - \x_*\|_2\leq (\delta_s + \theta_{s,s} + \delta_{3s})\Delta_t = \gamma\Delta_t = \gamma^t\Delta_1$. 
\end{proof}

\subsection{Proof of Theorem~\ref{thm:main2}}
The logic for proving Theorem~\ref{thm:main2} is similar to proving Theorem~\ref{thm:main1}. 
\begin{cor}\label{cor:xt2}
Let $S_t$ be the support set of $\x_t$ and $\S_*$ be the support set of $\x_*$. If $|\S_t\setminus S_*|\leq s$ and $\lambda_t\geq \|U^{\top}\e\|_\infty + \frac{\delta_s + \sqrt{2}\theta_{s,s}}{\sqrt{s}}\|\x_t-\x_*\|_2$, then  $|\S_{t+1}\setminus \S_*|\leq s$ and $|\S_*\cup \S_t\cup \S_{t+1}|\leq 3s$. 
\end{cor}
\begin{proof}
The $\x_{t+1}$ is given by 
\begin{align*}
\x_{t+1} = sign(\xh_t)\left[\left|\x_t - U^{\top}(U\x_t- \y)\right| - \lambda_t\right]_+
\end{align*}
Due to $\y = U\x_* + \e$, we have
\begin{align*}
\x_t - U^{\top}(U\x_t - \y) = \x_t - U^{\top}U(\x_t - \x_*) + U^{\top}\e
\end{align*}
By Proposition~\ref{thm:sabc}, we know that there are at most $s$ entries in $\left[\x_t - U^{\top}U(\x_t - \x_*)\right]_{\Sb_*}$ with magnitude  larger than $\frac{\delta_s +\sqrt{2}\theta_{s,s}}{\sqrt{s}}\|\x_t-\x_*\|_2$.  As a result, $[\x_t - U^{\top}(U\x_t - \y)]_{\Sb_*}$ has at most $s$ entries whose magnitudes larger than $\|U^{\top}\e\|_\infty + \frac{\delta_s +\sqrt{2} \theta_{s,s}}{\sqrt{s}}\|\x_t-\x_*\|_2$. Therefore, given the assumed value of $\lambda_t$,  $[\x_{t+1}]_{\Sb_*}$ has at most $s$ entries larger than zero. 
It concludes that $|\S_{t+1}\setminus \S_*|\leq s$ and $|\S_*\cup \S_t\cup \S_{t+1}|\leq 3s$. 
\end{proof}

\begin{prop} \label{thm:induction2}
Let $S_t$ be the support set of $\x_t$ and $\S_*$ be the support set of $\x_*$. If  $|\S_t\setminus S_*|\leq s$,  $\|\x_t - \x_*\|_2\leq \Delta_t$ and  $\displaystyle\lambda_t = \|U^{\top}\e\|_\infty+ \frac{\delta_s+\sqrt{2}\theta_{s,s}}{\sqrt{s}}\Delta_t$, 
 then we have
\[
    \|\x_{t+1} - \x_*\|_2\leq  (\delta_s + \sqrt{2}\theta_{s,s} + \delta_{3s})\Delta_t + (1+\sqrt{2})\sqrt{s}\|U^{\top}\e\|_{\infty}
\]
\end{prop}

\begin{proof}
Since $\y=U\x_* + \e$, we have
\begin{align*}
&  (\x_{t+1} - \x_*)^{\top}U^{\top}\left(U\x_t - \y\right) - (\x_t - \x_*) 
 = (\x_{t+1} - \x_*)^{\top}(U^{\top}U-I) (\x_t - \x_*) - (\x_{t+1}-\x_*)^{\top}U^{\top}\e
\end{align*}
Due to the restricted isometry property, we have
\begin{align*}
 &|(\x_{t+1} - \x_*)^{\top}(U^{\top}U-I)(\x_t -\x_*)|\leq \delta_{3s}\|\x_{t+1} - \x_*\|_2\|\x_t - \x_*\|_2
\end{align*}
and by Cauchy-Shwartz inequality, we have
\[
|(\x_{t+1}-\x_*)^{\top}U^{\top}\e|\leq \sqrt{2s}\|U^{\top}\e\|_\infty\|\x_{t+1} - \x_*\|_2
\]
where we use the fact $|\S_{t+1}\setminus \S_*|\leq s$ due to Corollary~\ref{cor:xt2}. Thus by combining the two inequalities with  Lemma~\ref{lem:fund} with $\x=\x_*$, we have
\begin{align*}
&\|\x_{t+1} - \x_*\|_2^2 \leq  \lambda_t \sqrt{s}\|\x_{t+1} - \x_*\|_2 + \delta_{3s}\|\x_{t+1}- \x_*\|_2\|\x_t-\x_*\|_2 + \sqrt{2s}\|U^{\top}\e\|_\infty\|\x_{t+1}- \x_*\|_2
\end{align*}
Then we get
\begin{align*}
\|\x_{t+1}-\x_*\|_2\leq \lambda_t\sqrt{s} + \delta_{3s} \|\x_t - \x_*\|_2  + \sqrt{2s}\|U^{\top}\e\|_\infty
\end{align*}
Plugging the value of $\lambda_t$, we have
\begin{align*}
\|\x_{t+1} -\x_*\|_2\leq& (\delta_s + \sqrt{2}\theta_{s,s} + \delta_{3s})\|\x_t - \x_*\|_2 + (1+\sqrt{2})\sqrt{s}\|U^{\top}\e\|_\infty
\end{align*}
\end{proof}

\begin{proof}[Proof of Theorem~\ref{thm:main2}]
First, we assume $\|\x_t-\x_*\|_2\leq \Delta_t$, then  by Proposition~\ref{thm:induction2}, we have
\begin{align*}
\|\x_{t+1} - \x_*\|_2&\leq \gamma\Delta_t + (1+\sqrt{2})\sqrt{s}\|U^{\top}\e\|_\infty\triangleq\Delta_{t+1}
\end{align*}
Similarly, we can use Corollary~\ref{cor:xt2} to show that $|\S_{t+1}\setminus \S_*|\leq s$ given $|\S_t \setminus \S_*|\leq s$. 
Since $\S_1=\emptyset$ and $\|\x_1- \x_*\|\leq \Delta_1$, therefore by induction we can complete the proof. 

\end{proof}

\section{Nearly-Sparse Signal Recovery}\label{sec:nearly}
In this section, we present algorithms and analysis  for finding a sparse solution that approximates a nearly-sparse signal $\x_*$ with a small error. 

\subsection{Algorithms and Main Results}In order to derive a practical algorithm and a better recovery result, we instead assume that the random measurement matrix $U\in\R^{n\times d}$ contains sub-gaussian measurements, i.e., each element $U_{ij}$ is a sub-gaussian random variable and has mean zero and variance $1/n$.  The details of the algorithm is presented in Algorithm~\ref{alg:2}. The value of $\Delta_1$ and $\Lambda$ can be set according to our analysis.  In the sequel, we abuse the notation $\S_*$ to denote the support set of $\x^s_*$.  We first state the main theorem regarding the nearly-sparse signal recovery of Algorithm~\ref{alg:2}. 
\begin{thm}\label{thm:nearly-sparse} Let $\gamma = (1+\sqrt{2})\eta<1$.  For any $\tau>0$, assume
\[
\Lambda {\triangleq}  \sqrt{s}\|U^{\top}\e\|_\infty + cD(\x_*, \x^s_*), \quad n\geq \frac{c^2(\tau + s\log[d/s])}{\eta^2}.
\]
where 
\begin{align*}
D(\x_*, \x_*^s)\hspace*{-2pt}= \|(\x_* - \x^s_*)^s\|_2 + \sqrt{\frac{\tau + s\ln [d/s]}{n}}\|\x_* - \x^s_*\|_{2}
\end{align*}
where $c$ is some universal constant.  Let $\{\Delta_t, t=1,\ldots, T\}$ be a sequence such that $\Delta_1\geq \max\left(\|\x^s_*\|_2, \Lambda\right)$, and 
\[
\Delta_{t+1} = \gamma\Delta_t+(1+\sqrt{2})\Lambda.
\] With a probability $1 - 2te^{-\tau}$, we have for all $t\geq 0$\[
|\S_{t+1} \setminus \S_*| \leq s, \quad \|\x_{t+1} - \x^s_*\|_2 \leq \Delta_{t+1}\]
In particular, let $T_0$ be the smallest  value such that
\[
\gamma^{T_0}\Delta_1 \leq \frac{\Lambda}{1-\gamma}
\] 
We run Algorithm~\ref{alg:2} with $T_0$ iterations and denote by $\bar\x$ the output solution.  With a probability $1 - 2T_0e^{-\tau}$, we have
\begin{align}\label{eqn:err}
    \|\bar\x - \x^s_*\|_2 \leq \frac{\sqrt{2}(1+\sqrt{2})}{1-\gamma}\Lambda.
\end{align}
\end{thm}
\begin{algorithm}[t]
\caption{Homotopy Proximal Mapping for learning a Sparse Solution (HPM1)}
{\bf Input:} initial size $\Delta_1\geq\max(\|\x^s_*\|_2, \Lambda)$,  the target sparsity $s$, a random measurement matrix $U \in \R^{d\times n}$ and measurements $\y \in \R^n$, and $\eta<\sqrt{2}-1$

\begin{algorithmic}[1]
\STATE Initialize $\x_1 = 0$, $\gamma=(1+\sqrt{2})\eta$
\FOR{$t=1, 2, \ldots, T$}
   \STATE $\lambda_t =(\Lambda + \eta\Delta_t)/\sqrt{s}$
    \STATE $\xh_{t+1} = \x_t - U^{\top}(U\x_t - \y)$
    \STATE $\x_{t+1} = \sgn(\xh_{t+1})\left[\xh_{t+1} - \lambda_t \right]_+$
    \STATE $\Delta_{t+1} = \gamma\Delta_t + (1+\sqrt{2})\Lambda$
\ENDFOR
\end{algorithmic} \label{alg:2}
{\bf Return} $\x_{T+1}$
\end{algorithm}

{\bf Remark 6:} we note that the final solution returned by Algorithm 2 is at most $2s$-sparse. We can also take the $s$-largest element in $\bar\x$ to form a $s$-sparse approximation. The Proposition~\ref{lem:2sparse} in the appendix  guarantees that the error $\|\bar\x^s - \x^s_*\|_2$ is only amplified by a constant factor of $\sqrt{3}$. 

{\bf Remark 7:}  It can be seen that when $\x_* = \x_*^s$, i.e., the signal is sparse, the problem boils down to  sparse signal recovery with noisy observations and the result in Theorem~\ref{thm:nearly-sparse} is similar to Theorem~\ref{thm:main2} except that the RIP constants are replaced with a quantity dependent on $n$ since we directly bound RIP constants of a sub-gaussian matrix. Further, when $\e=0$, then we can set $\Lambda=0$ in Algorithm 2 and the result in Theorem~\ref{thm:nearly-sparse} is similar to that in Theorem~\ref{thm:main1} for sparse signal recovery under noiseless observations.

{\bf Remark 8:}  The result in Theorem~\ref{thm:nearly-sparse} also implies that more observations (i.e., larger $n$) may lead to more accurate recovery and fast convergence. And also we note that the key property of the measurement matrix $U$ is that it satisfies the JL lemma with a high probability. Therefore, any JL transforms can be used, including sparse JL transform based on random hashing~\citep{Dasgupta:2010:SJL,Kane:2014:SJT:2578041.2559902}, which can speed up the computation. 

\begin{algorithm}[t]
\caption{Homotopy Proximal Mapping for recovering a sparse solution (HPM2)}
{\bf Input:} the target sparsity $s$, a random measurement matrix $U \in \R^{d\times n}$ and measurements $\y \in \R^n$ and $\eta$, and the total number of iterations $T$. 

\begin{algorithmic}[1]
\STATE Initialize $\x_1 = 0$, $\gamma=2(1+\sqrt{2})\eta$, and $\lambda_1 = 2\eta\Delta_1/\sqrt{s}$
\FOR{$t=1, 2, \ldots, T$}
    \STATE $\xh_{t+1} = \x_t - U^{\top}(U\x_t - \y)$
    \STATE $\x_{t+1} = \sgn(\xh_{t+1})\left[\xh_{t+1} - \lambda_t \right]_+$
    \STATE $\lambda_{t+1} = \gamma\lambda_t$
    \IF{$\|\x_{t+1}\|_0>2s$}
        \STATE Set $\xh = \x_{t}$
        \STATE Break
    \ENDIF
\ENDFOR
\end{algorithmic} \label{alg:3}
{\bf Return} $\xh$
\end{algorithm}
One issue of Algorithm 2 is that it needs to estimate $\|U^{\top}\e\|_\infty$ and $cD(\x_*, \x_*^s)$ for setting $\lambda_t$ and for stopping the algorithm, which could be difficult  in many circumstances. In addition, an overestimated $\Lambda$ could increase the number of iterations and the recovery error. To alleviate this issue,  below we present a more practical algorithm for nearly sparse signal recovery which could perform better in absence of prior knowledge. The key idea is motivated by Theorem~\ref{thm:nearly-sparse}. At earlier stage of Algorithm 2, we would expect that $\Lambda\leq O(\Delta_t)$ and therefore we can absorb $\Lambda$ into $\Delta_t$ for setting $\lambda_t$. And for stopping the algorithm we note that as long as $|\S_{t+1}|\leq 2s$, we can have the recovery error bounded by $\Delta_{t+1}$ (Theorem~\ref{thm:5}) or $O(\Lambda)$ (Theorem~\ref{thm:6}), therefore we stop the algorithm when $|\S_{t+1}|>2s$. The detailed steps of the practical algorithm are presented in Algorithm 3. The recovery error of Algorithm 3 is provided by the following theorem. 
\begin{thm} \label{thm:11}
Let $\Delta_1 \geq \|\x_*^s\|_2$ be a constant. Let $\xh$ be the solution output from Algorithm~\ref{alg:3} and $T$ is the maximum number of iteration allowed.  Assume 
\[
c\sqrt{\frac{\tau+ s\log(d/s)}{n}} \leq \eta\leq \frac{1}{2(1+\sqrt{3})}
\]
Then, with a probability at least $1 - 6Te^{-\tau}$, we have
\[
    \|\xh - \x_*^{s}\|_2 \leq \max\left(\frac{\Lambda}{\eta},  \gamma^T\Delta_1  \right)
\]
where $\gamma = 2(1+\sqrt{2})\eta <1$, $\Lambda=\sqrt{s}\|U^{\top}\e\|_\infty + cD(\x_*, \x^s_*)$, $D(\x_*, \x^s_*)$ is defined in Theorem~\ref{thm:nearly-sparse} and $c$ is some universal constant.
\end{thm}

{\bf Remark 9: } Although in Algorithm~\ref{alg:3} we still use an estimate $\Delta_1\geq \|\x_*^s\|_2$ for setting the initial value of $\lambda$, in practice we can set it to a sufficiently large value (e.g., $\|U^{\top}\y\|_\infty$) such that $\x_{2}=0$. 

{\bf Remark 10:} The universal constant $c$ in Theorem~\ref{thm:11} should not be treated literally the same as in Theorem~\ref{thm:nearly-sparse}. In numerical simulations, we observe that Algorithm~\ref{alg:3} is more robust to smaller values of $\eta$ than Algorithm~\ref{alg:2}.

{\bf Remark 11:} Theorem~\ref{thm:11} reveals a tradeoff in setting the value of $\eta$. A smaller value of $\eta$ will lead to faster convergence but larger recovery error.

\subsection{Proof of Theorem~\ref{thm:nearly-sparse}}
We first give the following lemma. 
\begin{lemma}\label{lem:fund2}
Assume  $U\in\R^{n\times d}$ is a sub-gaussian measurement matrix, where each element in $U$ has zero mean and variance $1/n$. If $|\S_t\setminus\S_*|\leq s$, then with a probability $1 - 2e^{-\tau}$, we have
\begin{align*}
&\left\|(U^{\top}\left(U\x_t - \y\right) - (\x_t - \x^s_*))^s\right\|_2\leq  \sqrt{s}\|U^{\top}\e\|_\infty+ cD(\x_*, \x_*^s)+ c\sqrt{\frac{\tau + s\log[d/s]}{n}} \|\x_t - \x^s_*\|_2
\end{align*}
where $D(\x_*, \x_*^s)$ is defined in Theorem~\ref{thm:nearly-sparse}
and $c$ is some universal constant. 
\end{lemma}
Lemma~\ref{lem:fund2} is proved in the appendix. Following  Lemma~\ref{lem:fund2}, we prove the following Corollary. 
\begin{cor} \label{cor:11}
Let $\S_t$ and $\S_{t+1}$ be the support sets of $\x_t$ and $\x_{t+1}$, respectively. If $|\S_t \setminus \S_*| \leq s$, then with a probability $1 - 2e^{-\tau}$, we have
\[
    |\S_{t+1} \setminus \S_*| \leq s
\]
provided that
\begin{eqnarray}\label{eqn:lamdat}
\lambda_t\hspace*{-2pt} \geq\hspace*{-2pt} \|U^{\top}\e\|_\infty\hspace*{-2pt} +\hspace*{-2pt} \frac{cD(\x_*, \x^s_*)}{\sqrt{s}}\hspace*{-2pt}+\hspace*{-2pt} \frac{c}{\sqrt{s}}\sqrt{\frac{\tau + s\log[d/s]}{n}} \|\x_t - \x^s_*\|_2
\end{eqnarray}
\end{cor}
\begin{proof}
{The proof is similar to that of Corollary~\ref{cor:xt2}. Recall that $\x^s$ denotes the vector $\x$ with all but the $s$-largest entries (in magnitude) set to zero, and $\S_*$ denotes the support of the $s$-largest entries in $\x_*$. 
From Lemma~\ref{lem:fund2}, we can  conclude that $[\x_t - U^{\top}\left(U\x_t - \y\right)]_{\Sb_*}$ has at most $s$ entries with magnitude larger than the quantity in the R.H.S of~(\ref{eqn:lamdat}). This can be verified by contradiction. If there exists $\A\subseteq\Sb_*$  such that $|\A|>s$ and for all $i\in\A$
\[
[\x_t - U^{\top}\left(U\x_t - \y\right)]_{i}\geq  \|U^{\top}\e\|_\infty\hspace*{-2pt} +\hspace*{-2pt} \frac{cD(\x_*, \x^s_*)}{\sqrt{s}}\hspace*{-2pt}+\hspace*{-2pt} \frac{c}{\sqrt{s}}\sqrt{\frac{\tau + s\log[d/s]}{n}} \|\x_t - \x^s_*\|_2
\]
Let $\A_s\subseteq\A$ such that $|\A_s|=s$.  Then 
\[
\|[\x_t - \x_*^s - U^{\top}\left(U\x_t - \y\right)]_{\A_s}\|_2\geq  \sqrt{s}\|U^{\top}\e\|_\infty+ cD(\x_*, \x_*^s)+ c\sqrt{\frac{\tau + s\log[d/s]}{n}} \|\x_t - \x^s_*\|_2
\]
where we use the fact $[\x_*^s]_{\A_s}=0$. However, the above inequality contradicts to Lemma~\ref{lem:fund2}.  Since $\x_{t+1}$ is given by 
\begin{align*}
\x_{t+1} = sign(\xh_t)\left[\left|\x_t - U^{\top}(U\x_t- \y)\right| - \lambda_t\right]_+
\end{align*}
 Therefore, given the assumed value of $\lambda_t$,  $[\x_{t+1}]_{\Sb_*}$ has at most $s$ entries larger than zero. It concludes that $|\S_{t+1}\setminus \S_*|\leq s$. }
\end{proof}

Based on the above corollary, we can prove the following proposition that serves the key to prove the main theorem. 
\begin{prop}\label{thm:13}
Assume $|\S_t \setminus \S_*| \leq s$, $\|\x_t - \x^s_*\|_2 \leq \Delta_t$, and define
\begin{eqnarray}
\Lambda\triangleq \sqrt{s}\|U^{\top}\e\|_\infty + cD(\x_*, \x^s_*) \label{eqn:Lambda}
\end{eqnarray}
Let $\displaystyle \lambda_t = \frac{\Lambda + \eta\Delta_t}{\sqrt{s}}$. Then, with a probability $1 - 2e^{-\tau}$, we have
\[
|\S_{t+1} \setminus \S_*| \leq s, \quad \|\x_{t+1} - \x^s_*\|_2 \leq\Delta_{t+1}\triangleq (1+\sqrt{2})\eta\Delta_t +(1+\sqrt{2})\Lambda
\]
provided
\[
c\sqrt{\frac{\tau + s\log[d/s]}{n}} \leq \eta
\]
\end{prop}
\begin{proof}
It is easy to verify that the condition for $\lambda_t$ in Corollary~\ref{cor:11} is satisfied. Combining  with the fact that $\x_t$ is $2s$-sparse vector, we have $|\S_{t+1} \setminus \S_*| \leq s$ due to Corollary~\ref{cor:11}. Applying Lemma~\ref{lem:fund} with $\x=\x_*^s$, we have 
\begin{align}\label{eqn:22}
&\|\x_{t+1} - \x^s_*\|_2^2 \leq \lambda_t\sqrt{s}\left\|\x_{t+1} - \x^s_*\right\|_2 + |(\x_{t+1} - \x^s_*)^{\top}(U^{\top}\left(U\x_t - \y\right) - (\x_t -\x_*^s))|
\end{align}
According to Lemma~\ref{lem:fund2}, with a probability $1 - 2e^{-\tau}$, we have
\begin{align}\label{eqn:lU}
\left\|\left(U^{\top}\left(U\x_t - \y\right) - (\x_t - \x^s_*)\right)^s\right\|_2 \leq \Lambda + \eta \Delta_t.
\end{align}
Thus
\begin{align*}
& | (\x_{t+1} - \x^s_*)^{\top}\left(U\left(U^{\top}\x_t - \y\right) - (\x_t - \x^s_*)\right)| \\
& \leq \left| [\x_{t+1} - \x^s_*]_{\S_*}^{\top}\left[U\left(U^{\top}\x_t - \y\right) - (\x_t - \x^s_*)\right]_{\S_*}\right| \\
& + \left| [\x_{t+1} - \x^s_*]_{\S_{t+1}\setminus\S_*}^{\top}\left[U\left(U^{\top}\x_t - \y\right) - (\x_t - \x^s_*)\right]_{\S_{t+1}\setminus\S_*}\right| \\
 &\leq \left(\left\|[\x_{t+1} - \x^s_*]_{\S_*}\right\|_2 + \left\|[\x_{t+1} - \x^s_*]_{\S_{t+1}\setminus\S_*}\right\|_2\right)(\Lambda +\eta\Delta_t) \\
& \leq  \sqrt{2}(\Lambda+\eta\Delta_t)\|\x_{t+1} - \x^s_*\|_2
\end{align*}
where we use the fact that $|\S_*|\leq s$ and $|\S_{t+1}\setminus\S_*|\leq s$ and inequality in~(\ref{eqn:lU}), and the last inequality uses the fact $a + b \leq \sqrt{2(a^2 + b^2)}$.  Combining the above inequality with~(\ref{eqn:22}), with a probability $1 - 2e^{-\tau}$, we have
\begin{align*}
&\|\x_{t+1} - \x^s_*\|_2^2 \leq \left(\lambda_t\sqrt{s} + \sqrt{2}\eta\Delta_t+\sqrt{2}\Lambda\right)\|\x_{t+1} - \x^s_*\|_2 \\
&\leq[(1+\sqrt{2})\eta\Delta_t+(1+\sqrt{2})\Lambda] \|\x_{t+1} - \x^s_*\|_2,
\end{align*}
Therefore 
\[
\|\x_{t+1} - \x^s_*\|_2\leq (1+\sqrt{2})\eta\Delta_t +(1+\sqrt{2}) \Lambda
\]
\end{proof}
\begin{proof}[Proof of Theorem~\ref{thm:nearly-sparse}]
Following Proposition~\ref{thm:13} and by induction, we can prove for any $t$, we have with a probability $1-2te^{-\tau}$
\begin{align*}
&|\S_{t+1}\setminus\S_*|\leq s\\
&\|\x_{t+1} - \x_*^s\|_2\leq \Delta_{t+1}
\end{align*}
Since $\Delta_{t+1} = \gamma\Delta_t + (1+\sqrt{2})\Lambda$, we have
\begin{align*}
\Delta_{t+1}&\leq \gamma^t\Delta_1 + \frac{1-\gamma^t}{1-\gamma}(1+\sqrt{2})\Lambda\leq\gamma^t\Delta_1 + \frac{1}{1-\gamma}(1+\sqrt{2})\Lambda
\end{align*}
Let $t=T_0$ such that $\gamma^{T_0}\Delta_1\leq \Lambda/(1-\gamma)$, we then have
\[
\|\x_{T_0+1}-\x^s_*\|_2\leq \frac{\sqrt{2}(1+\sqrt{2})\Lambda}{1-\gamma}
\]
with a probability $1-2T_0e^{-\tau}$, which completes the proof of Theorem~\ref{thm:nearly-sparse}. 
\end{proof}
\subsection{Proof of Theorem~\ref{thm:11}}
We first state two theorems that are central to our analysis. Theorem~\ref{thm:5} reveals that the recovery error of Algorithm~\ref{alg:3} will decrease by a constant factor at the beginning, and Theorem~\ref{thm:6} shows that the recovery error will keep small in the later stage. We abuse the notation $\Delta_t$.  
\begin{thm} \label{thm:5}
Let $\Delta_1 \geq \|\x_*^s\|_2$ be a constant, $\gamma=2(1+\sqrt{2})\eta$, and $\{\x_{t}, t=1,\ldots\}$ be solutions generated by Algorithm~\ref{alg:3}. 
Assume $|\S_t \setminus \S_*| \leq s$, $\|\x_t - \x_*^{s}\| \leq \Delta_t$, and $\Lambda \leq \eta\Delta_t $. Then, with a probability at least $1 - 2e^{-\tau}$, we have
\[
|\S_{t+1} \setminus \S_*| \leq s \textrm{ and } \|\x_{t+1} - \x_*^{s}\|_2 \leq \Delta_{t+1}{ \triangleq} \gamma\Delta_{t}
\]
provided the condition in Theorem~\ref{thm:11} is true.
\end{thm}
\begin{thm} \label{thm:6}
Let  $\{\x_{t}, t=1,\ldots\}$ be solutions generated by Algorithm~\ref{alg:3}. Assume $|\S_t | \leq 2s$, $\|\x_t - \x_*^{s}\| \leq  \Lambda/\eta $, and $\Lambda > \eta \Delta_t$. If $|\S_{t+1}| \leq 2s$, then with a probability at least  $1 - 2e^{-\tau}$, we have
\[
\quad \|\x_{t+1} - \x_*^{s}\|_2 \leq  2 (1 + \sqrt{3}) \Lambda\leq \Lambda/\eta
\]
provided the condition in Theorem~\ref{thm:11} is true.
\end{thm}
\begin{proof}[Proof of Theorem~\ref{thm:11}]
Let
\[
k=\min \left\{t:\Lambda > \eta\Delta_t \right\}.
\]
Assume $k\geq 1$, otherwise Theorem~\ref{thm:11} holds with $T=1$. 
In the following, we consider two cases $T<k$ and $T\geq k$, { where $T$ is the input to the algorithm}. 
\paragraph{$T < k$.}  Since the condition $\Lambda \leq \eta\Delta_t$ holds for $t=1,\ldots,T$, we can apply Theorem~\ref{thm:5} to bound the recovery error in each iteration. Thus, with a probability at least $1 - 2Te^{-\tau}$, we have
 \[
 \|\xh - \x_*^{s}\|_2 = \|\x_{T+1}-\x_*^{s}\|_2 \leq  \Delta_{T+1} = \gamma^T\Delta_1.
 \]

\paragraph{$T \geq k$.} From the above analysis, with a probability at least $1 - 2(k-1)e^{-\tau}$, we have $\|\x_{k}-\x_*^{s}\|_2  \leq \Delta_k$ and $|\S_k \setminus \S_*| \leq s$, which also means our algorithm arrives at the $k$-th iteration. In the $k$-th iteration, there will be two cases: $|\S_{k+1}| > 2s$ and $|\S_{k+1}| \leq 2s$. For the first case, our algorithm terminates, and return $\x_k$ as the final solution implying
\[
\|\xh - \x_*^{s}\|_2 = \|\x_{k}-\x_*^{s}\|_2   \leq \Delta_k \leq \Lambda/\eta.
\]
For the second case, Algorithm~\ref{alg:3} keeps running, and we can bound the recovery error  by Theorem~\ref{thm:6}. In particular, if at $T'\geq k$ $|\S_{T'+1}| > 2s$, our algorithm terminates and returns $\x_{T'}$ as the final solution, which implies $|\S_{t}|<2s, t\leq T'$. By applying induction of Theorem~\ref{thm:6} from $t=k$, we can have
\[
\|\xh - \x_*^{s}\|_2 = \|\x_{T'}-\x_*^{s}\|_2   \leq  2 (1 + \sqrt{3}) \Lambda\leq \Lambda/\eta.
\]
\end{proof}
\begin{proof}[Proof of Theorem~\ref{thm:5}]
The proof is very similar to that of Theorem~\ref{thm:13} by noting that $\lambda_t =2\eta\Delta_t/\sqrt{s}> (\Lambda + \eta\Delta_t)/\sqrt{s}$ satisfies the condition in Corollary~\ref{cor:11}. 
\end{proof}
\begin{proof}[Proof of Theorem~\ref{thm:6}]
First we note that $\x_t - \x_*^s$ is at most $3s$-sparse (at most $3s$ elements in $\x_t - \x_*$ are non-zeros). With a slight change of the universal constant, we still have Lemma~\ref{lem:fund2} (c.f. the proof of Lemma~\ref{lem:fund2} in the appendix). Then, with a probability at least $1 - 2e^{-\tau}$, we have
\[
\begin{split}
 & \left\|\left[U\left(U^{\top}\x_t - \y\right) - (\x_t - \x_*^{s})\right]^{s}\right\|_2 \leq   \Lambda+ c\sqrt{\frac{\tau+ s\log(d/s)}{m}} \|\x_t - \x_*^{s}\|_2\leq  \Lambda +\eta \|\x_t - \x_*^{s}\|_2 \leq 2  \Lambda.
\end{split}
\]
Notice that $\x_{t+1} - \x_*^{s}$ is $3s$-sparse in this case, and we can verify that
\[
\begin{split}
 & \left|(\x_{t+1} - \x_*^{s})^{\top}\left(U\left(U^{\top}\x_t - \y\right) - (\x_t - \x_*^{s})\right) \right| \leq  2 \sqrt{3} \Lambda \|\x_{t+1} - \x_*^{s}\|_2.
\end{split}
\]
{To see this, we can split $\x_{t+1} - \x_*^s = \a + \b + \c$ into three components each with at most $s$ non-zero entries and non-overlapping support. Then 
\begin{align*}
 &\left|(\x_{t+1} - \x_*^{s})^{\top}\left(U\left(U^{\top}\x_t - \y\right) - (\x_t - \x_*^{s})\right) \right| \\
 &\leq  (\|\a\|_2  + \|\b\|_2 + \|\c\|_2)\left\|\left[U\left(U^{\top}\x_t - \y\right) - (\x_t - \x_*^{s})\right]^{s}\right\|_2\\
 &\leq  (\|\a\|_2  + \|\b\|_2 + \|\c\|_2)2\Lambda\leq 2\sqrt{3}\Lambda \|\x_{t+1} - \x_*^s\|_2
\end{align*}
where we use the fact $a + b + c\leq \sqrt{3(a^2+b^2+c^2)}$. } Applying Lemma~\ref{lem:fund} with $\x=\x^s_*$, we have, with a probability at least $1 - 2e^{-\tau}$,
{\begin{align*}
\|\x_{t+1} - \x^s_*\|^2_2&\leq \lambda_t\sqrt{s}\|\x_{t+1} - \x^s_*\|_2 + |(\x_{t+1} - \x^s_*)^{\top}(U^{\top}(U\x_t-\y) -( \x_t - \x^s_*))|\\
&\leq \left(\lambda_t \sqrt{s} + 2 \sqrt{3} \Lambda  \right)\|\x_{t+1} - \x_*^{s}\|_2 \leq  2(1 + \sqrt{3})  \Lambda  \|\x_{t+1} - \x_*^{s}\|_2,
\end{align*}
where we use the fact $\eta = \frac{2\eta\Delta_t}{\sqrt{s}}\leq \frac{2\Lambda}{\sqrt{t}}$ due to the assumption $\Lambda>\eta\Delta_t$. }
Thus
\[
\|\x_{t+1} - \x_*^{s}\|_2 \leq 2(1 + \sqrt{3}) \Lambda .
\]
\end{proof}


\section{Numerical Simulations}
In this section, we conduct numerical simulations to verify the proposed algorithms and the developed analysis, and also compare with previous algorithms. 

\paragraph{Verifying Theorem~\ref{thm:nearly-sparse}}
We first conduct some empirical studies to verify the results in Theorem~\ref{thm:nearly-sparse}. We generate a measurement  matrix $U\in \R^{n\times d}$ such that each element follows an i.i.d distribution $\mathcal N(0, 1/n)$. To generate a $s$-sparse target signal, we sample the non-zeros elements from standard normal distribution followed by $\ell_2$ norm normalization. To generate a nearly sparse target signal, we set the $i$-th element of $\x_*$ to $i^{-1}$ followed by $\ell_2$ norm normalization. The noise vector $\e$ is drawn from uniform distribution $[-\sigma, \sigma]$. We run Algorithm~\ref{alg:2} with hundreds of iterations and plot error/sparsity versus the number of iterations  in Figure~\ref{fig:1}, Figure~\ref{fig:2} and Figure~\ref{fig:3} for $n=2000, d=10000, s=20, \sigma=0.001$ under three different settings, respectively.  The value of  $\Delta_1$ is set to $\|\x^s_*\|_2$, the value of $\eta$ is set to $0.4, 0.4, 0.3$ for three different settings, respectively, and the value of $\Lambda$ is set to $0, \sqrt{s}\|U^{\top}\e\|_\infty$ and $\sqrt{s}\|U^{\top}\e\|_\infty + \eta\|\x_* - \x^s_*\|_2$, respectively. These results clearly validate Theorem~\ref{thm:nearly-sparse}.

\begin{figure}[t]
\centering 
\subfigure[error]{\includegraphics[scale=0.26]{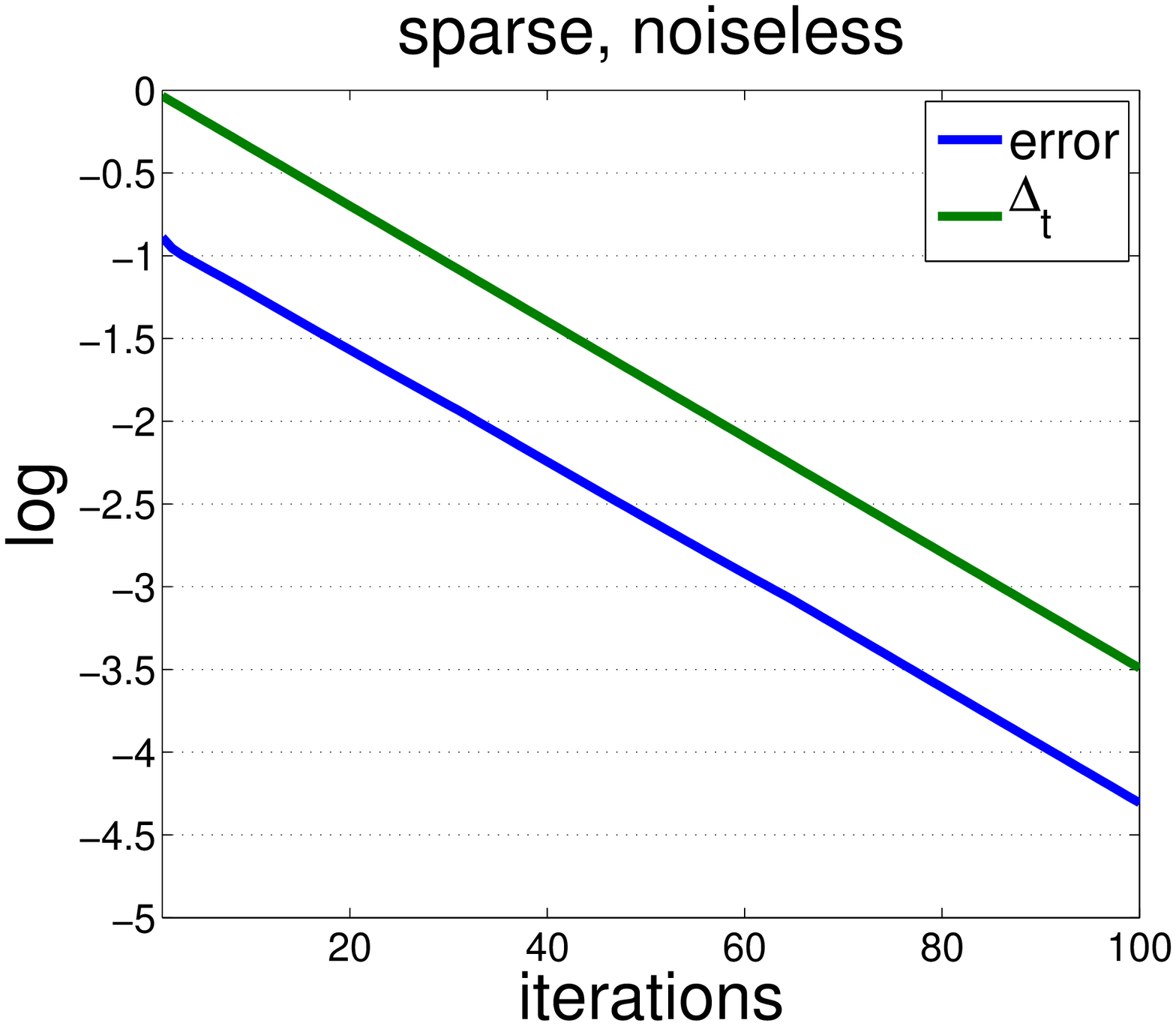}}
\subfigure[sparsity]{\includegraphics[scale=0.26]{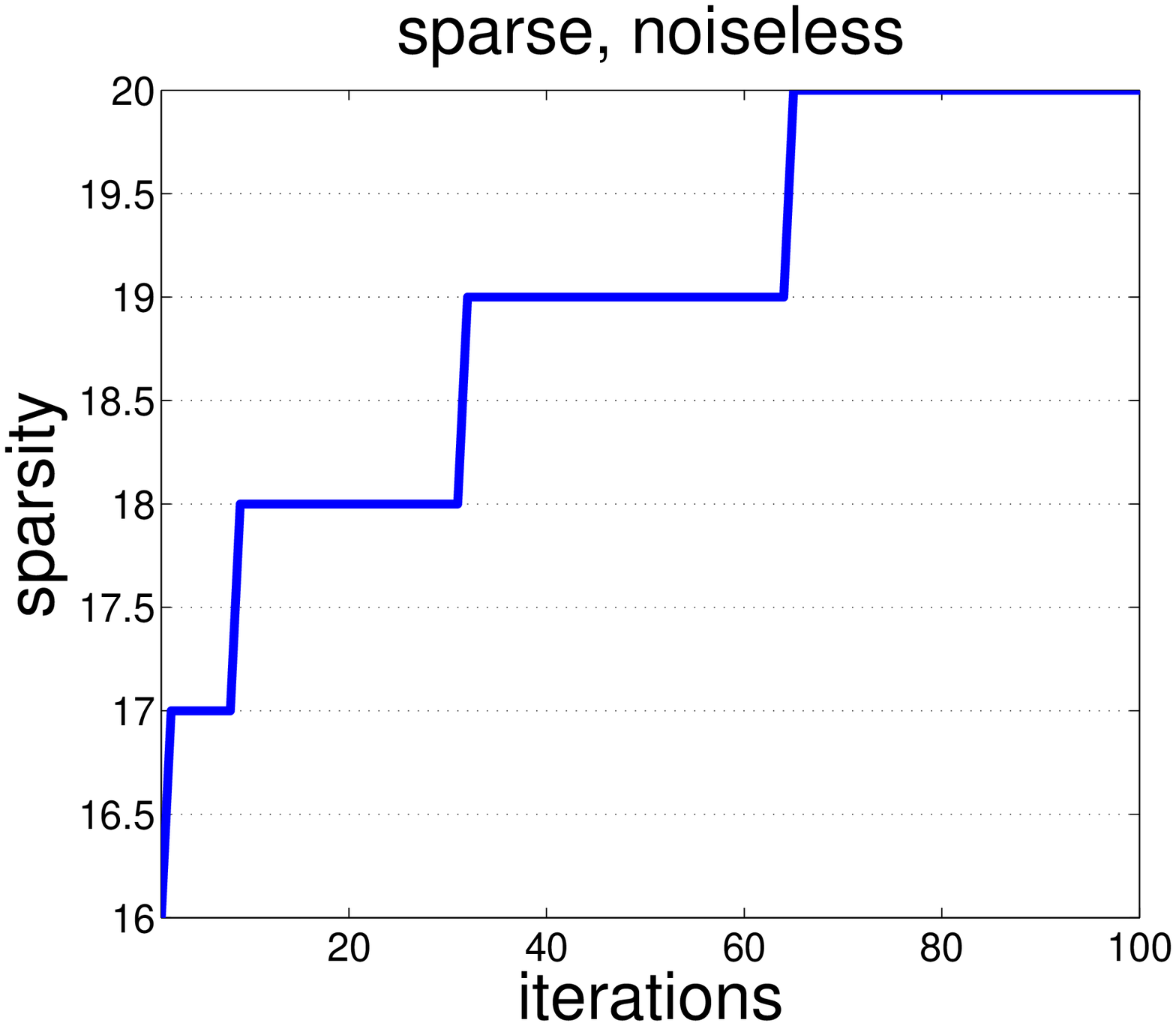}}
\caption{Recovery error and sparsity versus iterations in setting I. }\label{fig:1}
\subfigure[error]{\includegraphics[scale=0.26]{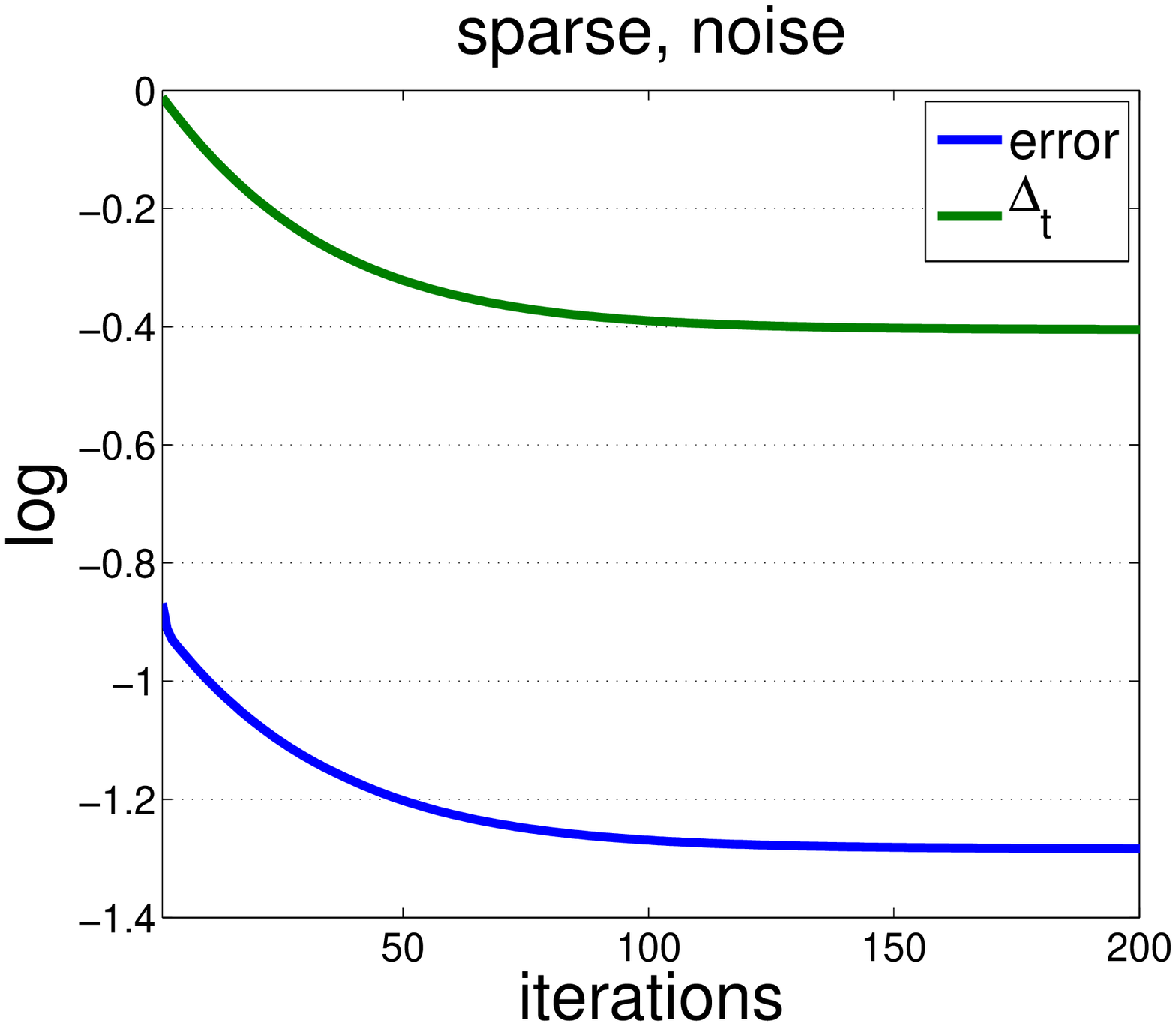}}\hspace*{-0.1in}
\subfigure[sparsity]{\includegraphics[scale=0.26]{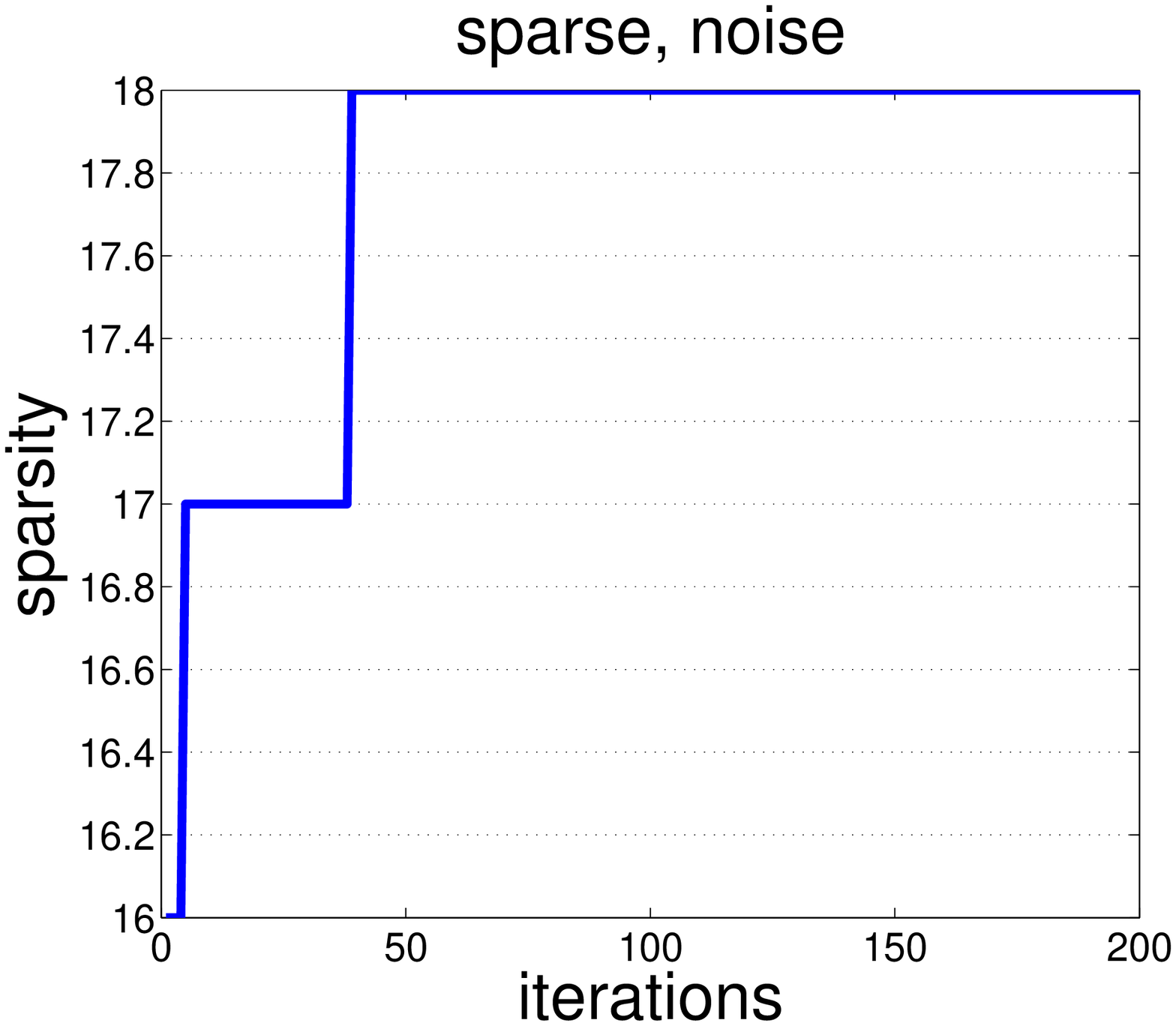}}
\caption{Recovery error and sparsity versus iterations in setting II. }\label{fig:2}
\centering 
\subfigure[error]{\includegraphics[scale=0.26]{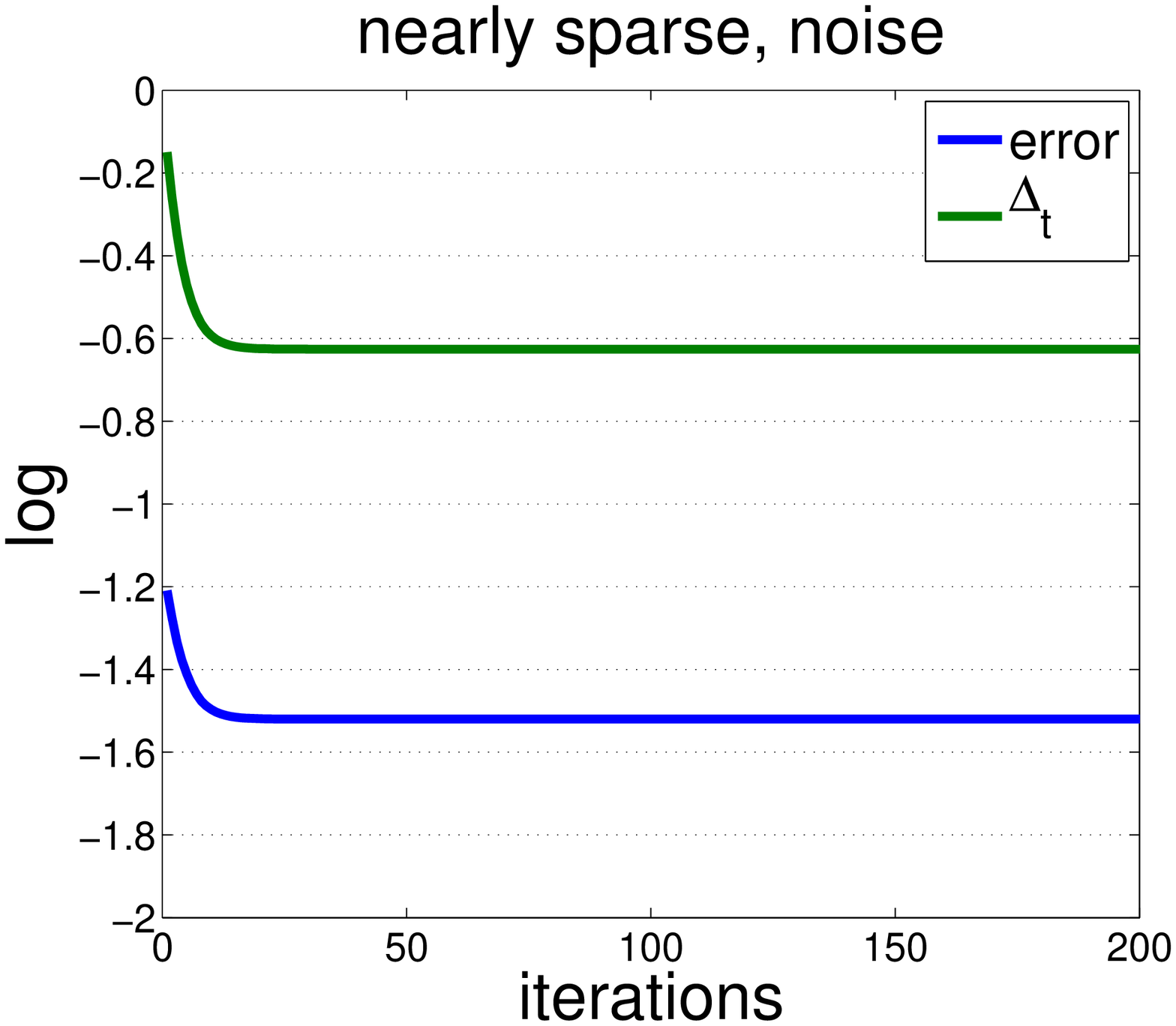}}
\subfigure[sparsity]{\includegraphics[scale=0.26]{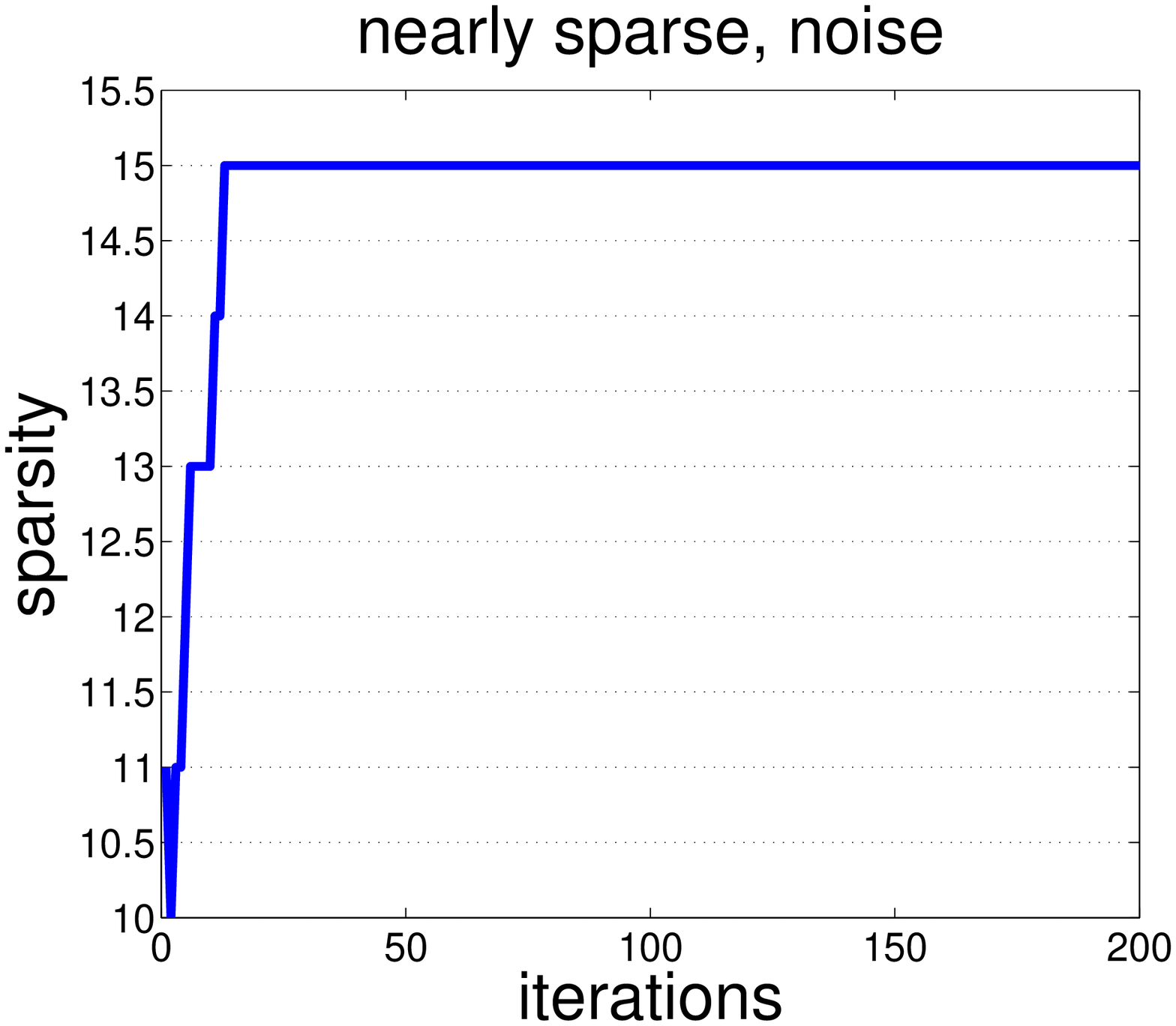}}
\caption{Recovery error and sparsity versus iterations in setting III. }\label{fig:3}
\end{figure}

\paragraph{Varying $\eta$}
We conduct more experiments to demonstrate that the robustness of the proposed HPM algorithm (Algorithm~\ref{alg:2}) with respect to the value of $\eta$. The data is generated similarly as before with $n=1000, d=10000, s=20, \sigma=0.001$. The results are shown in Figure~\ref{fig:4}, Figure~\ref{fig:5} and Figure~\ref{fig:6} for different values of $\eta$, not exceeding its upper limit $\sqrt{2}-1=0.414$. The smallest value of $\eta$ in each Figure is the smallest one~\footnote{We start a value of $\eta=0.41$ and decrease by $0.1$ until we observe divergence.} that guarantees  convergence. From these results, we have several interesting observations: (i) from noise to noiseless observations and from sparse signal to nearly sparse signal, the algorithm becomes more robust to smaller values of $\eta$ and less robust to larger values of $\eta$. For example in setting I, the smallest value of $\eta$ that guarantees convergence is $0.32$, but when adding some noise to the observations, the smallest value of $\eta$ reduces to $0.3$. However, the value $\eta=0.41$ which originally works for noiseless observations will cause the algorithm not to converge in setting II. (ii)  As long as convergence is observed, a smaller value of $\eta$ yields  faster convergence in all cases and  more accurate recovery in settings II and III. (ii) Even though the sparsity of intermediate solutions exceeds $2s$, the algorithm still converges.

\begin{figure}[t]
\centering 
\subfigure[error]{\includegraphics[scale=0.26]{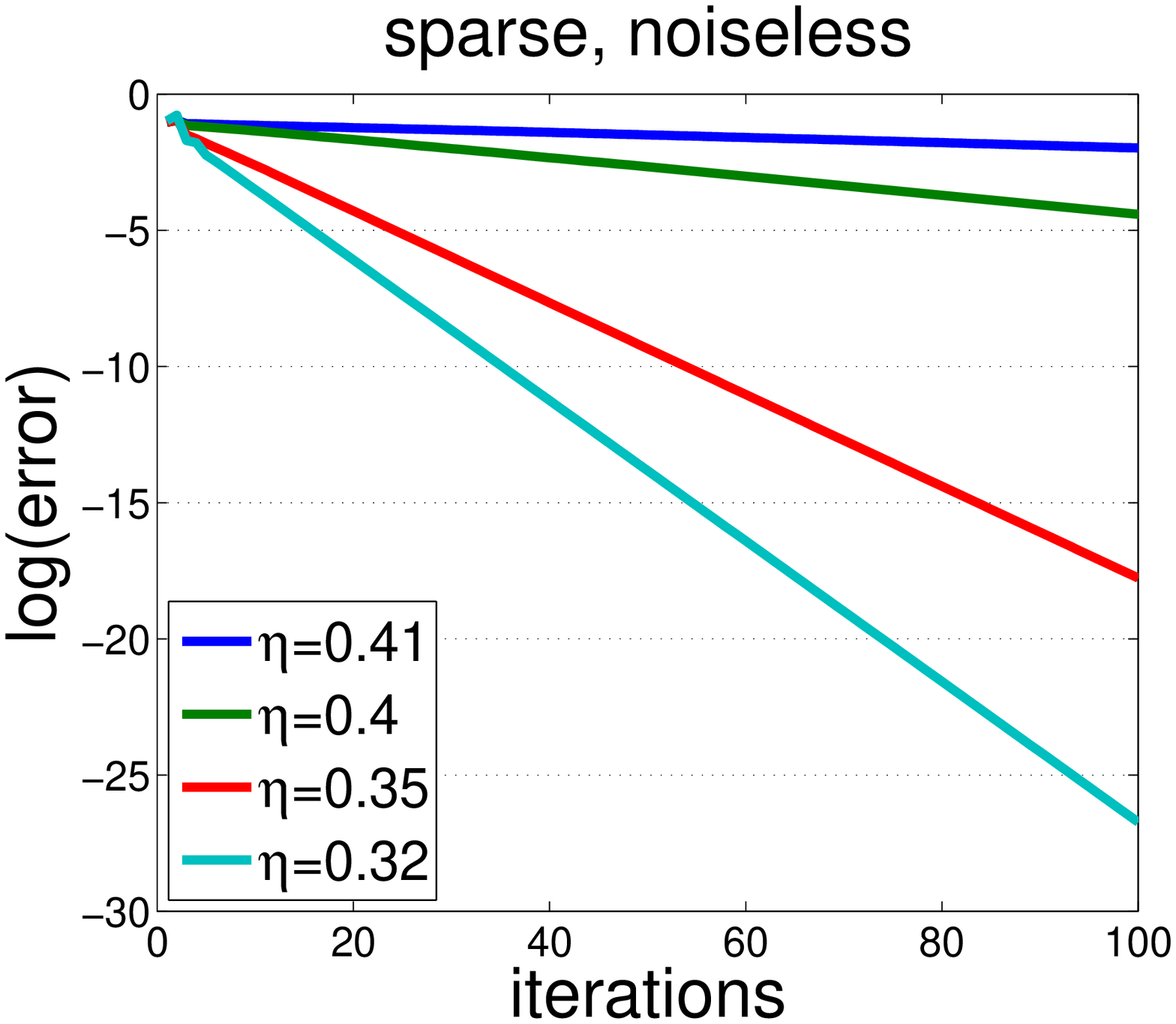}}
\subfigure[sparsity]{\includegraphics[scale=0.26]{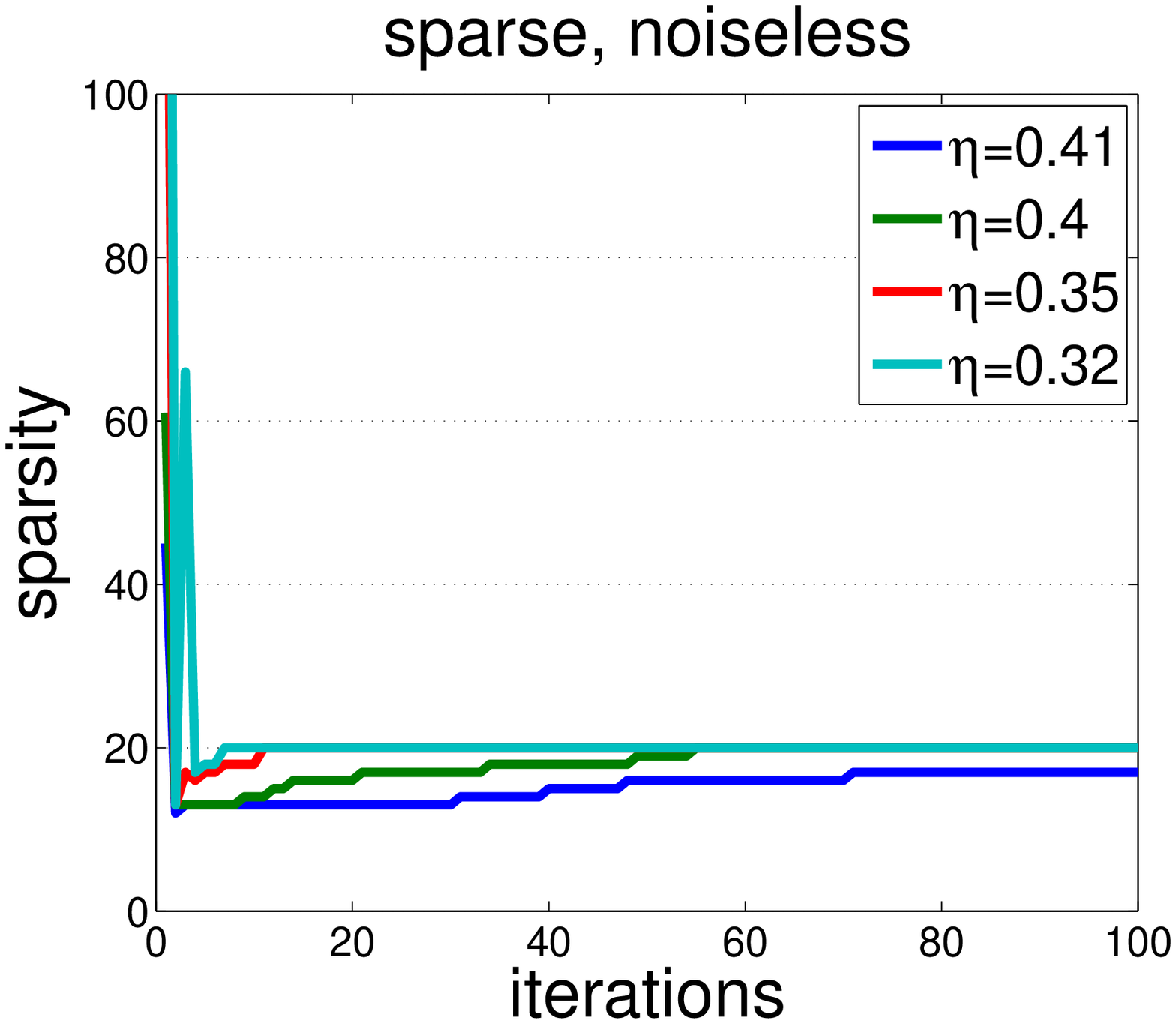}}
\caption{Recovery error and sparsity  versus iterations in setting I for different values of $\eta$.}\label{fig:4}
\centering
\subfigure[error]{\includegraphics[scale=0.26]{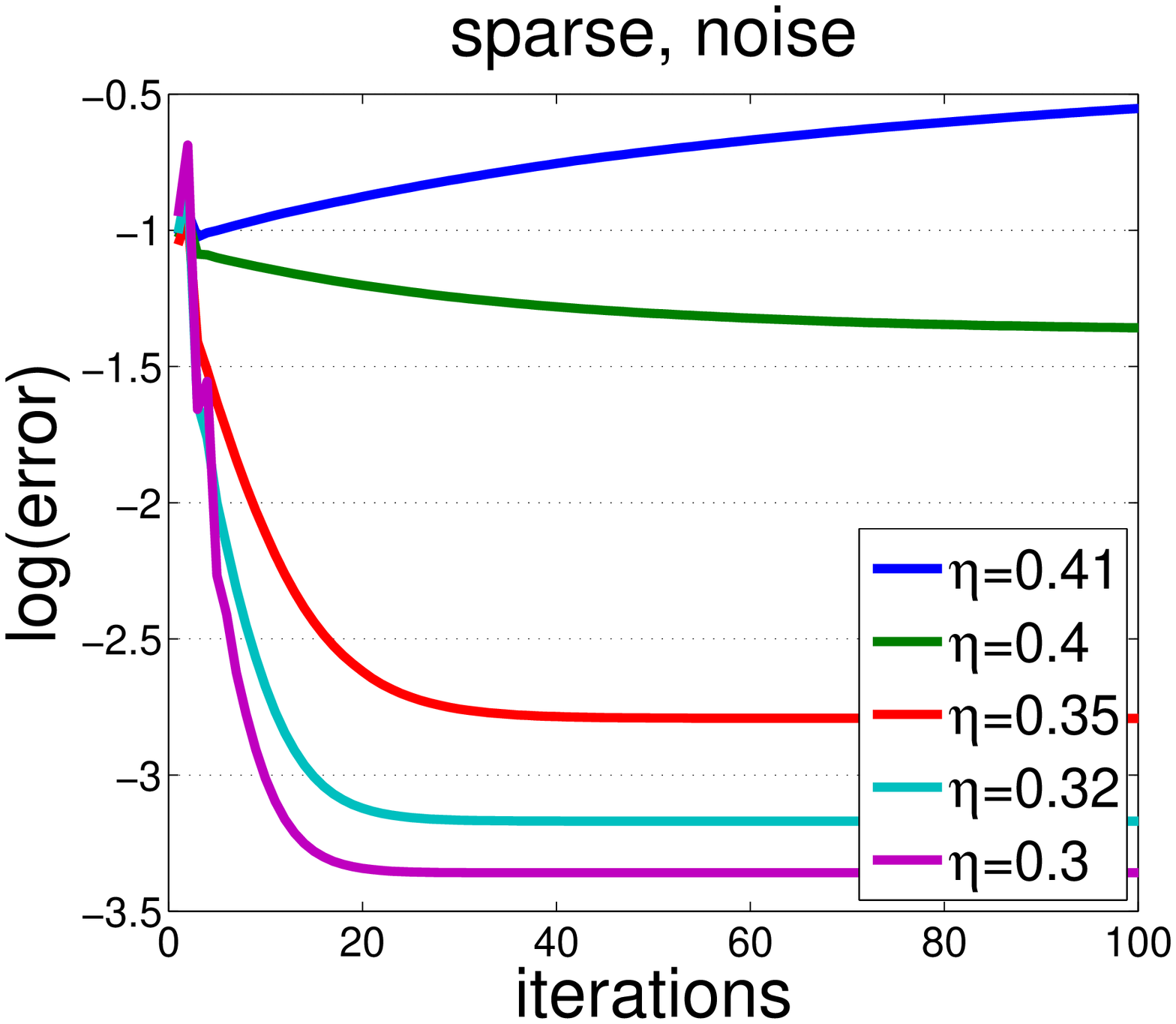}}
\subfigure[sparsity]{\label{fig:3-1}\includegraphics[scale=0.26]{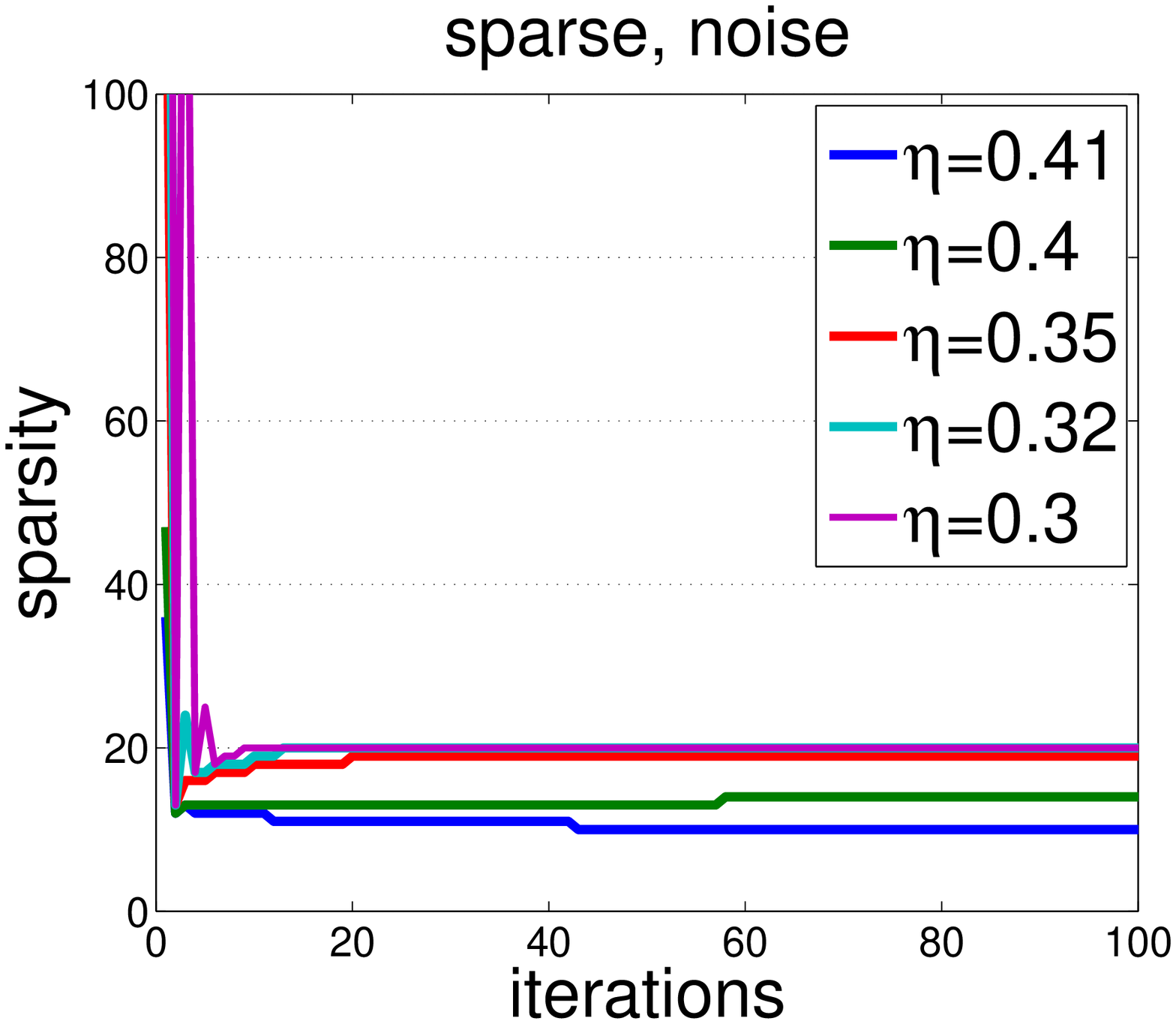}}
\caption{Recovery error and sparsity  versus iterations in setting II for different values of $\eta$.}\label{fig:5}
\end{figure}
\begin{figure}[t]

\centering
\subfigure[error]{\label{fig:2b}\includegraphics[scale=0.26]{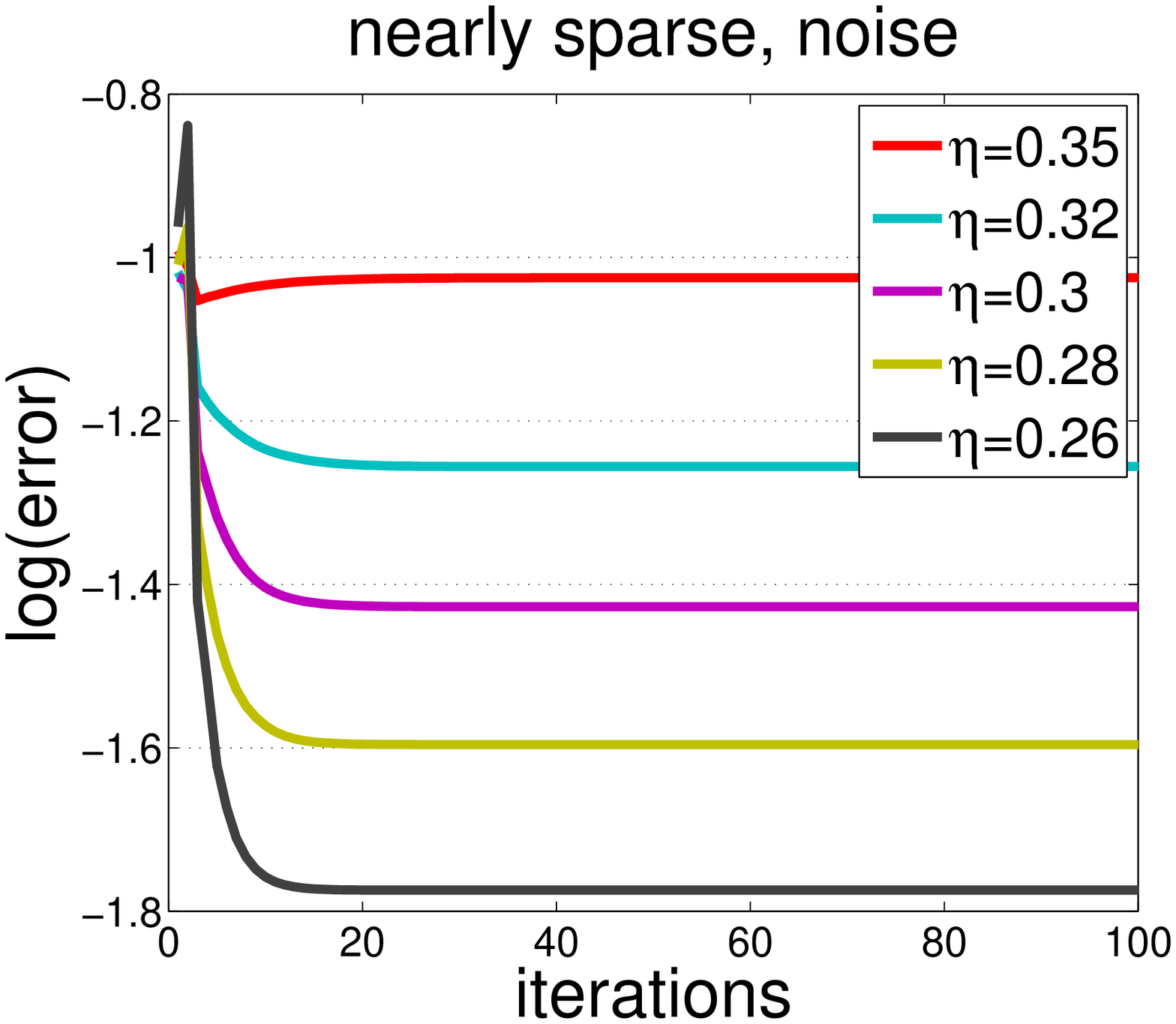}}
\subfigure[sparsity]{\label{fig:3-1}\hspace*{-0in}\includegraphics[scale=0.26]{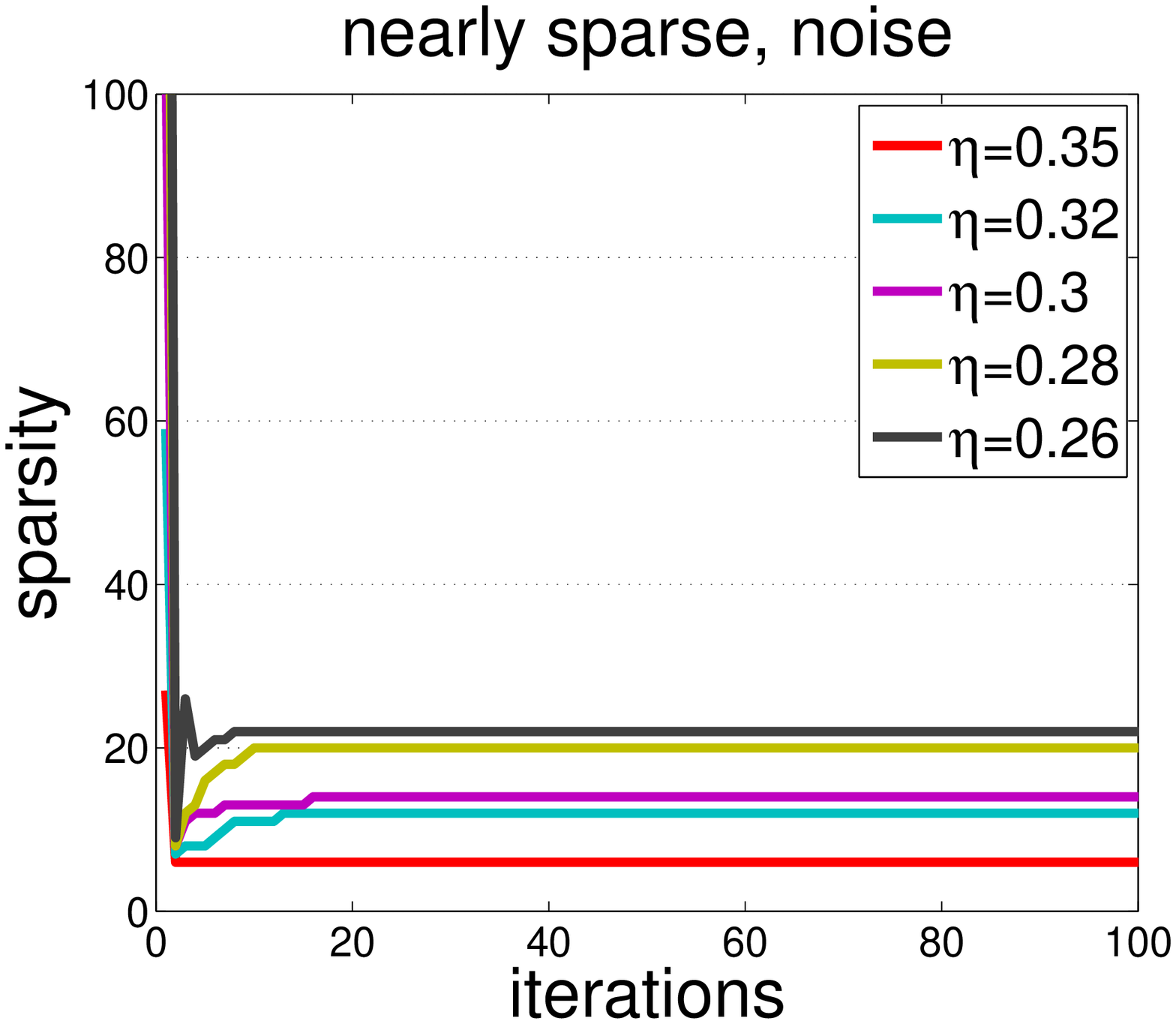}}
\caption{Recovery error and sparsity  versus iterations in setting III for different values of $\eta$.} \label{fig:6}
\end{figure}

\begin{figure}[t]
\centering
\subfigure[setting I: different $n$]{\label{fig:2c}\includegraphics[scale=0.26]{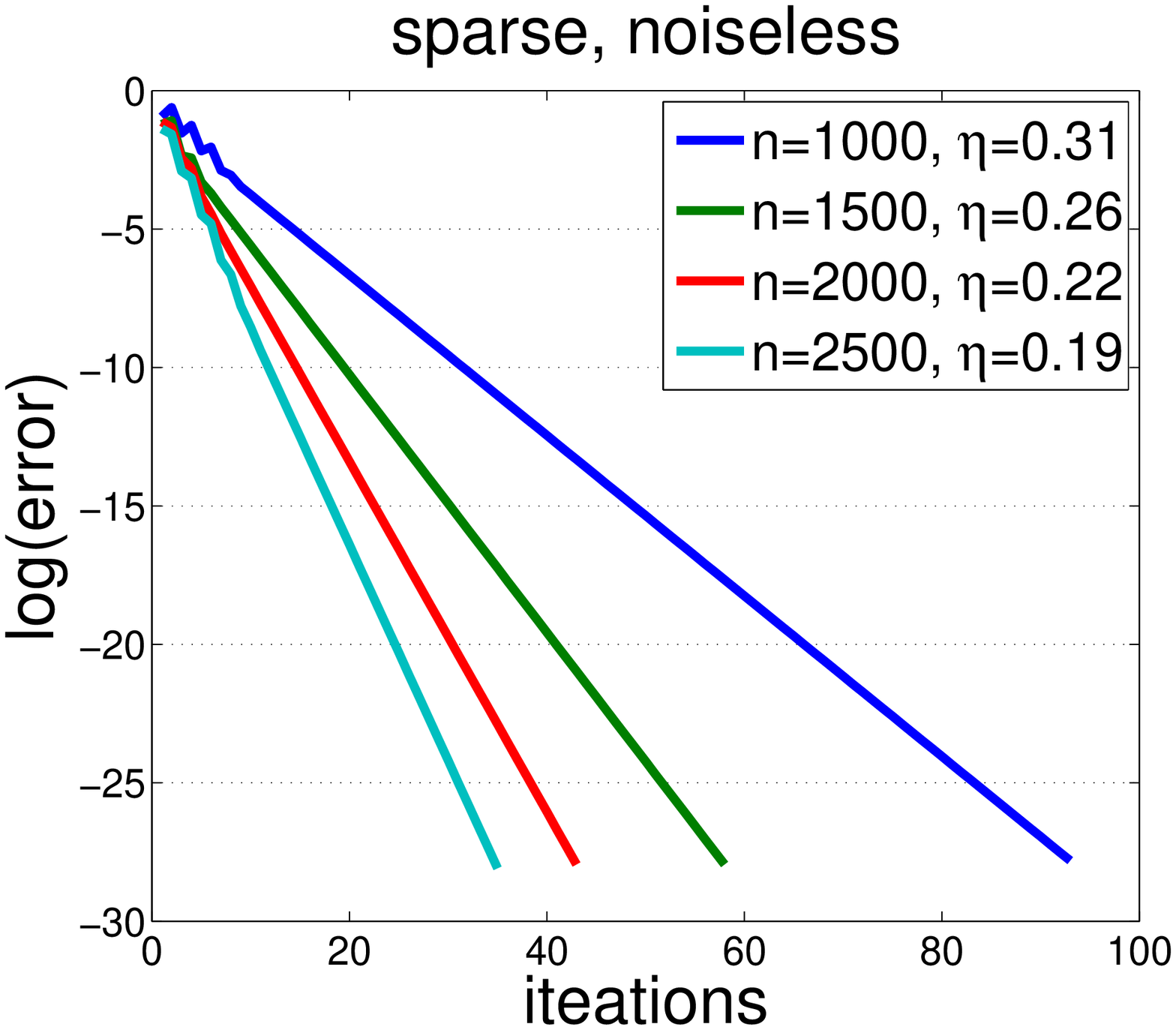}}
\subfigure[setting II: different $n$]{\label{fig:2c}\includegraphics[scale=0.26]{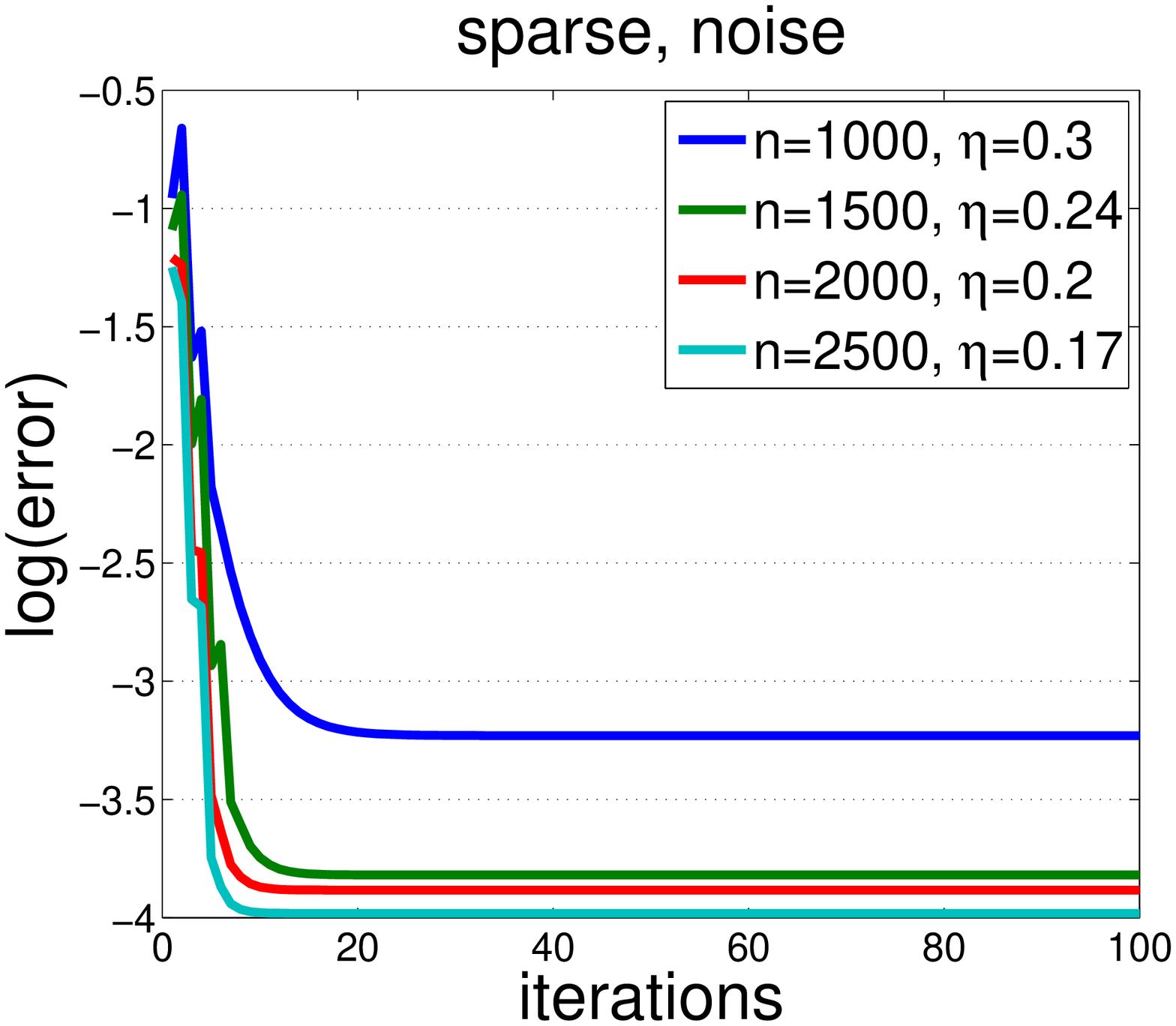}}
\caption{Recovery error vs different $n$. The value of $\eta$ is chosen as the best one for each value of $n$.}\label{fig:7}
\end{figure}

\begin{figure}
\centering
\subfigure[error]{\label{fig:9}\includegraphics[scale=0.26]{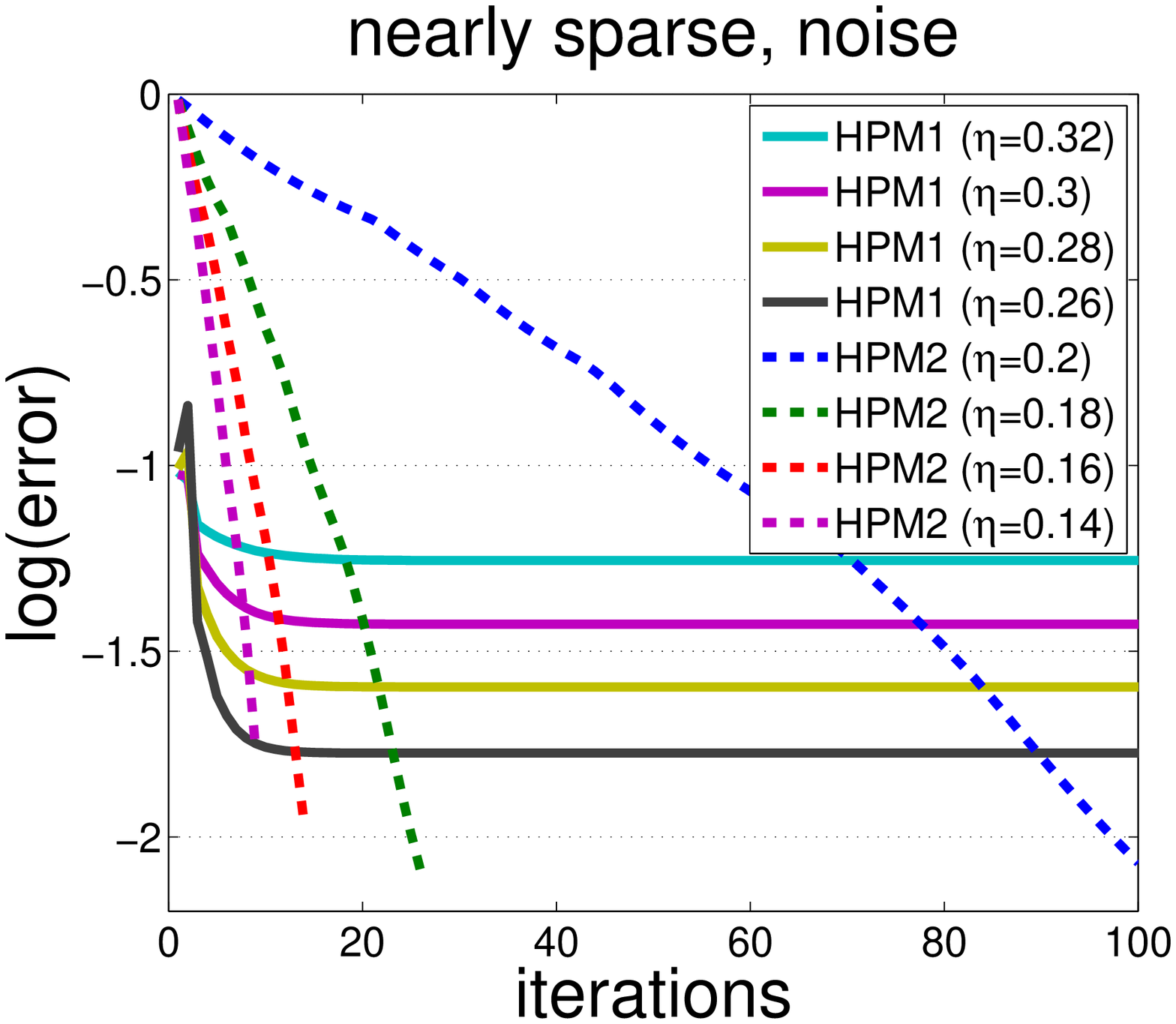}}
\subfigure[error]{\label{fig:8b}\includegraphics[scale=0.26]{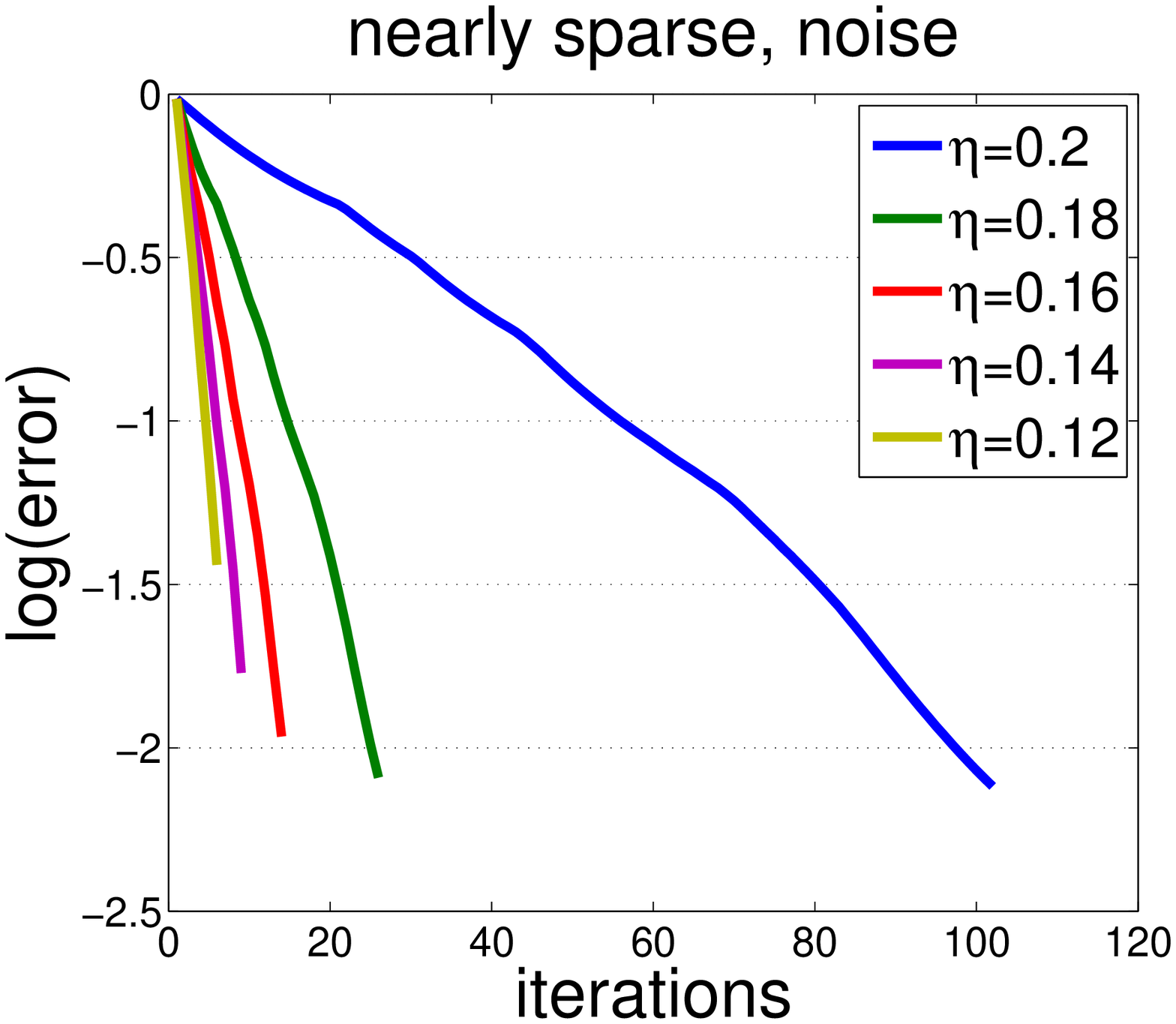}}\hspace*{-0.1in}
\subfigure[sparsity]{\label{fig:8c}\includegraphics[scale=0.26]{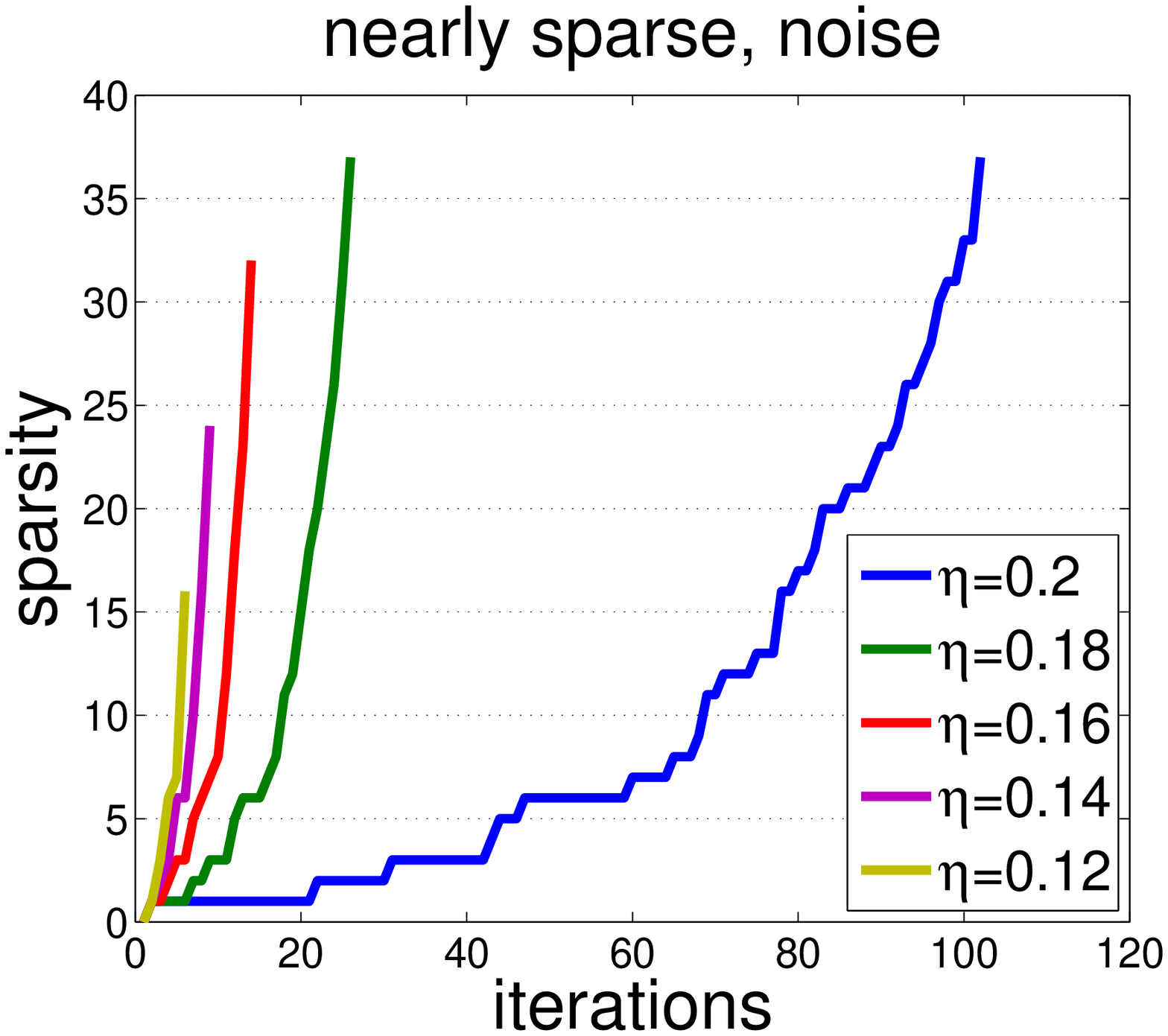}}\hspace*{-0.1in}
\caption{(a) HPM1 vs HPM2 with different values of $\eta$ in setting III. (b,c) recovery error and sparsity  of solutions in HPM2 versus iterations in setting III for different values of $\eta$.}\label{fig:8}
\end{figure}

\paragraph{Varying $n$}
We also verify that more observations can lead to faster convergence and more accurate recovery. To this end, we generate data similarly as before with different values of $n=1000, 1500, 2000, 2500$. For each value of $n$, we choose the smallest $\eta$ that can guarantee the convergence. The results for the first two settings are shown in Figure~\ref{fig:7}, which clearly demonstrate that the with more observations, we can use smaller $\eta$ to get a faster convergence and a more accurate recovery in setting II. Similar result has been observed for setting III. 

\paragraph{HPM1 vs. HPM2}
We also compare HPM2 with HPM1 in setting III to demonstrate the benefit of HPM2. The data is generated similarly as before with $n=1000, d=10000, s=20, \sigma=0.001$.  The result is shown in Figure~\ref{fig:9}. The initial value of $\lambda$ in HPM2 is set to $\|U^{\top}\y\|_\infty$.   It shows that HPM2 with an appropriate value of $\eta$ can achieve similar convergence speed and even more accurate recovery than HPM1. We also plot the recovery error and sparsity of intermediate solutions for HPM2 in Figure~\ref{fig:8b} and~\ref{fig:8c}. The curves  exhibit a tradeoff in setting the value of $\eta$, namely a smaller  value of $\eta$ leads to a faster convergence but a worse recovery, which is consistent with  Theorem~\ref{thm:11}. 


\paragraph{Comparing with Proximal-Gradient Homotopy Method (PGH)}
We compare HPM2 with the PGH method that  solves the BPDN problem for sparse signal recovery~\citep{DBLP:journals/siamjo/Xiao013}. The data is generated exactly the same as in~\citep{DBLP:journals/siamjo/Xiao013}. In particular, we generate a random measurement matrix $U\in\R^{n\times d}$ with $n = 1000$ and $d = 5000$. The entries of the matrix $U$ are generated independently with the uniform distribution over the interval $[-1,+1]$ and are scaled to have a variance $1/n$. The vector $\x_* \in \R^d$ is generated with the same distribution at $100$ randomly chosen coordinates (i.e., $\S_* = 100$). The noise $\e\in\R^n$ is a dense vector with independent random entries with the uniform distribution over the interval $[-\sigma, \sigma]$, where $\sigma$ is the noise magnitude and is set to $0.01$. Finally the vector $\y$ was obtained as $\y = U\x_* + \e$. The target value of $\lambda$  in PGH is chosen to be $\lambda_{\text{target}} = 1$ according to~\citep{DBLP:journals/siamjo/Xiao013}. The parameters  in PGH (e.g., $\gamma_{inc}, \gamma_{dec}, \eta, \delta$) are exactly the same as those used in~\citep{DBLP:journals/siamjo/Xiao013}. The initial value of $\lambda$ for both PGH and HPM2 is set to $\|U^{\top}\y\|_\infty$. We plot the recovery error  and sparsity of generated solutions versus the number of proximal updates in Figure~\ref{fig:10}. We can see that HPM2  achieves  faster convergence and better recovery than PGH for sparse signal recovery. 

\begin{figure}[t]
\centering
\includegraphics[scale=0.26]{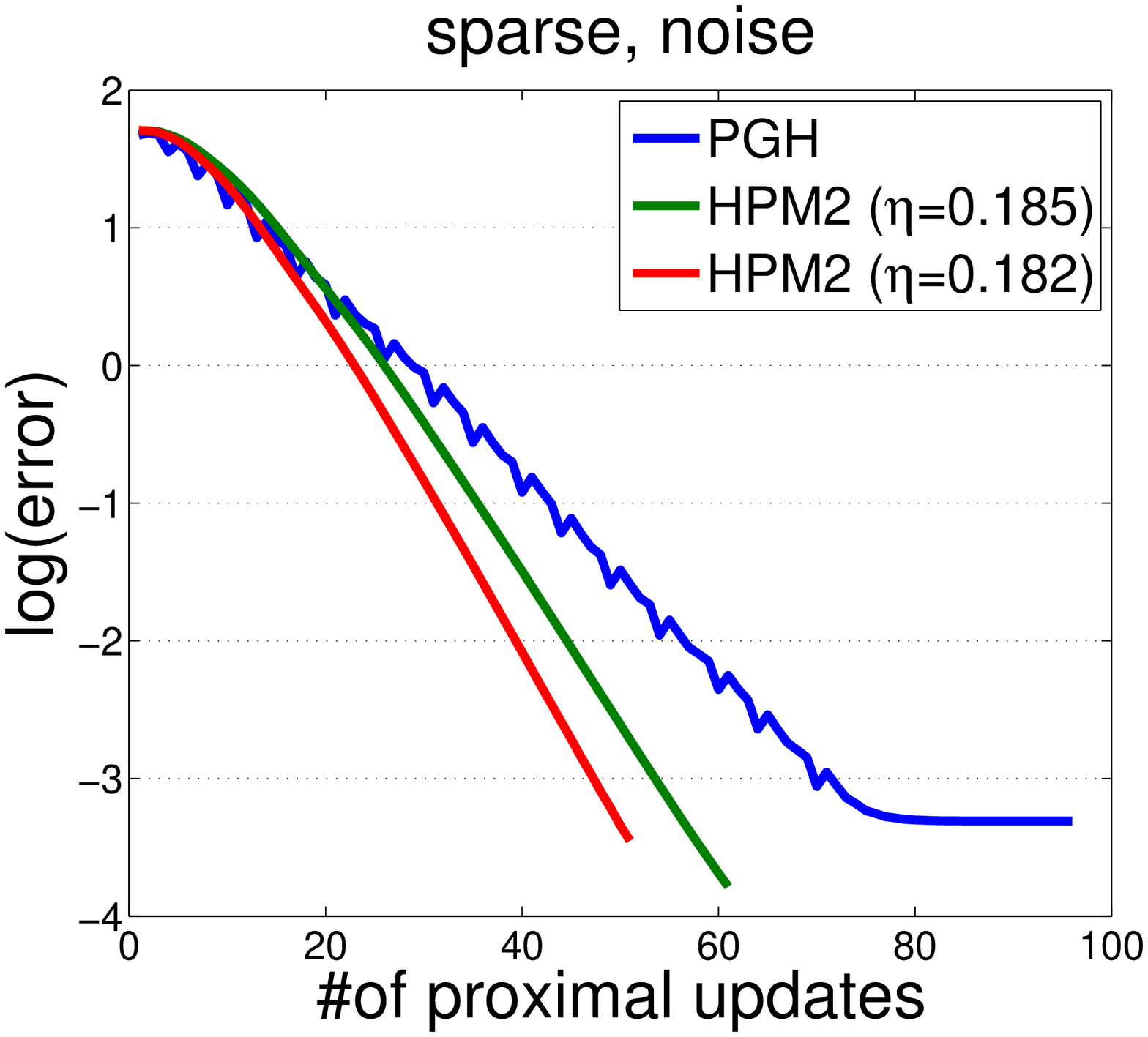}
\includegraphics[scale=0.26]{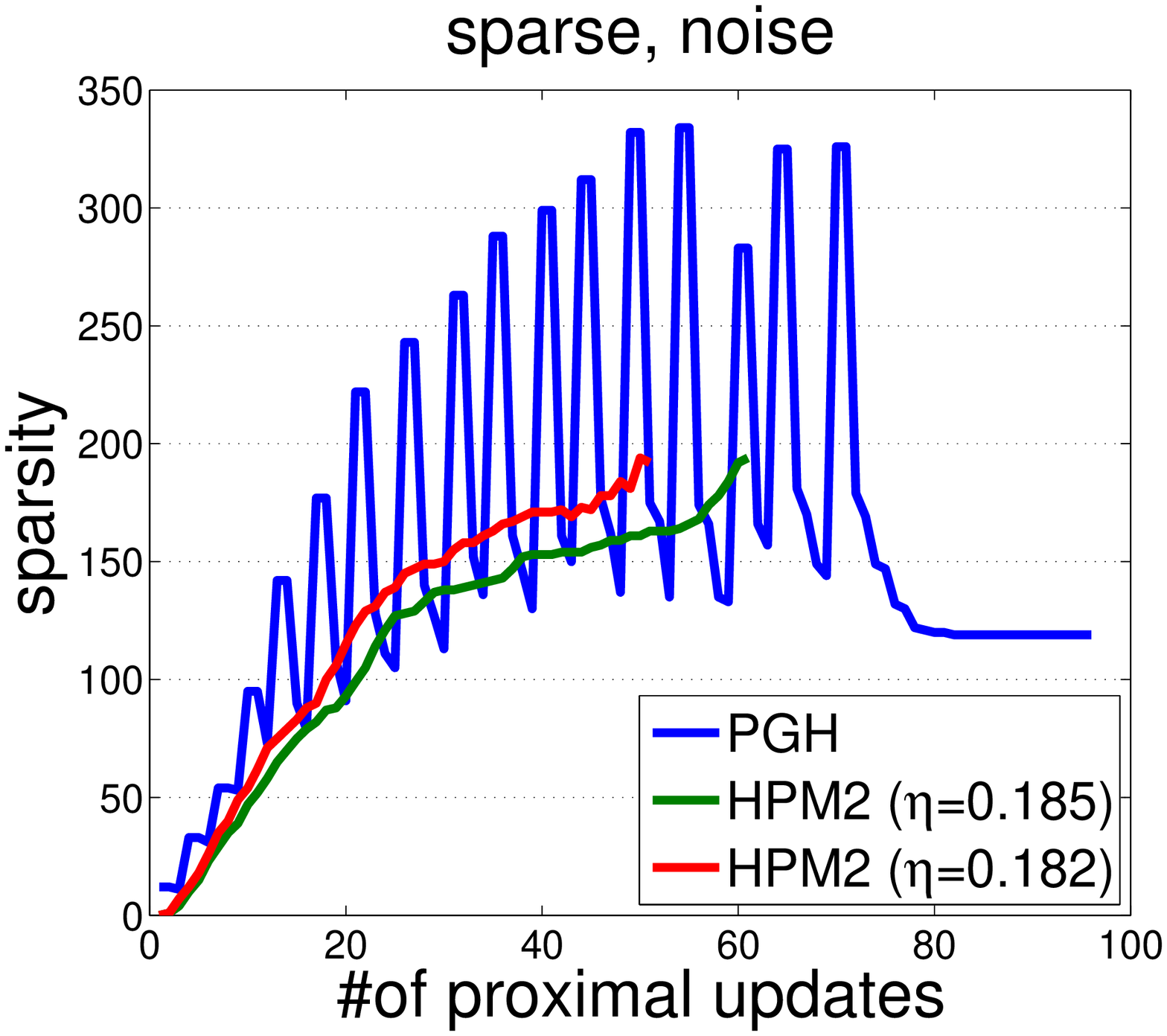}
\caption{HPM2 vs. PGH. Two different  values of $\eta$ are used in HPM2. PGH requires $96$ proximal updates and HPM2 requires $51$ and $61$ proximal updates with $\eta=0.182$ and $\eta=0.185$, respectively. The error of the final solution returned by PGH is $0.0365$, and  the error of the final solution returned by HPM2 with the two different values of $\eta$ is $0.0317$,  and $0.0227$, respectively. The recovery error of the $100$-sparse solution formed by taking the top $100$ elements in the returned solution by HPM2 is $0.0312$ and $0.0223$.}\label{fig:10}
\end{figure}

\paragraph{Comparing with Iterative Soft Thresholding algorithm (ISTA) and Iterative Hard Thresholding (IHT).}
Finally, we compare HPM2 with two other algorithms, namely ISTA and IHT~\citep{Garg:2009:GDS:1553374.1553417}. The measurement matrix $U$ and the noise vector $\e$ are generated the same as above, i.e., $U\in\R^{1000\times 5000}$ and each entry is sampled from a uniform distribution over $[-1, +1]$ and is scaled to have a variance of $1/n$. For the ground-truth signal $\x_*$, we consider two scenarios: (i) a sparse signal with $100$ randomly chosen coordinates sampled from the uniform distribution over $[-1, +1]$; (ii) a nearly sparse signal such that the entries follow an exponential decay, i.e., $[\x_*]_i = e^{-i}$. Since the proposed HPM2 and IHT require a parameter $s$ that estimates the sparsity of the target signal,  in the first scenario we vary $s$ among three values $s=100, s=200$ and $s=400$. In the second scenario, we fix $s=100$. For other parameters that each algorithm relies on (e.g., $\eta$ in HPM2, the step size parameter $1/\gamma$ in IHT and the regularization parameter $\lambda$ in ISTA), we tune them among numerous values and report the performance of the best one. We vary the value of $\eta$ in $[0.1,0.2]$, the value of $\gamma$ in $[1, 10]$ and the value of $\lambda$ in $[0.001, 1]$. The recovery error  measured by the difference between the top $s$ components of the returned solution and the top $s$ components of the ground-truth  signal is plotted in Figure~\ref{fig:11}.  From the results, we observe that (i) IHT and HPM2 converge much faster than ISTA; (ii) when the ground-truth signal is sparse and the parameter $s$ is set right to number of non-zeros in the ground-truth signal, IHT performs better than HPM2; (iii) however, when the parameter $s$ is overestimated and the  ground-truth signal is not exactly sparse, the proposed algorithm HPM2 performs better than IHT, where the later case is consistent with our comparison in Section II.

\begin{figure}[t]
\centering
\subfigure[$s=100$]{\includegraphics[scale=0.26]{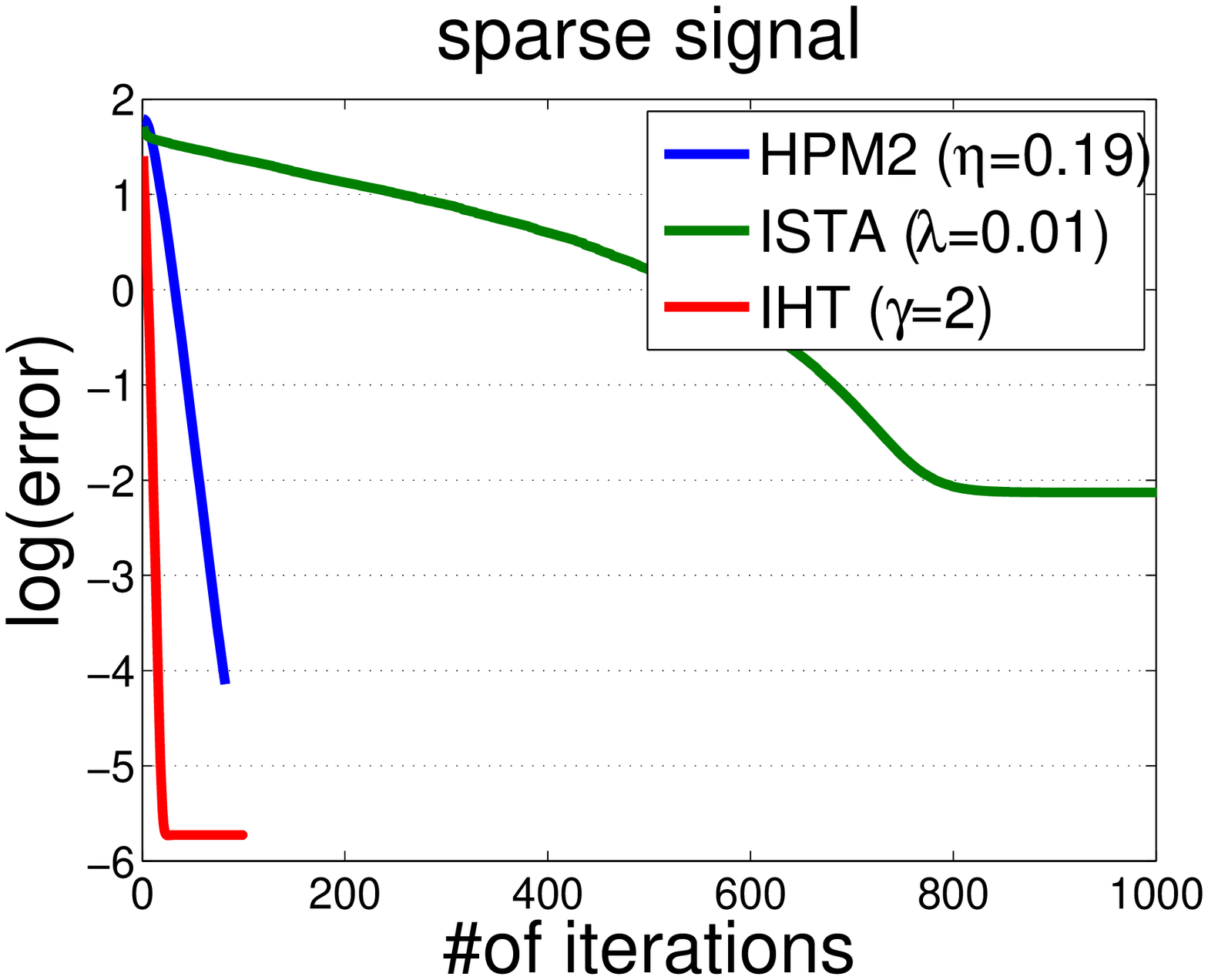}}\hspace*{-0.1in}
\subfigure[$s=200$]{\includegraphics[scale=0.26]{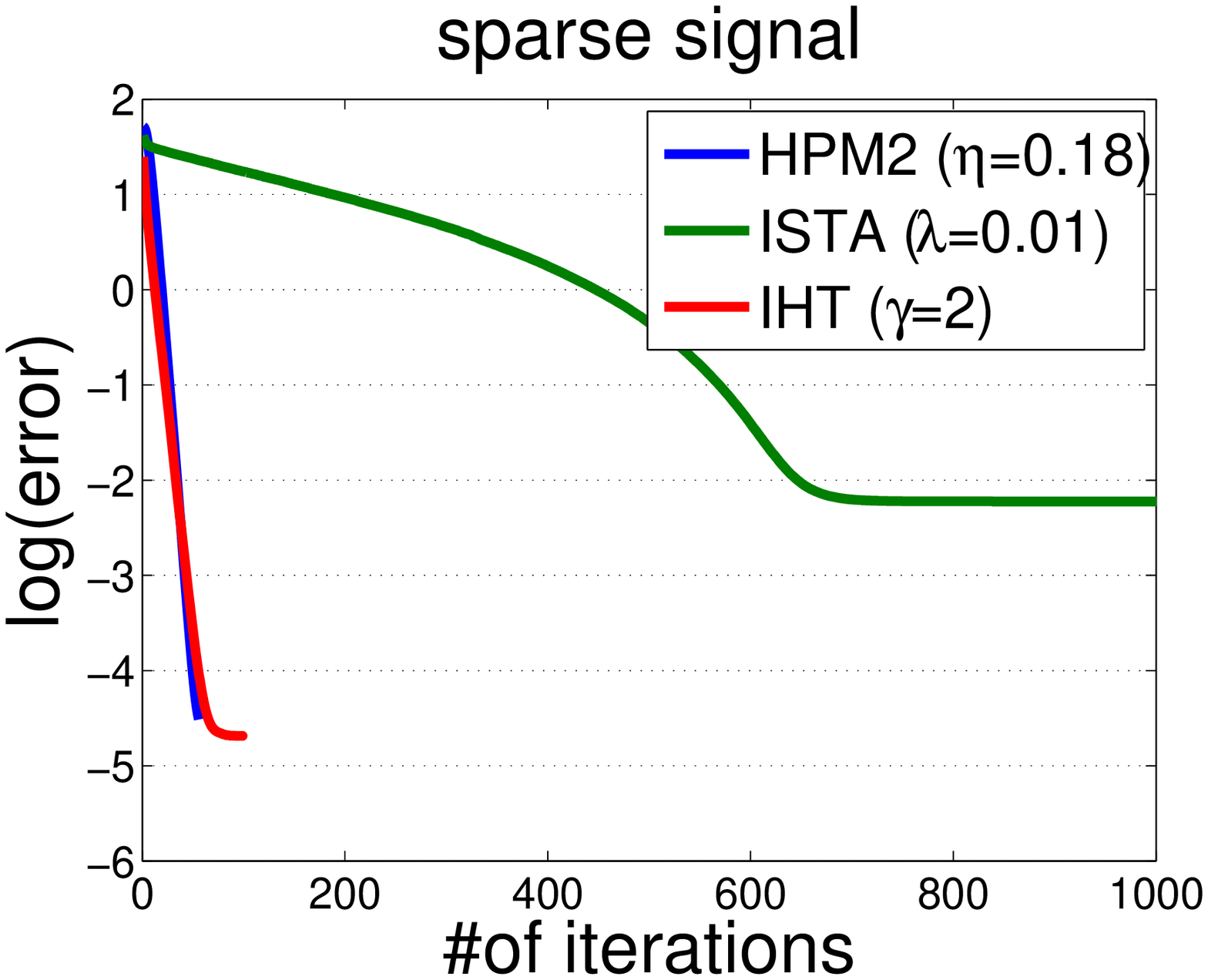}}

\subfigure[$s=400$]{\includegraphics[scale=0.26]{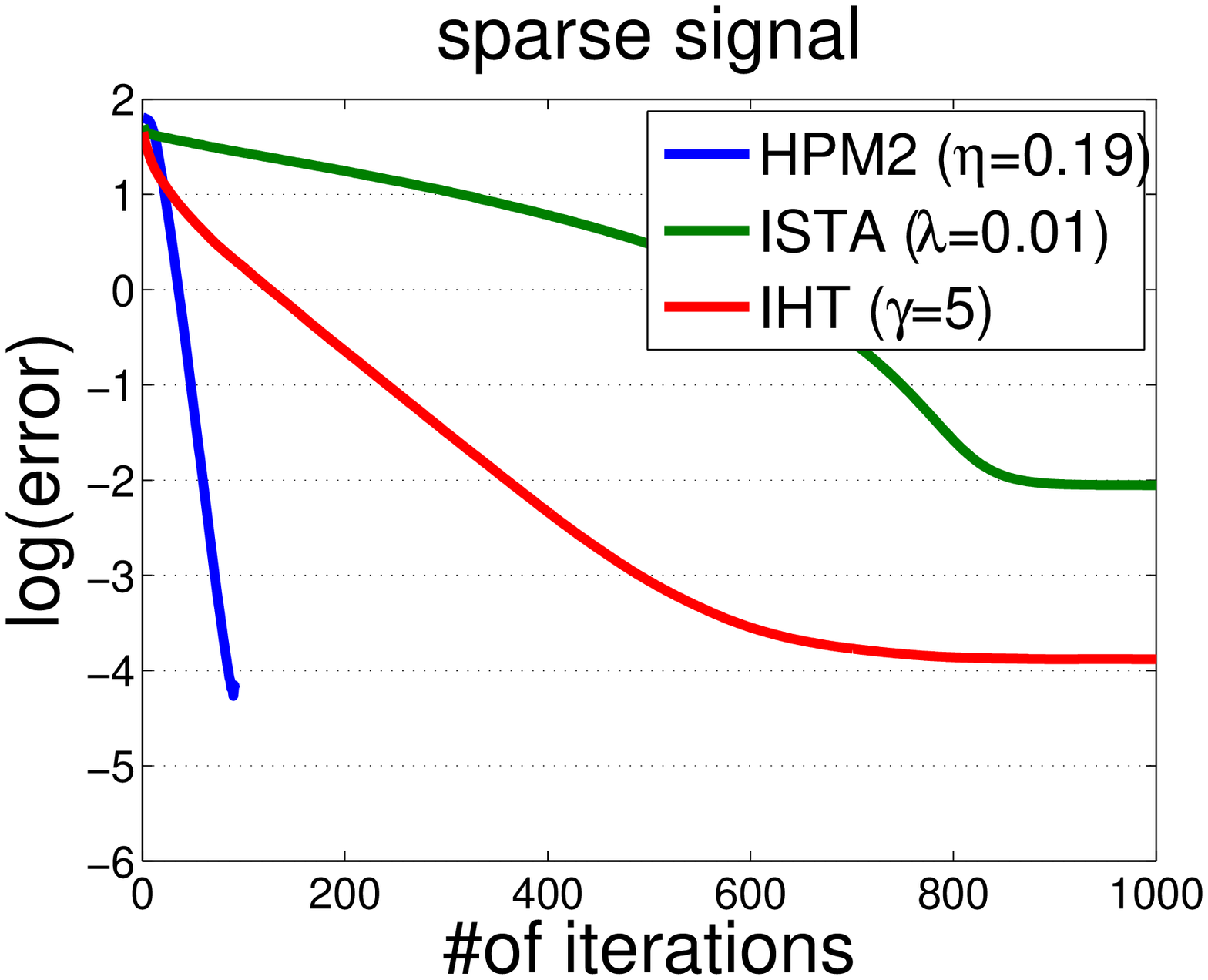}}\hspace*{-0.1in}
\subfigure[$s=100$]{\includegraphics[scale=0.26]{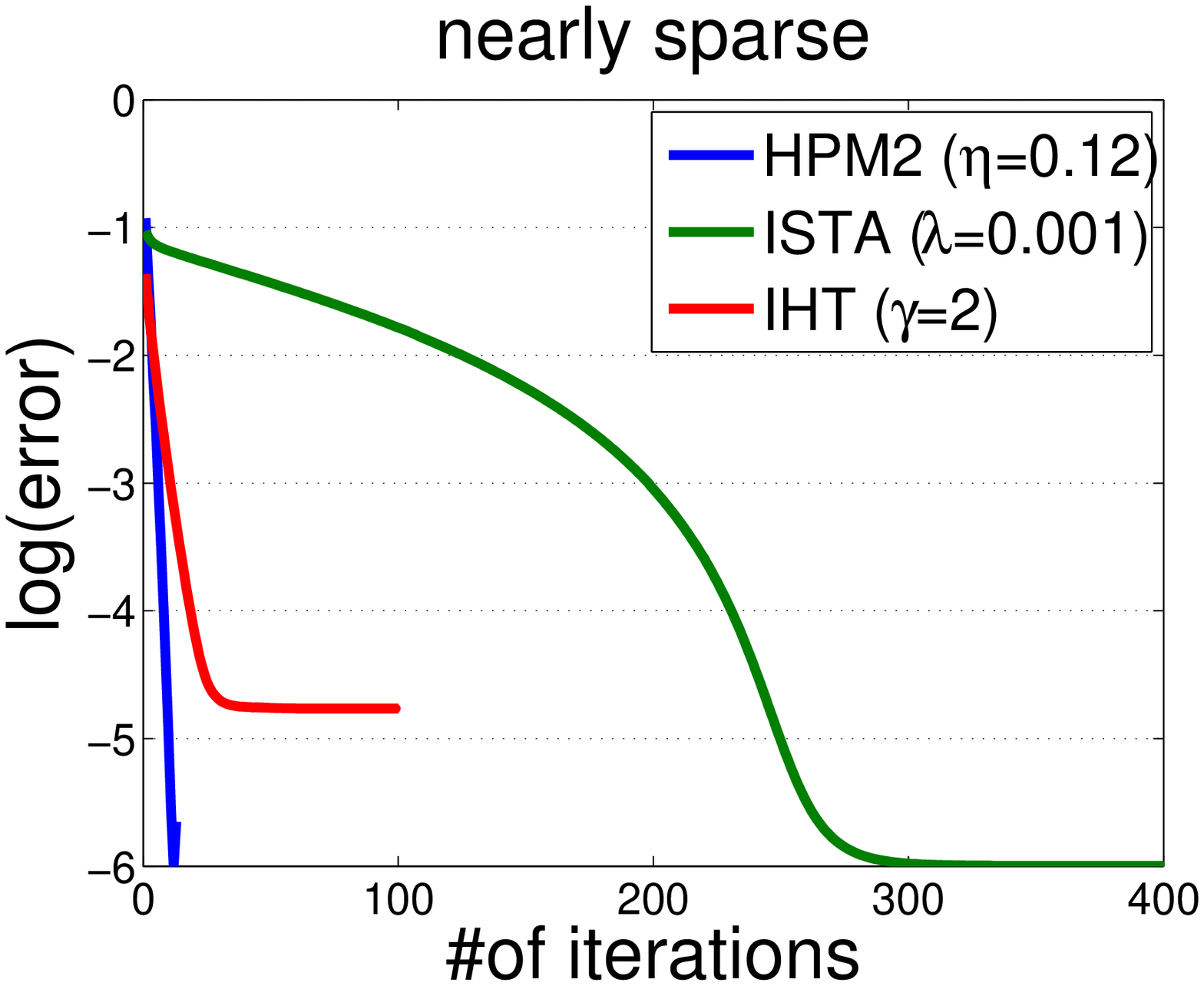}\hspace*{-0.1in}}
\caption{HPM2 vs. IST and IHT. From left to right: (a) $\x_*$ is sparse with only $100$ non-zero entries and the parameter $s$ in HPM2 and IHT is set to $100$;  (b) $\x_*$ is sparse with only $100$ non-zero entries and the parameter $s$ in HPM2 and IHT is set to $200$; (c) $\x_*$ is sparse with only $100$ non-zero entries and the parameter $s$ in HPM2 and IHT is set to $400$;    (d) $\x_*$ is nearly sparse and the parameter $s$ in HPM2 and IHT is set to $100$. 
}\label{fig:11}
\end{figure}

\section{Conclusions}
In this paper, we have presented a simple homotopy proximal mapping algorithm for compressive sensing. We proved a global linear convergence for the proposed homotopy proximal mapping algorithm for solving compressive sensing under three different settings. For sparse signal recovery, our results almost recover the best condition on the RIP constants for compressive sensing. For nearly sparse signal recovery, our result is better than previous results for instance-level recovery. In addition, we develop a practical algorithm that can run without any knowledge of noise level. Numerical simulations verify the proposed algorithms and the established theorems. 

\appendix
\section{Proof of Lemma~\ref{lem:fund}}
Define $L_t(\x)$ as 
\begin{align*}
L_t(\x) =\frac{1}{2}\left\|\x - \left(\x_t- U^{\top}(U\x_t - \y)\right)\right\|_2^2 + \lambda_t\|\x\|_1 
\end{align*}
Since $\x_{t+1}$ is the optimal solution to $\min_{\x}L_t(\x)$, therefore, we have for any $\x$
\begin{align*}
(\x_{t+1} - \x)^{\top}\partial L_t(\x_{t+1})\leq 0
\end{align*}
i.e., there exists a $g_{t+1}\in\partial\|\x_{t+1}\|_1$
\begin{align*}
&(\x_{t+1} - \x)^{\top}(\x_{t+1} - \x_t) + (\x_{t+1} - \x)^{\top}U^{\top}\left(U\x_t - \y\right) + \lambda_t(\x_{t+1} - \x)^{\top}g_{t+1}\leq 0
\end{align*}
Let $\x$ be a $s$-sparse vector with support set $\S$. 
Then we have
\begin{align*}
&(\x_{t+1} - \x)^{\top}(\x_{t+1} - \x_t) + (\x_{t+1} - \x)^{\top}U^{\top}\left(U\x_t - \y\right) + \lambda_t\|[\x_{t+1} ]_{\S_{t+1}\setminus\S}\|_1 \leq \lambda_t \|[\x_{t+1} - \x]_{\S}\|_1
\end{align*}
where we use $([\x_{t+1}]_{\S_{t+1}\setminus\S})^{\top}g_{t+1} = \|\x_{t+1}\|_1$ and $\|g_{t+1}\|_\infty\leq 1$. Note that 
\begin{align*}
& (\x_{t+1} - \x)^{\top}(\x_{t+1} - \x_t) + (\x_{t+1} - \x)^{\top}U^{\top}\left(U\x_t - \y\right) \\
& =  \|\x_{t+1} - \x\|_2^2+ (\x_{t+1} - \x)^{\top}\left(U^{\top}\left(U\x_t - \y\right) - (\x_t - \x)\right)
\end{align*}
We complete the proof by noting that $ \lambda_t\|[\x_{t+1}]_{\S_{t+1}\setminus\S}\|_1 \geq 0$ and 
\[
\|[\x_{t+1} - \x]_\S\|_1\leq \sqrt{s}\|\x_{t+1} - \x\|_2.
\]

\section{Proof of Lemma~\ref{lem:fund2}}
We first decompose $U^{\top}\left(U\x_t - \y\right) - (\x_t - \x^s_*)$ into 3 components. 
\[
\begin{split}
& U^{\top}\left(U\x_t - \y\right) - (\x_t - \x_*^{s})\\
= & U^{\top}\left(U\x_t - U \x_* - \e \right) - (\x_t - \x_*^{s}) \\
= &  \underbrace{U^{\top}U(\x_*^{s} - \x_*)}_{:=\w_a} + \underbrace{(U^{\top}U - I)(\x_t - \x_*^{s})}_{:=\w_b} - \underbrace{U^{\top}\e}_{:=\w_c}.
 \end{split}
\]
Then, we have
\begin{equation}\label{eqn:split}
\begin{split}
 & \left\| \left[U^{\top}\left(U\x_t - \y\right) - (\x_t - \x_*^{s})\right]^s\right\|_2\leq  \|\w_a^s\|_2+\|\w_b^s\|_2+\|\w_c^s\|_2.
\end{split}
\end{equation}
The last term can be bounded by $\|\w^s_c\|_2\leq \sqrt{s}\|U^{\top}\e\|_\infty$. In the following analysis, we intend to bound $\|(U^{\top}U\z)^s\|_2$ for a fixed vector $\z$, $\left\|((UU^{\top} - I)\z)^s\right\|_2$ for any sparse vector $\z$. We will address these two bounds in the following two subsections.

\subsection{Bounding $\|(U^{\top}U\z)^s\|_2$ for a fixed $\z$}
First, we define
\[
\K_{d,s} = \left\{\w \in \R^d: \|\w\|_2 \leq 1, \|\w\|_0 \leq s \right\}
\]
and
\[
\Er_s(\z) = \max\limits_{\w \in \K_{d,s}} \w^{\top}U^{\top}U\z
\]
It is easy to verify that
\[
\|(U^{\top}U\z)^s\|_2 =\Er_s(\z)
\]
Hence, to bound $\|(U^{\top}U\z)^s\|_2$, we need to bound $\Er_s(\z)$.
\begin{thm} \label{thm:a1}
For a fixed $\z$, with a probability $1 - e^{-\tau}$ for any $\tau> 0$, we have
\[
\Er_s(\z) \leq c\left(\sqrt{\frac{\tau + s\log (d/s)}{n}}\|\z\|_2 + \|\z^s\|_2\right)
\]
where $c$ is some universal constant. 
\end{thm}
\begin{proof}[Proof of Theorem~\ref{thm:a1}]
Let $\K_{d,s}(\epsilon)$ be the proper $\epsilon$-net for $\K_{d,s}$ with the smallest cardinality (i.e. covering number), and let $N(\K_{d,s}, \epsilon)$ be the covering number for $\K_{d,s}$. The following Lemma bounds the covering number $N(\K_{d,s}, \epsilon)$.
\begin{lemma} (Lemma 3.3 from~\citep{DBLP:journals/corr/abs-1109-4299})
For $\epsilon \in (0,1)$ and $s \leq d$, we have
\[
    \log N(\K_{d,s}, \epsilon) \leq s \log\left( \frac{9 d}{\epsilon s}\right)
\]
\end{lemma} \label{lemma:covering-number}
Using the $\epsilon$-net $\K_{d,s}(\epsilon)$, we define a discretized version of $\Er_s(\z)$ as
\[
\Er_s(\z, \epsilon) = \max\limits_{\w \in \K_{d,s}(\epsilon)} \w^{\top}U^{\top}U\z
\]
The following lemma relates $\Er_s(\z, \epsilon)$ with $\Er_s(\z)$.
\begin{lemma} \label{lemma:discrete-bound}
For $\epsilon \in (0, 1/\sqrt{2})$, we have
\[
\Er_s(\z) \leq \frac{\Er_s(\z, \epsilon)}{1 - \sqrt{2}\epsilon}
\]
\end{lemma}
Based on the conclusion from Lemma~\ref{lemma:discrete-bound}, it is sufficient to bound $\Er_s(\z, \epsilon)$. The following  lemma follows from the JL lemma for a sub-gassian matrix.
\begin{lemma}\label{lem:a7}
For fixed $\w$ and $\z$ such that $\|\w\|_2\leq 1$, with a probability $1-e^{-\tau}$, we have
\[
   \w^{\top}U^{\top}U\z - \w^{\top}\z \leq c\sqrt{\frac{\tau}{n}}\|\z\|_2
\]
where $c$ is some universal constant. 
\end{lemma} 
{\bf Continued Proof of Theorem~\ref{thm:a1}.}
By taking the union bound for $\w\in\K_{d,s}(\epsilon)$, we have, with a probability $1 - e^{-\tau}$,
\[
    \max\limits_{\w \in \K_{d,s}(\epsilon)} \left|\w^{\top}U^{\top}U\z - \w^{\top}\z\right| \leq c\sqrt{\frac{\tau + s\log (9d/[\epsilon s])}{n}}|\z|_2
\]
if $\epsilon \in (0, 1)$, and therefore
\[
\Er_s(\z, \epsilon) \leq c\sqrt{\frac{\tau + s\log (9d/[\epsilon s])}{n}}\|\z\|_2 + \|\z^s
\|_2
\]
 We complete the proof by using Lemma~\ref{lemma:discrete-bound} with $\epsilon = \frac{1}{2\sqrt{2}}$ and assuming $d$ is sufficiently large.
\end{proof}

\subsection{Bound $\left\|((U^{\top}U - I)\z)^s\right\|_2$ for any sparse $\z$}
\begin{thm} \label{thm:a2}
With a probability $1 - e^{-\tau}$, for any $\z$ with $\|\z\|_0 \leq s$, we have
\[
\left\|((U^{\top}U - I)\z)^s\right\|_2 \leq c\sqrt{\frac{\tau + s\log[d/s]}{n}}\|\z\|_2
\]
where $c$ is some universal constant. 
\end{thm}
\begin{proof}[Proof of Theorem~\ref{thm:a2}]
We define $\Sigma_s(\z)$ as
\[
    \Sigma_s(\z) = \max\limits_{\w \in \K_{d,s}} \w^{\top}(U^{\top}U - I)\z
\]
It is easy to see $\left\|((U^{\top}U - I)\z)^s\right\|_2 = \Sigma_s(\z)$. 
Following the analysis of Theorem~\ref{thm:a1}, it is easy to verify that, with a probability $1 - e^{-\tau}$, for a fixed $\z$, we have
\[
\Sigma_s(\z) \leq c\sqrt{\frac{\tau+ s\log [d/s]}{n}}\|\z\|_2
\]
for some universal constant $c$. 
To extend this result to any $s$-sparse $\z$, we define
\[
\mu_s = \max\limits_{\z \in \K_{d,s}} \Sigma_s(\z)
\]
Evidently, for any $\z$ with $\|\z\|_0 \leq s$, we have
\[
\Sigma_s(\z) \leq \mu_s \|\z\|_2
\]
Using the same idea as Theorem~\ref{thm:a1}, we define a discrete version of $\mu_s$ as
\[
\mu_s(\epsilon) = \max\limits_{\z \in \K_{d,s}(\epsilon)} \Sigma_s(\z)
\]
and following the same argument as Lemma~\ref{lemma:discrete-bound}, we have
\[
\mu_s \leq \frac{\mu_s(\epsilon)}{1 - \sqrt{2}\epsilon}
\]
Since for any fixed $\z \in \K_{d,s}$, with a probability $1 - e^{-\tau}$, we have
\[
\Sigma_s(\z) \leq c\sqrt{\frac{\tau+ s\log[d/s]}{n}}
\]
By taking the union bound and using the relationship between $\mu_s$ and $\mu_s(\epsilon)$, with a probability $1 - e^{-\tau}$, we have
\[
\mu_s \leq c\sqrt{\frac{\tau + s\log[d/s]}{n}}
\]
We complete the proof by using $\Sigma_s(\z) \leq \mu_s\|\z\|_2$.
\end{proof}
\begin{proof}[Proof of Lemma~\ref{lem:fund2}]
Combining the above results, we can complete the proof of Lemma~\ref{lem:fund2}. In particular, we apply Theorem~\ref{thm:a1} to bound $ \|\w_a^s\|_2 = \|(U^{\top}U(\x_*^s - \x_*))^s\|_2$ in~(\ref{eqn:split}), and apply Theorem~\ref{thm:a2} to bound $\|\w_b^s\|_2 = \|(U^{\top}U -I)(\x_t - \x_*^s)\|_2$ for any $2s$-sparse $\x_t - \x_*^s$. 
\end{proof}

\section{Other Lemmas and Proofs}

\subsection{Proof of Lemma~\ref{lemma:discrete-bound}}
The analysis is the same as that for Lemma 9.2 of~\citep{koltchinskii2011oracle}, we include it for completeness. For any $\x, \x' \in \K_{d,s}$, we can always find two vectors $\y$, $\y'$ such that
\[
\x-\x'=\y-\y', \ \|\y\|_0 \leq s, \ \|\y'\|_0 \leq s, \ \y^\top \y'=0.
\]
Thus
\[
\begin{split}
& \langle \x-\x', UU^{\top}\z \rangle  = \langle \y, UU^{\top}\z \rangle + \langle -\y', UU^{\top}\z \rangle \\
= & \|\y\|_2 \left \langle \frac{\y}{\|\y\|_2}, UU^{\top}\z \right \rangle + \|\y'\|_2 \left \langle \frac{-\y'}{\|\y'\|_2}, UU^{\top}\z \right \rangle \\
\leq & (\|\y\|_2+ \|\y'\|_2) \Er_s(\z)  \leq \Er_s(\z)  \sqrt{2} \sqrt{\|\y\|_2^2+ \|\y'\|_2^2}\\
=& \Er_s(\z)  \sqrt{2} \|\y-\y'\|_2= \Er_s(\z)  \sqrt{2} \|\x-\x'\|_2.
\end{split}
\]

Then, we have
\[
\begin{split}
& \Er_s(\z) = \max\limits_{\w \in \K_{d,s}} \w^{\top}UU^{\top}\z\\
\leq & \Er_s(\z, \epsilon) + \sup_{\x \in \K_{d,s}, \x' \in \K_{d,s}(\epsilon), \|\x-\x'\|_2 \leq \epsilon} \langle \x-\x', UU^{\top}\z \rangle \\
\leq & \Er_s(\z, \epsilon) +  \sqrt{2} \epsilon \Er_s(\z)
\end{split}
\]
which implies
\[
\Er_s(\z) \leq \frac{\Er_s(\z, \epsilon)}{1-\sqrt{2} \epsilon}.
\]

\subsection{Proof of Lemma~\ref{lem:a7}}
Let us first assume $\|\z\|_2=1$, otherwise
\[
   \w^{\top}U^{\top}U\z - \w^{\top}\z \leq (\w^{\top}U^{\top}U\z' - \w^{\top}\z')\|\z\|_2
\]
where $\z' = \z/\|\z\|_2$. Following JL lemma for a sub-gaussian matrix~\citep{Neilson}, we know that with a probability $1- \exp(-c\epsilon^2 n)$, where $c$ is some constant, 
\[
(1-\epsilon)\|\z\|_2^2\leq \|U\z\|^2_2 \leq \left(1 + \epsilon\right)\|\z\|_2^2
\]
Therefore, 
\begin{align*}
 \w^{\top}U^{\top}U\z - \w^{\top}\z &= \frac{\|U(\w+\z)\|_2^2 - \|U(\w-\z)\|_2^2}{4} - \w^{\top}\z\\
 &\leq \frac{\epsilon}{2}(\|\w\|_2^2 + \|\z\|_2^2)\leq \epsilon
\end{align*}
Therefore with a probability $1 - e^{-\tau}$, we have
\[
 \w^{\top}U^{\top}U\z - \w^{\top}\z\leq c\sqrt{\frac{\tau}{n}}\|\z\|_2
\]

\section{Top-$s$ recovery error}
\begin{prop}\label{lem:2sparse}
Let $\y\in\R^{2s}$ be an arbitrary $s$-sparse vector. Then we have
\[
\|\x^s - \y\|_2\leq \sqrt{3}\|\x - \y\|_2, \quad\forall \x\in\R^{2s}
\]
\end{prop}
\begin{proof}
Let $\X$ and $\Y$ be the support set of $\x$ and $\y$, respectively. If $|\X| \leq s$, we have
\[
\|\x^{s} - \y\|_2 = \|\x-\y\|_2.
\]
Thus, in the following, we only need to consider the case $|\X| > s$. Let $\A$ be the indices of the $s$ largest elements in $\x$, and $\B=\X\setminus \A$. Then, we have
\[
\begin{split}
\|\x-\y\|_2^2= & \sum_{i \in \A \setminus \Y} x_i^2 +\sum_{i \in \A\cap\Y} (x_i-y_i)^2  +\sum_{i \in \B \cap\Y} (x_i-y_i)^2  +\sum_{i \in \B \setminus \Y} x_i^2, \\
\|\x^{s} - \y\|_2^2 = & \sum_{i \in \A \setminus \Y} x_i^2 +\sum_{i \in \A\cap\Y} (x_i-y_i)^2  + \sum_{i \in \B \cap\Y} y_i^2.
\end{split}
\]
Since
\[
|\A \setminus \Y| +|\A \cap \Y| = |\A| = s \geq |\Y|=|\A \cap \Y| + |\B \cap\Y|
\]
we have $|\A \setminus \Y| \geq |\B \cap\Y|$. As a result, we must have
\begin{equation} \label{eqn:lem:1}
\sum_{i \in \B \cap \Y} x_i^2  \leq \sum_{i \in \A \setminus \Y} x_i^2.
\end{equation}

Since
\[
\begin{split}
 \sum_{i \in \B \cap\Y} y_i^2 \leq 2 \sum_{i \in \B \cap\Y} (x_i-y_i)^2 +  2 \sum_{i \in \B \cap\Y} x_i^2 \overset{\text{(\ref{eqn:lem:1})}}{\leq} & 2 \sum_{i \in \B \cap\Y} (x_i-y_i)^2 +  2 \sum_{i \in \A \setminus \Y} x_i^2,
\end{split}
\]
we have
\[
\begin{split}
 \|\x^{s} - \y\|_2^2 \leq   3 \sum_{i \in \A \setminus \Y} x_i^2 +\sum_{i \in \A\cap\Y} (x_i-y_i)^2  +  2 \sum_{i \in \B \cap\Y} (x_i-y_i)^2 \leq  3\|\x-\y\|_2^2.
\end{split}
\]
\end{proof}

\section{Upper bound of $\|U^{\top}\e\|_\infty$}
\begin{prop}\label{lem:l2linf}
Let $U\in\R^{n\times d}$ be a  random matrix with subGaussian entries of mean $0$ and variance $1/n$. Then with a probability $1-2e^{-\tau}$, we have
\begin{eqnarray}
 \|U^{\top}\e\|_{\infty} \leq \theta\|\e\|_2 \sqrt{\frac{\tau + \log d}{n}} \label{eqn:bound-err-c}
\end{eqnarray}
where $\theta>0$ is a constant. 
\end{prop}
\begin{proof} 
Let $\u_i$ denote the $i$-th column vector of $U$. Since $[\u_i]_j, j=1,\ldots, n$ are independent $(1/\sqrt{n})$ sub-gaussian variables, therefore $\u_i^T\e$ is $(\|\e\|_2/\sqrt{n})$ sub-gaussian variable.  According to the property of a sub-gaussian vector, there exists $\theta>0$,  we have
\[
\|\u^{\top}_i\e\|_{\psi_2} \leq \theta \frac{\|\e\|_2}{\sqrt{n}}, i=1, \ldots, d
\]
where  $\|\cdot\|_{\psi_2}$ is called an Orlicz norm~\citep{rao91}. Using the following property of Orliz norm~\citep{koltchinskii2011oracle}, with a probability $1 - 2e^{-\tau}$, we have
\[
    |\u^{\top}_i\e| \leq \|\u^{\top}_i\e\|_{\psi_2}\sqrt{\tau} = \theta\|\e\|_2 \sqrt{\frac{\tau}{n}}
\]
Taking the union bound, we can complete the proof. 
\end{proof}

\bibliography{cs-opt}

\begin{thebibliography}{52}
\providecommand{\natexlab}[1]{#1}
\providecommand{\url}[1]{\texttt{#1}}
\expandafter\ifx\csname urlstyle\endcsname\relax
  \providecommand{\doi}[1]{doi: #1}\else
  \providecommand{\doi}{doi: \begingroup \urlstyle{rm}\Url}\fi

\bibitem[Agarwal et~al.(2010)Agarwal, Negahban, and
  Wainwright]{DBLP:conf/nips/AgarwalNW10}
Alekh Agarwal, Sahand Negahban, and Martin~J. Wainwright.
\newblock Fast global convergence rates of gradient methods for
  high-dimensional statistical recovery.
\newblock In \emph{Advances in Neural Information Processing Systems 23}, pages
  37--45, 2010.

\bibitem[Beck and Teboulle(2009)]{Beck:2009:FIS:1658360.1658364}
Amir Beck and Marc Teboulle.
\newblock A fast iterative shrinkage-thresholding algorithm for linear inverse
  problems.
\newblock \emph{SIAM J. Img. Sci.}, 2\penalty0 (1):\penalty0 183--202, March
  2009.
\newblock ISSN 1936-4954.

\bibitem[Becker et~al.(2011)Becker, Bobin, and
  Cand\`{e}s]{Becker:2011:NFA:2078698.2078702}
Stephen Becker, J{\'e}r\^{o}me Bobin, and Emmanuel~J. Cand\`{e}s.
\newblock Nesta: A fast and accurate first-order method for sparse recovery.
\newblock \emph{SIAM J. Img. Sci.}, 4:\penalty0 1--39, 2011.
\newblock ISSN 1936-4954.

\bibitem[Bickel et~al.(2009)Bickel, Ritov, and
  Tsybakov]{Bickel09simultaneousanalysis}
Peter~J. Bickel, Ya'acov Ritov, and Alexandre~B. Tsybakov.
\newblock Simultaneous analysis of lasso and dantzig selector.
\newblock \emph{Annals of Statistics}, 37\penalty0 (4), 2009.

\bibitem[Blumensath and Davies(2009)]{Blumensath_iterativehard}
Thomas Blumensath and Mike~E. Davies.
\newblock Iterative hard thresholding for compressed sensing.
\newblock \emph{Appl. Comp. Harm. Anal}, pages 265--274, 2009.

\bibitem[Bredies and Lorenz(2008)]{bredies-2008-linear}
Kristian Bredies and Dirk~A Lorenz.
\newblock Linear convergence of iterative soft-thresholding.
\newblock \emph{Journal of Fourier Analysis and Applications}, 14\penalty0
  (5-6):\penalty0 813--837, 2008.

\bibitem[Cai and Zhang(2014)]{DBLP:journals/tit/CaiZ14}
T.~Tony Cai and Anru Zhang.
\newblock Sparse representation of a polytope and recovery of sparse signals
  and low-rank matrices.
\newblock \emph{{IEEE} Transactions on Information Theory}, 60\penalty0
  (1):\penalty0 122--132, 2014.

\bibitem[Cand\`{e}s(2008)]{citeulike:10486729}
E.~Cand\`{e}s.
\newblock The restricted isometry property and its implications for compressed
  sensing.
\newblock \emph{C. R. Acad. des Sci Serie I}, pages 589--592, 2008.

\bibitem[Cand\`{e}s and Tao(2007)]{CT07}
E.~Cand\`{e}s and T.~Tao.
\newblock The dantzig selector: Statistical estimation when $p$ is much larger
  than $n$.
\newblock \emph{Ann. Statist.}, 35\penalty0 (6):\penalty0 2313--2351, 2007.

\bibitem[Cand\`{e}s and Tao(2005)]{Candes:2005:DLP:2263433.2271950}
E.~J. Cand\`{e}s and T.~Tao.
\newblock Decoding by linear programming.
\newblock \emph{IEEE Trans. Inf. Theor.}, 51:\penalty0 4203--4215, 2005.

\bibitem[Cand\`{e}s and Wakin(2008)]{citeulike:2891154}
E.~J. Cand\`{e}s and M.~B. Wakin.
\newblock An introduction to compressive sampling.
\newblock \emph{Signal Processing Magazine, IEEE}, 25:\penalty0 21--30, 2008.

\bibitem[Cand\`{e}s et~al.(2006)Cand\`{e}s, Romberg, and
  Tao]{citeulike:2688127}
Emmanuel~J. Cand\`{e}s, Justin~K. Romberg, and Terence Tao.
\newblock Stable signal recovery from incomplete and inaccurate measurements.
\newblock \emph{Comm. Pure Appl. Math.}, 59:\penalty0 1207--1223, 2006.
\newblock URL \url{http://dx.doi.org/10.1002/cpa.20124}.

\bibitem[Chen et~al.(1998)Chen, Donoho, and
  Saunders]{Chen:1998:ADB:305219.305222}
Scott~Shaobing Chen, David~L. Donoho, and Michael~A. Saunders.
\newblock Atomic decomposition by basis pursuit.
\newblock \emph{SIAM J. Sci. Comput.}, 20\penalty0 (1):\penalty0 33--61, 1998.

\bibitem[Chen et~al.(2001)Chen, Donoho, and
  Saunders]{Chen:2001:ADB:588736.588850}
Scott~Shaobing Chen, David~L. Donoho, and Michael~A. Saunders.
\newblock Atomic decomposition by basis pursuit.
\newblock \emph{SIAM Rev.}, 43:\penalty0 129--159, 2001.

\bibitem[Dasgupta et~al.(2010)Dasgupta, Kumar, and
  Sarl\'{o}s]{Dasgupta:2010:SJL}
Anirban Dasgupta, Ravi Kumar, and Tam\'{a}s Sarl\'{o}s.
\newblock A sparse johnson--lindenstrauss transform.
\newblock In \emph{Proceedings of the 42nd ACM symposium on Theory of
  computing}, STOC '10, pages 341--350, 2010.

\bibitem[Davis et~al.(2004)Davis, Mallat, and Avellaneda]{Davis2004}
G.~Davis, S.~Mallat, and M.~Avellaneda.
\newblock Adaptive greedy approximations.
\newblock \emph{Constructive approximation}, 13:\penalty0 57--98, 2004.

\bibitem[Donoho et~al.(2012)Donoho, Tsaig, Drori, and
  Starck]{Donoho:2012:SSU:2331864.2332326}
D.~L. Donoho, Y.~Tsaig, I.~Drori, and J.~L. Starck.
\newblock Sparse solution of underdetermined systems of linear equations by
  stagewise orthogonal matching pursuit.
\newblock \emph{IEEE Trans. Inf. Theor.}, 58:\penalty0 1094--1121, 2012.
\newblock ISSN 0018-9448.

\bibitem[Donoho(2006)]{Donoho06compressedsensing}
David~L. Donoho.
\newblock Compressed sensing.
\newblock \emph{IEEE Trans. Inform. Theory}, 52:\penalty0 1289--1306, 2006.

\bibitem[Donoho and Tsaig(2008)]{journals/tit/DonohoT08}
David~L. Donoho and Yaakov Tsaig.
\newblock Fast solution of l1-norm minimization problems when the solution may
  be sparse.
\newblock \emph{IEEE Transactions on Information Theory}, 54:\penalty0
  4789--4812, 2008.

\bibitem[Efron et~al.(2004)Efron, Hastie, Johnstone, and
  Tibshirani]{Efron04leastangle}
Bradley Efron, Trevor Hastie, Iain Johnstone, and Robert Tibshirani.
\newblock Least angle regression.
\newblock \emph{Annals of Statistics}, 32:\penalty0 407--499, 2004.

\bibitem[Eldar and Kutyniok(2012)]{eldar2012compressed}
Y.C. Eldar and G.~Kutyniok.
\newblock \emph{Compressed Sensing: Theory and Applications}.
\newblock Compressed Sensing: Theory and Applications. Cambridge University
  Press, 2012.
\newblock ISBN 9781107005587.

\bibitem[Foucart(2011)]{Foucart:2011:HTP:2340478.2340494}
Simon Foucart.
\newblock Hard thresholding pursuit: An algorithm for compressive sensing.
\newblock \emph{SIAM J. Numer. Anal.}, 49\penalty0 (6):\penalty0 2543--2563,
  2011.

\bibitem[Garg and Khandekar(2009)]{Garg:2009:GDS:1553374.1553417}
Rahul Garg and Rohit Khandekar.
\newblock Gradient descent with sparsification: An iterative algorithm for
  sparse recovery with restricted isometry property.
\newblock In \emph{Proceedings of the 26th Annual International Conference on
  Machine Learning}, pages 337--344. ACM, 2009.

\bibitem[Hale et~al.(2008)Hale, Wotao, and Zhang]{hale-2008-fixed}
E.~T. Hale, Y.~Wotao, and Y.~Zhang.
\newblock Fixed-point continuation for l1-minimization: methodology and
  convergence.
\newblock \emph{SIAM J. on Optimization}, 19\penalty0 (3):\penalty0 1107--1130,
  2008.

\bibitem[Johnson and Lindenstrauss(1984)]{citeulike:7030987}
William Johnson and Joram Lindenstrauss.
\newblock Extensions of {L}ipschitz mappings into a {H}ilbert space.
\newblock In \emph{Conference in modern analysis and probability (New Haven,
  Conn., 1982)}, volume~26, pages 189--206. 1984.

\bibitem[Kane and Nelson(2014)]{Kane:2014:SJT:2578041.2559902}
Daniel~M. Kane and Jelani Nelson.
\newblock Sparser johnson-lindenstrauss transforms.
\newblock \emph{Journal of the ACM}, 61:\penalty0 4:1--4:23, 2014.

\bibitem[Kim et~al.(2008)Kim, Koh, Lustig, Boyd, and
  Gorinevsky]{kim2008interior}
S.J. Kim, K.~Koh, M.~Lustig, S.~Boyd, and D.~Gorinevsky.
\newblock {An Interior-Point Method for Large-Scale l1-Regularized Least
  Squares}.
\newblock \emph{Selected Topics in Signal Processing, IEEE Journal of},
  1:\penalty0 606--617, 2008.

\bibitem[Koltchinskii(2011)]{koltchinskii2011oracle}
V.~Koltchinskii.
\newblock \emph{Oracle Inequalities in Empirical Risk Minimization and Sparse
  Recovery Problems: {\'E}cole D'{\'E}t{\'e} de Probabilit{\'e}s de
  Saint-Flour XXXVIII-2008}.
\newblock Ecole d'{\'e}t{\'e} de probabilit{\'e}s de Saint-Flour. Springer,
  2011.
\newblock ISBN 9783642221460.

\bibitem[Kyrillidis and Cevher(2012)]{conf/isit/KyrillidisC12}
Anastasios~T. Kyrillidis and Volkan Cevher.
\newblock Combinatorial selection and least absolute shrinkage via the clash
  algorithm.
\newblock In \emph{ISIT}, pages 2216--2220, 2012.

\bibitem[Kyrillidis and Cevher(2014)]{DBLP:journals/jmiv/KyrillidisC14}
Anastasios~T. Kyrillidis and Volkan Cevher.
\newblock Matrix recipes for hard thresholding methods.
\newblock \emph{Journal of Mathematical Imaging and Vision}, 48\penalty0
  (2):\penalty0 235--265, 2014.

\bibitem[Mallat and Zhang(1993)]{Mallat:1993:MPT:2198030.2203996}
S.G. Mallat and Zhifeng Zhang.
\newblock Matching pursuits with time-frequency dictionaries.
\newblock \emph{Trans. Sig. Proc.}, 41:\penalty0 3397--3415, 1993.
\newblock ISSN 1053-587X.

\bibitem[Meinshausen and B\"{u}hlmann(2006)]{citeulike:2823189}
Nicolai Meinshausen and Peter B\"{u}hlmann.
\newblock High-dimensional graphs and variable selection with the lasso.
\newblock \emph{The Annals of Statistics}, 34\penalty0 (3):\penalty0
  1436--1462, 2006.

\bibitem[Needell and Tropp(2010)]{Needell:2010:CIS:1859204.1859229}
Deanna Needell and Joel~A. Tropp.
\newblock Cosamp: Iterative signal recovery from incomplete and inaccurate
  samples.
\newblock \emph{Commun. ACM}, 53:\penalty0 93--100, 2010.
\newblock ISSN 0001-0782.

\bibitem[Needell and Vershynin(2009)]{Needell:2009:UUP:1530720.1530722}
Deanna Needell and Roman Vershynin.
\newblock Uniform uncertainty principle and signal recovery via regularized
  orthogonal matching pursuit.
\newblock \emph{Found. Comput. Math.}, 9:\penalty0 317--334, 2009.
\newblock ISSN 1615-3375.

\bibitem[Nelson(2013)]{Neilson}
Jelani Nelson.
\newblock Johnson-lindenstrauss notes.
\newblock Technical report, 2013.

\bibitem[Nesterov(2007)]{RePEc:cor:louvco:2007076}
Yu. Nesterov.
\newblock Gradient methods for minimizing composite objective function.
\newblock Core discussion papers, Université catholique de Louvain, Center for
  Operations Research and Econometrics (CORE), 2007.

\bibitem[Osborne et~al.(2000)Osborne, Presnell, and Turlach]{citeulike:7675119}
M.~R. Osborne, B.~Presnell, and B.~A. Turlach.
\newblock A new approach to variable selection in least squares problems.
\newblock \emph{IMA Journal of Numerical Analysis}, 20:\penalty0 389--403,
  2000.
\newblock ISSN 1464-3642.

\bibitem[Osborne et~al.(1999)Osborne, Presnell, and Turlach]{Osborne99onthe}
Michael~R. Osborne, Brett Presnell, and Berwin~A. Turlach.
\newblock On the lasso and its dual.
\newblock \emph{Journal of Computational and Graphical Statistics}, 9:\penalty0
  319--337, 1999.

\bibitem[Plan and Vershynin(2011)]{DBLP:journals/corr/abs-1109-4299}
Yaniv Plan and Roman Vershynin.
\newblock One-bit compressed sensing by linear programming.
\newblock \emph{CoRR}, abs/1109.4299, 2011.

\bibitem[Rao and Ren(1991)]{rao91}
M.M. Rao and Z.D. Ren.
\newblock \emph{Theory of Orlicz Spaces}.
\newblock Chapman and Hall Pure and Applied Mathematics. CRC Press, 1991.
\newblock ISBN 978-0824784782.

\bibitem[Tibshirani(1996)]{tibshirani96regression}
R.~Tibshirani.
\newblock Regression shrinkage and selection via the lasso.
\newblock \emph{Journal of the Royal Statistical Society (Series B)},
  58:\penalty0 267--288, 1996.

\bibitem[Tropp(2006{\natexlab{a}})]{Tropp:2006:GGA:2263414.2271370}
J.~A. Tropp.
\newblock Greed is good: Algorithmic results for sparse approximation.
\newblock \emph{IEEE Trans. Inf. Theor.}, 50:\penalty0 2231--2242,
  2006{\natexlab{a}}.
\newblock ISSN 0018-9448.

\bibitem[Tropp(2006{\natexlab{b}})]{journals/tit/Tropp06}
Joel~A. Tropp.
\newblock Just relax: convex programming methods for identifying sparse signals
  in noise.
\newblock \emph{IEEE Transactions on Information Theory}, 52:\penalty0
  1030--1051, 2006{\natexlab{b}}.

\bibitem[Tropp and Gilbert(2007)]{4385788}
Joel~A. Tropp and Anna~C. Gilbert.
\newblock Signal recovery from random measurements via orthogonal matching
  pursuit.
\newblock \emph{Information Theory, IEEE Transactions on}, 53:\penalty0
  4655--4666, 2007.

\bibitem[Tseng(2008)]{citeulike:6604666}
P.~Tseng.
\newblock {On accelerated proximal gradient methods for convex-concave
  optimization}.
\newblock \emph{submitted to SIAM Journal on Optimization}, 2008.

\bibitem[Turlach et~al.(2005)Turlach, Venables, and Wright]{citeulike:259101}
Berwin~A. Turlach, William~N. Venables, and Stephen~J. Wright.
\newblock Simultaneous variable selection.
\newblock \emph{Technometrics}, 47:\penalty0 349--363, 2005.

\bibitem[van~de Geer and B{\"u}hlmann(2009)]{GeeBue09}
Sara~A. van~de Geer and Peter B{\"u}hlmann.
\newblock On the conditions used to prove oracle results for the lasso.
\newblock \emph{Electron. J. Statist.}, 3:\penalty0 1360--1392, 2009.

\bibitem[Wainwright(2009)]{Wainwright:2009:STH:1669487.1669506}
Martin~J. Wainwright.
\newblock Sharp thresholds for high-dimensional and noisy sparsity recovery
  using l1-constrained quadratic programming (lasso).
\newblock \emph{IEEE Trans. Inf. Theor.}, 55:\penalty0 2183--2202, 2009.

\bibitem[Xiao and Zhang(2013)]{DBLP:journals/siamjo/Xiao013}
Lin Xiao and Tong Zhang.
\newblock A proximal-gradient homotopy method for the sparse least-squares
  problem.
\newblock \emph{{SIAM} Journal on Optimization}, 23\penalty0 (2):\penalty0
  1062--1091, 2013.

\bibitem[Zhang and Huang(2008)]{zhang2008}
Cun-Hui Zhang and Jian Huang.
\newblock The sparsity and bias of the lasso selection in high-dimensional
  linear regression.
\newblock \emph{The Annals of Statistics}, 36:\penalty0 1567--1594, 2008.

\bibitem[Zhang(2009)]{citeulike:5426408}
Tong Zhang.
\newblock Some sharp performance bounds for least squares regression with l1
  regularization.
\newblock \emph{Ann. Statist.}, 37:\penalty0 2109--2144, 2009.

\bibitem[Zhao and Yu(2006)]{Zhao:2006:MSC:1248547.1248637}
Peng Zhao and Bin Yu.
\newblock On model selection consistency of lasso.
\newblock \emph{J. Mach. Learn. Res.}, 7:\penalty0 2541--2563, 2006.

\end{thebibliography}




\vskip 0.2in

\end{document}